%% file: journal_version.tex
\documentclass[acmsmall]{acmart}
\AtBeginDocument{%
  }

\setcopyright{cc}
\setcctype{by}
\acmJournal{PACMMOD}
\acmYear{2026} \acmVolume{4} \acmNumber{3 (SIGMOD)} \acmArticle{165}
\acmMonth{6} \acmDOI{10.1145/3802042}

\newcommand{\M}{\mathcal{M}}
\newcommand{\R}{\mathcal{R}}
\newcommand{\bsmu}{\boldsymbol{\mu}}
\newcommand{\GS}{\operatorname{GS}}
\newcommand{\LS}{\operatorname{LS}}
\newcommand{\DS}{\operatorname{DS}}
\newcommand{\Gauss}{\operatorname{Gauss}}
\newcommand{\Lap}{\operatorname{Lap}}

\usepackage{algorithm}
\usepackage[noend]{algpseudocode}

\algnewcommand\algorithmicforeach{\textbf{for each}}
\algdef{S}[FOR]{ForEach}[1]{\algorithmicforeach\ #1\ \algorithmicdo}
\usepackage{multirow}
\usepackage{subcaption}
\usepackage{pgfplots}
\usepackage{caption}
\usepackage{subcaption}
\usepackage{calc}
\usepackage{comment}
\usepackage{listings}
\usepackage{xcolor} 
\usepackage{multicol}
\definecolor{r2tCol}{HTML}{D55E00}
\definecolor{spCol}{HTML}{0072B2}
\definecolor{dpsCol}{HTML}{009E73}
\definecolor{pmsjaCol}{HTML}{CC79A7}
\usepackage{ifthen}
\newboolean{fullver}
\setboolean{fullver}{true}

\begin{document}

\title{DP-S4S: Accurate and Scalable Select-Join-Aggregate Query Processing with User-Level Differential Privacy}

\author{Yuan Qiu}
\email{yuanqiu@seu.edu.cn}
\orcid{0000-0002-3488-6386}
\affiliation{%
  \institution{Southeast University}
  \city{Nanjing}
  \state{Jiangsu}
  \country{China}
}

\author{Xiaokui Xiao}
\email{xkxiao@nus.edu.sg}
\orcid{0000-0003-0914-4580}
\affiliation{%
  \institution{National University of Singapore}
  \city{Singapore City}
  \country{Singapore}
}

\author{Yin Yang}
\email{yyang@hbku.edu.qa}
\orcid{0000-0002-0549-3882}
\affiliation{%
  \institution{Hamad Bin Khalifa University}
  \city{Doha}
  \country{Qatar}
}


\begin{abstract}
 Answering Select-Join-Aggregate (SJA) queries with differential privacy (DP) is a fundamental problem with important applications in various domains. The current state-of-the-art methods ensure user-level DP (i.e., the adversary cannot infer the presence or absence of any given individual user with high confidence) and achieve instance-optimal accuracy on the query results. However, these solutions involve solving expensive optimization programs, which may incur prohibitive computational overhead for large databases.

 One promising direction to achieve scalability is through sampling, which provides a tunable trade-off between result utility and computational costs. However, applying sampling to differentially private SJA processing is a challenge for two reasons. First, it is unclear \textit{what to sample}, in order to achieve the best accuracy within a given computational budget. Second, prior solutions were not designed with sampling in mind, and their mathematical tool chains are not sampling-friendly. To our knowledge, the only known solution that applies sampling to private SJA processing is {\sf S\&E}, a recent proposal that (i) samples users and (ii) combines sampling directly with existing solutions to enforce DP. We show that both are suboptimal designs; consequently, even with a relatively high sample rate, the error incurred by {\sf S\&E} can be 10x higher than the underlying DP mechanism without sampling.

 Motivated by this, we propose \textit{Differentially Private Sampling for Scale} ({\sf DP-S4S}), a novel mechanism that addresses the above challenges by (i) sampling \textit{aggregation units} instead of users, and (ii) laying the mathematical foundation for SJA processing under \textit{R\'enyi-DP}, which composes more easily with sampling. Further, {\sf DP-S4S} can answer both scalar (e.g., COUNT, SUM, etc.) and vector (with GROUP BY predicts) SJA queries. Extensive experiments on real data demonstrate that DP-S4S enables scalable SJA processing on large datasets under user-level DP, while maintaining high result utility.
\end{abstract}



\received{October 2025}
\received[revised]{January 2026}
\received[accepted]{February 2026}

\maketitle

\sloppy

\input{sections/introduction}
\input{sections/preliminaries}

\input{sections/sampling-for-a-single-query}

\input{sections/sampling-for-multiple-queries}

\input{sections/experiments}
\input{sections/conclusion}

\begin{acks}
Yuan Qiu is supported by the Start-up Research Fund of Southeast University, under grant RF1028625150.
Xiaokui Xiao is supported by the Ministry of Education, Singapore, under an AcRF Tier-2 Grant (MOE-000761-01), and by Google under grant 00006186.
Yin Yang is supported by the Qatar National Research Fund, a member of the Qatar Foundation (No. PPM 05-0506-210017). Any opinions, findings, and conclusions, or recommendations expressed
in this material are those of the authors and do not reflect the views of the funding agencies.
We thank the anonymous reviewers for their constructive comments. We also thank Yangfan Jiang for discussions on privacy amplification, and Yiwei Chen for discussions on smooth sensitivity.
\end{acks}

\bibliographystyle{ACM-Reference-Format}
\bibliography{ref}

\ifthenelse{\boolean{fullver}}{
\appendix

\input{sections/useful-lemmas}
\input{sections/proofs}
\input{sections/ablation}

\input{sections/sande}
\input{sections/ablation-ss}
\input{sections/vector-error}

\input{sections/cdp}

\input{sections/data}

}{}

\end{document}

%% file: sections/introduction.tex
\section{Introduction}

{\it Differential privacy (DP)}~\cite{DBLP:journals/fttcs/DworkR14} provides a rigorous framework for releasing query results while protecting individuals' privacy, by requiring outputs to remain statistically similar on any two {\it neighboring datasets}, so that the adversary cannot infer with high confidence which dataset was used to produce the outputs. The precise privacy guarantee depends on the definition of neighboring datasets. In a relational database, {\it record-level DP} considers two datasets as neighbors if they differ in exactly one record~\cite{DBLP:conf/icalp/Dwork06,DBLP:conf/sigmod/0002HIM19,DBLP:journals/pacmmod/ZhangH23,DBLP:journals/vldb/TaoGMR25,DBLP:journals/pacmmod/GuanYZXCXXZ25}. {\it User-level DP}, in contrast, treats two datasets as neighbors if they differ in the set of \textit{all the records} contributed by one user~\cite{DBLP:conf/tcc/KasiviswanathanNRS13,DBLP:journals/pacmmod/FangY24}. The latter, stronger notion is necessary when each user may contribute multiple, potentially correlated records, e.g., connections in a social network, purchase histories in an e-commerce platform, etc. 




This paper focuses on enforcing user-level DP when answering {\it Select-Join-Aggregate (SJA)} queries, which combine selection predicates, joins, and aggregation operators to capture a broad spectrum of reporting workloads, making them a central target for private query processing systems~\cite{DBLP:journals/pvldb/KotsogiannisTHF19,DBLP:journals/pacmmod/0007S023}. A major challenge in processing SJA queries under user-level DP is that such a query often has a high \textit{sensitivity}, meaning that adding or removing a user (which may lead to the addition or removal of many records) can significantly alter the query output. As we review in Section~\ref{sec:dp-definition}, a high sensitivity necessitates large amounts of noise to be injected into the query result to satisfy DP, leading to poor result utility.



To address this issue, state-of-the-art approaches to SJA processing under user-level DP (e.g., \cite{DBLP:conf/sigmod/DongF0TM22, DBLP:journals/pacmmod/0007S023}) commonly employ a technique called {\it truncation}. The main idea is to impose a limit on the influence that any single user can have on the SJA query result. 
For instance, consider a simple query: edge counting in a social network, in which a user can have a large number of connections (e.g., up to 5,000 on Facebook~\cite{Ugander2011Facebook}). With truncation, the query engine can impose an upper bound, e.g., up to 10 edges are counted for each user, which reduces the query sensitivity dramatically, leading to a lower noise requirement to satisfy DP, as we explain in Section~\ref{sec:dp-definition}.



Truncation, however, also has a negative impact on query result accuracy, since it introduces a bias by capping the influence of every user. Thus, to design an effective truncation scheme that maximizes result utility, one must carefully balance the error due to truncation bias and the noise injected to satisfy DP, which is a non-trivial problem.
Prior work {\sf R2T}~\cite{DBLP:conf/sigmod/DongF0TM22} addresses this by considering all possible truncation thresholds that are a power of two, which guarantees asymptotically optimal result accuracy, as we elaborate further in Section~\ref{sec:pre-trunc}. The downside of {\sf R2T}, however, is that for each truncation threshold, it needs to solve a linear program to choose aggregation units to retain~\cite{DBLP:conf/tcc/KasiviswanathanNRS13}. Since it is computationally expensive to solve so many linear programs, {\sf R2T} does not scale well to large datasets~\cite{DBLP:journals/pacmmod/FangY24}, as shown in our experiments in Section~\ref{sec:exp}.
Further, R2T is limited to SJA queries with scalar outputs, i.e., each query returns a single numeric result. For SJA queries with \textit{vector outputs}, e.g., those involving a GROUP BY clause, the best known solution to our knowledge is {\sf PMSJA}~\cite{{DBLP:journals/pacmmod/0007S023}}, which involves solving even more expensive quadratically constrained quadratic programs (QCQPs)~\cite{NesterovNemirovskii1994, PardalosVavasis1991}, which exacerbates the scalability issue. 

A natural approach to improving efficiency is sampling, which enables the user to reach a desirable trade-off between query result accuracy and computational costs by choosing an appropriate sample size. Unfortunately, incorporating sampling into user-level private SJA processing is challenging, since (i) it is unclear what objects should be sampled to stay within a computational budget, and (ii) existing truncation-based pipelines are calibrated in the $(\varepsilon,\delta)$-DP model, which is not easy to compose with sampling.
To our knowledge, the only user-level DP solution with sampling is {\sf S\&E}~\cite{DBLP:journals/pacmmod/FangY24}, which (i) samples users and then extends the sample by exploring their neighboring records, and (ii) applies an existing solution such as {\sf R2T} or {\sf PMSJA} on the extended sample set. Since the sampled users may remain highly correlated, the privacy analysis demands injecting far more noise, yielding errors that can exceed those of {\sf R2T} or {\sf PMSJA} by more than an order of magnitude.

This paper presents \textit{Differentially Private Sampling for Scale} ({\sf DP-S4S}), a mechanism that delivers user-level DP accuracy comparable to {\sf R2T} and {\sf PMSJA} while scaling to much larger databases. {\sf DP-S4S} samples {\it aggregation units} (e.g., join tuples) instead of users, which keeps samples compact and avoids the worst cross-user correlations as in \textsf{S\&E}. We show that the error introduced by sampling can be partially compensated for by its privacy amplification effect, which reduces the noise level to satisfy DP. Meanwhile, for vector SJA queries, we build new mathematical infrastructure to adapt the idea of {\sf PMSJA} to satisfy \textit{R\'enyi DP} (which composes well with sampling), before converting the privacy guarantee to the standard DP notion. In particular, we design the first smooth sensitivity mechanism under R\'enyi DP, which can be of independent interest.



Extensive experiments with real data demonstrate that for scalar queries, {\sf DP-S4S} sustains the accuracy of {\sf R2T} while dramatically improving efficiency. For vector queries, the extended version of {\sf DP-S4S} matches {\sf PMSJA} at a fraction of the runtime. Both versions of {\sf DP-S4S} consistently outperform {\sf S\&E}, achieving more than $10\times$ lower error at smaller sample rates under the same privacy budget.

%% file: sections/preliminaries.tex
\section{Background} \label{sec:pre}

Section~\ref{sec:dp-definition} provides preliminaries on the definitions and properties of differential privacy. Section~\ref{sec:pre-def} clarifies the problem setting, and Section~\ref{sec:pre-trunc} reviews notable existing solutions.

\subsection{Differential Privacy} \label{sec:dp-definition}
In this paper, we follow the standard centralized differential privacy setting, in which a trusted data curator has access to the raw data, and releases only query results that under differential privacy constraints. The adversary aims to infer private information from the DP-compliant query outputs. Since DP offers information-theoretic privacy guarantees, we do not rely on any assumption on the computational power of the adversary.
Let $\mathcal{I}$ be the domain of all instances of a given relational database schema, \textit{differential privacy} (\textit{DP}) is defined as follows: 

\begin{definition}[Differential Privacy~\cite{DBLP:journals/fttcs/DworkR14}]
For $\varepsilon>0$ and $\delta\geq 0$, a randomized mechanism $\M:\mathcal{I}\to\mathcal{O}$ is $(\varepsilon,\delta)$-DP, if for any pair of neighboring database instances $\mathbf{I}\sim \mathbf{I}'$ and any subset of outputs $O\subseteq\mathcal{O}$, it holds that
\[
\Pr[\M(\mathbf{I})\in O]\leq e^\varepsilon \cdot\Pr[\M(\mathbf{I}')\in O]+\delta\,.
\]
When $\delta=0$, we say the mechanism satisfies \textit{pure DP}. Otherwise (i.e., when $\delta>0$), we say the mechanism satisifes \textit{approximate DP}. Note that in the above definition, the concept of ``neighboring'' database instances $\mathbf{I}\sim \mathbf{I}'$ is left abstract, which we will clarify with concrete definitions in Section~\ref{sec:pre-def}.

A common approach to answering numerical queries under differential privacy is to inject random noise into the query results, whose scale is calibrated to the \textit{sensitivity} of the query.
The widely used notion of \textit{global sensitivity} (\textit{GS}) considers all possible neighboring instances in the domain $\mathcal{I}$, as follows.

\begin{definition}[Global Sensitivity~\cite{DBLP:conf/focs/DworkRV10}]\label{def:gs}
Given $F:\mathcal{I}\to \mathbb{R}^d$, the global sensitivity of $F$ is defined as 
\[
\GS_F:=\max_{\mathbf{I},\mathbf{I}' \in \mathcal{I}:\mathbf{I}\sim \mathbf{I}'} \|F(\mathbf{I})-F(\mathbf{I}')\|\,.
\]
\end{definition}

$\GS$ can be defined over $L_1$ or $L_2$ norms. The Laplace mechanism calibrates noise to the $L_1$-$\GS$ of the query to achieve pure DP:

\begin{lemma}[Laplace Mechanism]
Given $F:\mathcal{I}\to \mathbb{R}^d$, 
\[
\M_{\Lap}(\mathbf{I}):=F(\mathbf{I})+b\cdot \Lap^d(0,1)
\]
satisfies $(\varepsilon,0)$-DP when $b\geq \GS_{L_1}/\varepsilon$.
\end{lemma}

An important property of DP mechanisms is \textit{composability}, meaning that multiple DP mechanisms can be combined to form a composite mechanism that still satisfies DP, as follows.

\begin{lemma}[Composition Theorems~\cite{DBLP:conf/focs/DworkRV10,DBLP:journals/fttcs/DworkR14}]\label{lemma:composition}
Let $\M$ be the adaptive composition of DP algorithms $\M_1,\dots\M_k$, then
\begin{enumerate}
\item (Basic Composition) If each $\M_i$ is $(\varepsilon_i, \delta_i)$-DP, then $\M$ is $(\sum_i\varepsilon_i, \sum_i \delta_i)$-DP.
\item (Advanced Composition) If each $\M_i$ is $(\varepsilon,\delta)$-DP, then for any $\delta'>0$, $\M$ is $(\varepsilon\sqrt{2k\ln(1/\delta')} + k\varepsilon\frac{e^\varepsilon-1}{e^\varepsilon+1}, k\delta+\delta')$-DP.
\end{enumerate}
\end{lemma}

\end{definition}

The composability properly enables DP mechanism designers to split the problem into subtasks, each with a portion of the privacy parameters $\varepsilon$ and $\delta$ called its \textit{privacy budget}.  

\medskip
\noindent
\textbf{R\'enyi DP.} In Lemma~\ref{lemma:composition}, the basic composition is clean but loose for approximate DP, whereas the advanced composition is tighter but with a messy formula. Next, we introduce an alternative definition of differential privacy: R\'enyi DP, which has tight and clean composition, and can be easily converted to standard DP, as follows.

\begin{definition}[R\'enyi Differential Privacy (RDP)~\cite{DBLP:conf/csfw/Mironov17}]\label{def:rdp}
For $\alpha>1$ and $\rho>0$, a randomized mechanism $\M:\mathcal{I}\to\mathcal{O}$ is $(\alpha,\rho)$-RDP, if for any pair of neighboring datasets $\mathbf{I}\sim \mathbf{I}'$, it holds that
\[
\mathrm{D}_\alpha(\M(\mathbf{I}) \| \M(\mathbf{I}'))\leq \rho\,,
\]
where $\mathrm{D}_\alpha(P\|Q):=\frac{1}{\alpha-1}\ln\int_{z\in \operatorname{supp}(Q)} P(z)^\alpha Q(z)^{1-\alpha}\,\mathrm{d}z$ is the R\'enyi divergence of order $\alpha$.
\end{definition}

\begin{lemma}[RDP Composition~\cite{DBLP:conf/csfw/Mironov17}]\label{lemma:rdp-composition}
Let $\M$ be the adaptive composition of DP algorithms $\M_1,\dots\M_k$, where each $\M_i$ is $(\alpha, \rho_i)$-RDP, then $\M$ is $(\alpha, \sum_i \rho_i)$-RDP.
\end{lemma}

\begin{lemma}[Conversion between DP and RDP~\cite{DBLP:conf/tcc/BunS16,DBLP:conf/csfw/Mironov17}] \label{lem:vec-RDP-convert}
If $\M$ is $(\varepsilon, 0)$-DP, then $\M$ is also $(\alpha, \frac{\alpha\varepsilon^2}{2})$-RDP for any $\alpha>1$.
If $\M$ is $(\alpha,\rho)$-RDP, then $\M$ is also $\left(\rho+\frac{\ln (1/\delta)}{\alpha-1},\delta\right)$-DP for any $\delta>0$.
\end{lemma}

A common mechanism for enforcing R\'enyi DP is the Gaussian mechanism, which is based on the concept of Gaussian R\'enyi divergence, as follows.


\begin{lemma}[Gaussian R\'enyi Divergence~\cite{DBLP:journals/isci/GilAL13}]\label{lem:vec-diverge}
Let $\mathcal{N}(\boldsymbol{\mu}, \sigma^2 \mathbf{1}_d)$ denote a multivariate Gaussian distribution, then
\begin{align*}
&\mathrm{D}_\alpha(\mathcal{N}(\bsmu_1, \sigma_1^2 \mathbf{1}_d)\|\mathcal{N}(\bsmu_2, \sigma_2^2 \mathbf{1}_d))\\
=&\frac{\alpha}{2}\cdot \frac{\|\bsmu_1-\bsmu_2\|_2^2}{(1-\alpha)\sigma_1^2+\alpha\sigma_2^2}+\frac{d}{2(\alpha-1)} \ln\frac{\sigma_1^{2(1-\alpha)}\sigma_2^{2\alpha}}{(1-\alpha)\sigma_1^2+\alpha\sigma_2^2}
\end{align*}
if $(1-\alpha)\sigma_1^2+\alpha\sigma_2^2>0$, otherwise infinite.
\end{lemma}

\begin{lemma}[Gaussian Mechanism] \label{lem:ec-gauss}
Given $F:\mathcal{I}\to \mathbb{R}^d$, 
\[
\M_{\operatorname{Gauss}}(\mathbf{I}):=F(\mathbf{I})+\sigma\cdot \mathcal{N}(0,\mathbf{1}_d)
\]
satisfies $(\alpha,\rho)$-RDP when $\sigma\geq \GS_{L_2}\cdot \sqrt{\alpha/2\rho}$.
\end{lemma}


\medskip
\noindent
\textbf{Local sensitivity and smooth sensitivity.} Lastly, we define {\it local sensitivity} of a function $F$ with respect to a given instance $\mathbf{I}$ as: 
\[
\LS_{F}(\mathbf{I}):=\max_{\mathbf{I'\in \mathcal{I}}:\mathbf{I}\sim \mathbf{I}'} \|F(\mathbf{I})-F(\mathbf{I}')\|\,.
\]

Sometimes, a function $F$'s local sensitivity is far smaller than its global sensitivity, which is a situation we encounter later in Section~\ref{sec:vec}. However, it is known that injecting noise directly based on local sensitivities cannot guarantee DP~\cite{DBLP:books/sp/17/Vadhan17}. The following definition provides a middle ground between global and local sensitivities.

\begin{lemma}[Smooth-sensitivity~\cite{DBLP:conf/stoc/NissimRS07}]\label{lemma:ss}
Given $F:\mathcal{I}\to\mathbb{R}^d$ and an upper bound function $B(\mathbf{I})\geq \LS_{F}(\mathbf{I})$, the $\gamma$-smooth-sensitivity of $F$ is defined as 
\[
\operatorname{SS}_F(\mathbf{I},\gamma):=\max_{\mathbf{I}'\in \mathcal{I}} \{B(\mathbf{I}') \cdot e^{-\gamma\cdot \operatorname{dist}(\mathbf{I}, \mathbf{I}')}\}\,,
\]
where $\operatorname{dist}(\cdot,\cdot)$ is defined over the neighboring relationship.
It satisfies $\operatorname{SS}_F(\mathbf{I},\gamma)\geq \LS_F  (\mathbf{I})$ and $\operatorname{SS}_F(\mathbf{I},\gamma)\leq e^\gamma \cdot \operatorname{SS}_F(\mathbf{I}',\gamma)$ for any $\mathbf{I}\sim\mathbf{I}'$.
\end{lemma}


\subsection{SJA Queries with User-Level DP} \label{sec:pre-def}

\textbf{SJA queries.} Consider a database schema $\mathcal{R}=\{R_1,\dots,R_s\}$ where each $R_i$ is a logical relation.
A multi-way join is defined as 
\[
J=R_1(\mathbf{x}_1)\Join R_2(\mathbf{x}_2)\Join\dots\Join R_s(\mathbf{x}_s)\,,
\]
where each $\mathbf{x}_i$ is a set of attributes.
The set of variables in the full join is denoted $var(J)=\mathbf{x}_1\cup\dots\cup\mathbf{x}_s$.

A database instance $\mathbf{I}$ defined over $\mathcal{R}$ consists of tables $\mathbf{I}(R_1),\dots,\mathbf{I}(R_s)$.
Correspondingly, the join result $J(\mathbf{I})$ is a flat table $T$ where each tuple takes value in $\operatorname{dom}(var(J))$. A Select-Join-Aggregate (SJA) query is defined by both the join $J$ and an aggregation function $f:\operatorname{dom}(var(J))\to \mathbb{R}_{\geq 0}$ that assigns a weight $f(t)=w_t$ to each join result $t\in J(\mathbf{I})$. Evaluating query $f$ on $J(\mathbf{I})$ gives the query result
\[
Q(\mathbf{I}):=f(T)=\sum_{t\in T} w_t\,,
\]
where $T=J(\mathbf{I})$. Our target is to answer scalar or vector SJA queries.
We refer to each join result $t\in T$ in the flat table as a \textit{join tuple} or an \textit{aggregation unit} interchangeably.
We also assume the range of $f$ is $w_t\in [0,1]$ for simplicity.
This can always be done through normalization: if $w_t\in [0, W]$, then we can consider an $f'=f/W$ and rescale the query result by $W$ to answer $f$.

\medskip
\noindent
\textbf{User-level DP for SJA queries.} Next, we define user-level DP in the context of SJA query processing. Recall from Section~\ref{sec:dp-definition} that in the DP definition, the concept of neighboring databases instances $\mathbf{I}\sim \mathbf{I}'$ is left abstract. 
In this paper, we adopt ``user-level DP with foreign-key constraints''~\cite{DBLP:journals/pvldb/KotsogiannisTHF19,DBLP:conf/sigmod/DongF0TM22,DBLP:journals/pacmmod/FangY24}.
Specifically, we assume that there is a primary private relation $R_p\in\mathcal{R}$ where each record $u\in \mathbf{I}(R_p)$ is a private user.
The set of all users is denoted as $U:=\mathbf{I}(R_p)$.

We say a record $r\in \mathbf{I}(R)$ is \textit{contributed} by user $u$, denoted $u\leadsto r$, following a recursive definition: 1) $u\leadsto u$; 2) if there is an $r'\in \mathbf{I}(R')$ such that $u\leadsto r'$, and there is a foreign key in $r$ that references the primary key in $r'$, then we have $r$ is also contributed by $u$, i.e., $u\leadsto r$.
Two instances $\mathbf{I}$ and $\mathbf{I'}$ are neighbors (i.e., $\mathbf{I}\sim \mathbf{I}'$), if there exists a user $u$ such that for each record $r$ that $\mathbf{I}$ and $\mathbf{I'}$ differ, we have $u\leadsto r$. We call such a user $u$ a {\it witness}.


Let $n=|J(\mathbf{I})|$ be the number of join results, and $m=|U|$ be the number of users.
Note that by our definition, a user may contribute to multiple tuples, and a tuple may be contributed by multiple users. We assume that each user can have at most $\Delta$ tuples, and that each tuple is contributed by $l$ users.
Note that this model generalizes DP definitions in the literature: 
when $l=1$, it degenerates to ``user-DP without self-joins''~\cite{DBLP:conf/nips/LevySAKKMS21,DBLP:conf/nips/GhaziK0MMZ23}; when $\Delta=l=1$, it reduces to tuple-level DP; without a constraint on $\Delta$ or $l$, it is user-level DP.




\subsection{Existing Solutions} \label{sec:pre-trunc}

Although there are generic mechanisms to enforce differential privacy, e.g., Laplace and Gaussian mechanisms reviewed in Section~\ref{sec:dp-definition}, they cannot be directly applied to the problem of SJA query processing under user-level DP as described in Section~\ref{sec:pre-def}, in which the problem setting incurs a prohibitively high global sensitivity (Definition~\ref{def:gs}). In particular, assume that a record can have up to $D$ matching partners in one join operation, and that the aggregate is COUNT. With $t$ joins, in the worst case, one user's record $u$ can contribute $O(D^t)$ join results, leading to a global sensitivity of $O(D^t)$, and, thus, overwhelming noise required to satisfy DP.




\medskip
\noindent
\textbf{Trucation~\cite{DBLP:conf/tcc/KasiviswanathanNRS13}.} To reduce the global sensitivity of the SJA query, a common technique is \textit{truncation}, which imposes a constraint on the maximum influence of a user on the query result.
Formally, given a \textit{truncation threshold} $\tau$, we compute the result of a truncated SJA query where the contribution of any user is capped to $\tau$, meaning that the global sensitivity of the truncated query is bounded by $\tau$~\cite{DBLP:conf/tcc/KasiviswanathanNRS13}. 
This step involves solving a linear program~\cite{DBLP:conf/tcc/KasiviswanathanNRS13}, which minimizes the error after injecting noise to satisfy DP.


Clearly, the result of the truncated SJA query can differ from the original query, which we call the \textit{truncation bias}. Hence, the truncation threshold $\tau$ is a critical system parameter that balances truncation bias with DP noise: a larger $\tau$ reduces the truncation error, 
but at the same leads to a larger global sensitivity, and, thus, higher noise needed to satisfy DP, and vice versa.
The optimal value for $\tau$ depends upon the maximum contribution of a user in the dataset, called the {\it downward sensitivity}, defined as 
\[
\DS_F(\mathbf{I}):=\max_{\mathbf{I}'\subseteq {\mathbf{I}}:\mathbf{I}\sim \mathbf{I}'} \|F(\mathbf{I})-F(\mathbf{I}')\|\,.
\]

Note that choosing the truncation threshold $\tau$ using the exact downward sensitivity $\DS_F(\mathbf{I})$ would violate differential privacy~\cite{DBLP:conf/tcc/KasiviswanathanNRS13, DBLP:conf/sigmod/DongF0TM22}, since $\DS_F(\mathbf{I})$ itself may depend upon private information.  


\medskip
\noindent
\textbf{\textsf{R2T~\cite{DBLP:conf/sigmod/DongF0TM22}}.} 
The \textsf{R2T} algorithm chooses an appropriate trucation threshold $\tau$
by considering all the 
candidate $\tau$'s that are powers of $2$, each of which is allocated with a portion of the privacy budget. For every $\tau$, the algorithm solves the linear program in~\cite{DBLP:conf/tcc/KasiviswanathanNRS13}, whose target value gives the truncated query result. Then, \textsf{R2T} injects random noise according to the corresponding $\tau$'s to satisfy DP.
After subtracting the error bound, all of them become under-estimates with high probability, and the maximum among them provides a balance between the truncation error and the DP noise.
Since the candidate $\tau$'s are powers of $2$, the number of candidates is logarithmic; accordingly, the authors of \cite{DBLP:conf/sigmod/DongF0TM22} prove that the utility guarantee of \textsf{R2T} matches the instance-optimal lower bound $\Omega(\DS(\mathbf{I}))$ up to polylogarithmic factors.
 
\medskip
\noindent
\textbf{\textsf{PMSJA~\cite{DBLP:journals/pacmmod/0007S023}}.} \textsf{R2T} is designed for {\it scalar} queries that return a single value. For SJA queries that return {\it vectors} (e.g., GROUP BY aggregates), a natural idea is to exploit the composition property in Lemma~\ref{lemma:composition} to split the privacy budget among the vector entries. In particular, 
using advanced composition in Lemma~\ref{lemma:composition}, the $L_2$ error of the noisy result is then $\tilde{O}\left(\sqrt{d}\cdot \sqrt{\sum_{i=1}^d \DS_i^2(\mathbf{I})}\right)$, where $\DS_i(\mathbf{I})$ is the maximum contribution of an user to the $i$-th entry.
This is optimal if the user with maximum contribution is the same across all entries.
However, if each user contributes to a different entry, this error becomes $\tilde{O}(d)$, whereas the optimal bound is $\tilde{O}(\sqrt{d})$ as $\DS_F(\mathbf{I})=1$.
To address this issue, the PMSJA algorithm~\cite{DBLP:journals/pacmmod/0007S023} resorts to using a quadratically constrained quadratic program (QCQP) combined with local-sensitivity based mechanisms to achieve instance-optimal error $\tilde{O}(\sqrt{d}\cdot \DS_{F}(\mathbf{I}))$, which we revisit in Section~\ref{sec:vec}.

\medskip
\noindent
\textbf{\textsf{S\&E}~\cite{DBLP:journals/pacmmod/FangY24}}. The user-level DP solutions presented so far for SJA queries all involve solving expensive optimization programs, which are inherently difficult to scale. To reduce the computational costs, \textit{sample-and-explore} (\textsf{S\&E}) applies sampling to the set of users.

Specifically, {\sf S\&E} takes as input the number of iterations; in each iteration, it samples a user uniformly at random. The algorithm then builds an ``explored'' instance containing all the aggregation units where the sampled user appears first.
After that, a DP mechanism (e.g. $\textsf{R2T}$ or $\textsf{PMSJA}$) is applied on the explored instance. To obtain the estimator for an iteration, the output is scaled back by the number of users $m$, which is also the inverse of the probability that each aggregation unit is sampled. 
{\sf S\&E} then takes the average over multiple iterations to reduce the sample variance.

The privacy guarantee of {\sf S\&E} relies on a given upper bound $C$ on the number of \textit{collaborators} of each user. Two users \textit{collaborate} if they contribute to the same aggregation unit. 
The intuition is that on a neighboring dataset with an extra user, the explored instance will be the same unless the sampled user happens to be the witness or one of their collaborators, which occurs with probability approximately $C/m$. 

A major issue for {\sf S\&E} is that its privacy cost is linear to $C$.
which leads to large errors as we show in Section~\ref{sec:sca-analysis}.
The proposed approach \textsf{DP-S4S}, presented in the next section, 
overcomes this problem. In particular, our amplification bound has no assumption on $C$, and achieves a strict amplification in the privacy budget. As a result, the result utility of \textsf{DP-S4S} may even surpass that of the {\sf R2T} mechanism \textit{without sampling} in some cases, as shown in our experiments in Section~\ref{sec:exp}.

%% file: sections/sampling-for-a-single-query.tex
\section{\textsf{DP-S4S} for Scalar Queries}\label{sec:sca}

\begin{figure}[t]
\centering

\begin{minipage}[t]{0.3\linewidth}\vspace{0pt}\centering
  \input{fig/TriCnt}
  \captionsetup{skip=2mm}
  \caption{A social network $G$.}
  \label{fig:pre-graph}
\end{minipage}
\hfill
\begin{minipage}[t]{0.38\linewidth}\vspace{0pt}\centering\small
  \begin{tabular}{|c|c|c|}\hline
  \multirow{2}{*}{\bf User} & \multicolumn{2}{|c|}{\bf Triangle Count} \\\cline{2-3}
   & before trunc. & after trunc. \\\hline
   A, G & 1 & 1 \\\hline
   B, F & 2 & 2 \\\hline
   C, D, E & 3 & 2 \\\hline
  \end{tabular}
  \vspace{2mm}
  \captionsetup{type=table,skip=0mm}
  \captionof{table}{Users’ triangle counts in $G$.}
  \label{tab:pre-tri-cnt}
\end{minipage}
\hfill
\begin{minipage}[t]{0.3\linewidth}\vspace{0pt}\centering\small
  \begin{tabular}{|c|c|}\hline
  {\bf Triangle} & {\bf Weight} \\\hline
  $\triangle_{\text{ABC}},\ \triangle_{\text{EFG}}$ & 1 \\\hline
  $\triangle_{\text{BCD}},\ \triangle_{\text{DEF}}$ & 1 \\\hline
  $\triangle_{\text{CDE}}$ & 0 \\\hline
  \end{tabular}
  \vspace{2mm}
  \captionsetup{type=table,skip=0mm}
  \captionof{table}{A truncation with $\tau=2$.}
  \label{tab:pre-thresh}
\end{minipage}

\end{figure}


This section presents \textsf{DP-S4S} for scalar SJA queries that return a single aggregate value. The extension to vector SJA queries is deferred to Section~\ref{sec:vec}. The main idea is that instead of sampling users as in \textsf{S\&E}, \textsf{DP-S4S} samples {\it aggregate units}, i.e., the join tuples upon which the scalar query is aggregating. 

For example, consider a triangle count query with pure DP in a social network $G$, illustrated in Figure~\ref{fig:pre-graph}, where each vertex is a user. Table~\ref{tab:pre-tri-cnt} lists the number of triangles contributed by each user, before and after applying a truncation threshold of 2. The $5$ triangles shown in Table~\ref{tab:pre-thresh} are the aggregate units. \textsf{DP-S4S} first constructs a sample set $S$ of such aggregate units, and then applies truncations on $S$ to reduce the $L_1$ sensitivity of the aggregation on $S$; after that, it generates the aggregation result on the truncated $S$, injects Laplace noise, and then returns the noisy result scaled up by a factor decided by the sampling rate. The truncation threshold $\tau$ used on $S$ is decided using a differentially private algorithm, which ensures that the choice of $\tau$ does not lead to a privacy breach. 

In the following, we first focus on pure DP (i.e., ($\varepsilon$, 0)-DP) in Sections~\ref{sec:sca-sample}-\ref{sec:sca-analysis}, and extend the result to approximate DP in Section~\ref{sec:rdp-sample}. In particular, Section~\ref{sec:sca-sample} explains the truncation of a sample set $S$ given a threshold $\tau$. Section~\ref{sec:sca-choice} focuses on the DP algorithm for choosing $\tau$. Section~\ref{sec:sca-analysis} analyzes the theoretical guarantee of \textsf{DP-S4S} and compares against existing solutions.

\subsection{Sample and Truncate} \label{sec:sca-sample}

Recall from Section~\ref{sec:pre-def} that for a query $f$ over the join tuples $T$, our target is to compute $f(T):=\sum_{t\in T} w_t$.
Similar to~\citet{DBLP:conf/sigmod/DongF0TM22}, the proposed mechanism \textsf{DP-S4S} applies after computing all the joins and before the final aggregation.
Specifically, \textsf{DP-S4S} samples each join tuple $t\in T$ with probability $q$ to form $S$, and rescales the aggregation result $f(S)$ by $1/q$ as an estimator for $f(T)$. Thus, the problem reduces to how to release $f(S)$ accurately with user-level DP. Given a truncation threshold $\tau$, we can simply apply the steps in \cite{DBLP:conf/tcc/KasiviswanathanNRS13, DBLP:conf/sigmod/DongF0TM22} reviewed in Section~\ref{sec:pre-def}, i.e., solving a linear program to compute the optimal truncation for each user, and then applying the Laplace mechanism to achieve pure DP. 

Formally, Algorithm~\ref{algo:sample-truncate} presents the above algorithm for releasing $f(T,\tau)$ for a given threshold $\tau$ under ($\varepsilon$, 0)-DP. The linear program for computing optimal truncation~\cite{DBLP:conf/tcc/KasiviswanathanNRS13,DBLP:conf/sigmod/DongF0TM22} is as follows.

\begin{align}\label{eq:r2t}
&\underset{x_t}{\text{maximize}} && \sum_{t\in S} x_t \\ 
&\text{subject to} && \sum_{t:u\leadsto t} x_t \leq \tau, &&\forall  u\in U \nonumber\\
& && 0\leq x_t\leq  w_t && \forall t\in S \nonumber
\end{align}
where each variable $x_t\in[0,w_t]$ is the fractional contribution of join tuple $t$ in the sample set $S$ to the truncated query, subject to the constraint that the total contribution of $u$ does not exceed the truncation threshold $\tau$.
For a given query, minimizing the truncation bias is equivalent to maximizing the remaining weights; thus, the optimal target value $\bar{f}(S,\tau):=\sum_{t\in S} x_t^*$ is used as the truncated query result.
Note that for any $\tau\geq \DS({\mathbf{I}})$ (see Section~\ref{sec:pre-trunc}), the optimal solution is $x_t^*=w_t$ and we have $\bar{f}(S,\tau)=\sum_{t\in S} w_t=f(S)$.

After obtaining the result of the truncated query, the algorithm proceeds to add Laplace noise with scale $\tau/\varepsilon$ (Line 5), and then scale the result back by $1/q$ to correct for sampling (Line 6). 

Although Algorithm~\ref{algo:sample-truncate} is simple, its privacy analysis is highly non-trivial.
A major contribution of this work is to show that sampling join tuples (rather than users), followed by an $\varepsilon$-user-level DP mechanism on the sample, preserves a stronger $\varepsilon'$-user-level DP guarantee, for $\varepsilon'\leq \varepsilon$.
Note that prior work has only considered either sampling users for user-level DP~\cite{DBLP:journals/pacmmod/FangY24} or tuples for tuple-level DP~\cite{DBLP:conf/nips/BalleBG18}, but not a combination of both.
In addition, when the subsequent pure DP mechanism is truncation-based, we can further achieve an \textit{amplified} privacy guarantee due to randomness in subsampling.
Formally, we have the following result~\ref{lemma:pure}.



\begin{algorithm}[t]
\caption{Sample-and-Truncate for pure DP}\label{algo:sample-truncate}
\begin{algorithmic}[1]
\Require Join tuples $T$, query $f$, truncation threshold $\tau$, sample rate $q$, privacy parameter $\varepsilon$
\Ensure Private query result $\hat{f}(T,\tau)$
\State $S\gets \varnothing$
\ForEach{$t\in T$}
    \State Sample $t$ into $S$ with probability $q$
\EndFor
\State $\bar{f}(S,\tau)\gets$ Optimal target value to problem \eqref{eq:r2t} on $S$, $\tau$
\State Release $\hat{f}(S,\tau)\gets\bar{f}(S,\tau)+\Lap(\frac{\tau}{\varepsilon})$
\State Report $\hat{f}(T,\tau)\gets\frac{1}{q}\cdot \hat{f}(S,\tau)$
\end{algorithmic}
\end{algorithm}

\begin{lemma}[Sample-and-Truncate Amplification]\label{lemma:pure}
Algorithm~\ref{algo:sample-truncate} satisfies $\varepsilon'$-user-level DP for 
\begin{align*}
\varepsilon'(\varepsilon,\tau,\Delta,q)=\ln&\bigg(\max\bigg\{\sum_{k=0}^{\lfloor\tau\rfloor} p_k\exp\left(\frac{k}{\tau}\varepsilon\right)+\Pr[\operatorname{Bin}(\Delta,q)>\tau]\exp(\varepsilon),\\
&\frac{1}{\sum_{k=0}^{\lfloor\tau\rfloor} p_k\exp\left(-\frac{k}{\tau}\varepsilon\right)+\Pr[\operatorname{Bin}(\Delta,q)>\tau]\exp(-\varepsilon)}\bigg\}\bigg)\,,
\end{align*}
where $p_k=\Pr[\operatorname{Bin}(\Delta,q)=k]=\binom{\Delta}{k}q^k(1-q)^{\Delta-k}$.
\end{lemma}

\begin{proof}
We will show that $\M(T)=\hat{f}(S,\tau)$ is private, as the final step is a post-processing.
Denote the independent Bernoulli (a.k.a. Poisson) subsampling process by $\lambda$, then the probability of observing output $z$ can be expressed as\footnote{We abuse the notation of $\Pr[\cdot]$ for probability density as with convention.} 
\[
\Pr[\M(T)=z]=\sum_{S\subseteq T} \Pr[\lambda(T)=S]\cdot \Pr[\hat{f}(S,\tau)=z]\,.
\]
Consider a user-level neighbor $T'=T\cup T_u$ where the witness is $u$ and $T_u$ are the records contributed by $u$.
WLOG, we assume $|T_u|=\Delta$, which brings the maximum divergence. 
Due to independent sampling,
\[
\Pr[\M(T')=z]= \sum_{S\subseteq T} \Pr[\lambda(T)=S]\ \cdot \sum_{S_u\subseteq T_u} \Pr[\lambda(T_u)=S_u] \cdot \Pr[\hat{f}(S\cup S_u,\tau)=z]\,.
\]

We will use the fact that for $p_i, A_i, B_i>0$, 
\[\frac{\sum_ip_iA_i}{\sum_ip_iB_i}\leq \max_i\frac{A_i}{B_i}\,,\]
which can be verified by $A_i\leq B_i\cdot \max_i\frac{A_i}{B_i}$ for any $i$.
Borrowing ideas from~\cite{DBLP:journals/pvldb/JiangLYX24}, we consider
\begin{align*}
\frac{\Pr[\M(T)=z]}{\Pr[\M(T')=z]}=&\frac{\sum_{S\subseteq T}\Pr[\lambda(T)=S]\cdot \Pr[\hat{f}(S,\tau)=z]}{\sum_{S\subseteq T}\Pr[\lambda(T)=S]\cdot \sum_{S_u\subseteq T_u}\Pr[\lambda(T_u)=S_u] \Pr[\hat{f}(S\cup S_u,\tau)=z]}\\
\leq & \max_{S\subseteq T}\frac{\Pr[\hat{f}(S,\tau)=z]}{\sum_{S_u\subseteq T_u}\Pr[\lambda(T_u)=S_u]\Pr[\hat{f}(S\cup S_u,\tau)=z]} \\
=&\frac{1}{\sum_{S_u\subseteq T_u}\Pr[\lambda(T_u)=S_u]\frac{\Pr[\hat{f}(S^*\cup S_u,\tau)=z]}{\Pr[\hat{f}(S^*,\tau)=z]}}\,,
\end{align*}
where $S^*=\arg\max_{S\subseteq T}\frac{\Pr[\hat{f}(S,\tau)=z]}{\sum_{S_u\subseteq T_u}\Pr[\lambda(T_u)=S_u]\Pr[\hat{f}(S\cup S_u,\tau)=z]}$.
Consider the optimal solution $\{x_t\}$ that achieves $S^*$, it is also feasible for $S^*\cup S_u$, thus $\bar{f}(S^*,\tau)\leq \bar{f}(S^*\cup S_u,\tau)$. On the other hand, given an optimal solution to $S^*\cup S_u$, we can obtain a feasible solution for $S^*$ by setting $x_t'=0$ for $t\in S_u$. In doing so, the weight reduced is $\sum_{t\in S_u} x_t\leq \min\{\tau, |S_u|\}$, since $w_t \in [0, 1]$ (see Section~\ref{sec:pre-def}). Thus, we also have $\bar{f}(S^*\cup S_u,\tau)\leq \bar{f}(S^*)+\min\{\tau,|S_u|\}$. In summary, we have $|\bar{f}(S^*\cup S_u,\tau)-\bar{f}(S^*)|\leq \min\{\tau,|S_u|\}$ for any $S_u$.

Based on the above fact, we have
\[
\frac{\Pr[\hat{f}(S^*\cup S_u,\tau)=z]}{\Pr[\hat{f}(S^*,\tau)=z]} \geq \exp\left(-\frac{\varepsilon}{\tau}\left|\bar{f}(S^*\cup S_u,\tau)-\bar{f}(S^*,\tau)\right|\right)
\geq \exp\left(-\frac{\min\{\tau, |S_u|\}}{\tau}\varepsilon\right)\,,
\]
thus, 
\[
\frac{\Pr[\M(T)=z]}{\Pr[\M(T')=z]}\leq \frac{1}{\sum_{k=0}^\Delta p_k\cdot \exp(-\frac{\min\{\tau,k\}}{\tau}\varepsilon)}\,,
\]
where $p_k=\Pr[\operatorname{Bin}(\Delta,q)=k]=\binom{\Delta}{k}q^k(1-q)^{\Delta-k}$.
Similarly,
\begin{align*}
\frac{\Pr[\M(T')=z]}{\Pr[\M(T)=z]}\leq & \max_{S\subseteq T}\frac{\sum_{S_u\subseteq T_u}\Pr[\lambda(T_u)=S_u]\Pr[\hat{f}(S\cup S_u,\tau)=z]}{\Pr[\hat{f}(S,\tau)=z]} \\
=&\sum_{S_u\subseteq T_u}\Pr[\lambda(T_u)=S_u]\frac{\Pr[\hat{f}(S^\circ\cup S_u,\tau)=z]}{\Pr[\hat{f}(S^\circ,\tau)=z]}\\
\leq &\sum_{k=0}^\Delta p_k\cdot \exp\left(\frac{\varepsilon}{\tau} \min\{\tau, k\}\right)\,,
\end{align*}
where $S^\circ=\arg\max_{S\subseteq T}\frac{\sum_{S_u\subseteq T_u}\Pr[\lambda(T_u)=S_u]\Pr[\hat{f}(S\cup S_u,\tau)=z]}{\Pr[\hat{f}(S,\tau)=z]}$.

Putting it together, we get Lemma~\ref{lemma:pure}.
\end{proof}

\noindent
\textbf{Optimization with a PGF trick.} Lemma~\ref{lemma:pure} requires computing $\sum_{k=0}^{\lfloor\tau\rfloor} p_k\exp\left(\frac{k}{\tau}\varepsilon\right)$ and $\sum_{k=0}^{\lfloor\tau\rfloor} p_k\exp\left(-\frac{k}{\tau}\varepsilon\right)$, which takes $O(\tau)=O(\Delta)$ time, which can be expensive when $\tau$ is large.
We show that with the help of probability-generating function (PGF), we can obtain a numerically stable result in $O(1)$ time, assuming the CDF function of a binomial distribution can be evaluated in $O(1)$.

Let $K\sim \operatorname{Bin}(\Delta,q)$, consider the PGF of $K$, defined as
\[
G(z):=\mathbf{E}[z^K]=\sum_{k=0}^\Delta p_k\cdot z^k=\sum_{k=0}^\Delta \binom{\Delta}{k} (1-q)^{\Delta-k}(qz)^k=(1-q+qz)^\Delta\,.
\]
WLOG, we assume an integer $\tau$.
Then the target is to compute the sum of the first $\tau+1$ terms in $G(z)$ for $z=\exp(\pm\frac{\varepsilon}{\tau})$. 
Let $A=1-q+qz$. We divide each summand by $A^\Delta$. Then, the $k$-th term becomes
\[
\binom{\Delta}{k} \left(\frac{1-q}{A}\right)^{\Delta-k}\left(\frac{qz}{A}\right)^k =  \binom{\Delta}{k} \left(1-\frac{qz}{A}\right)^{\Delta-k}\left(\frac{qz}{A}\right)^k
= \Pr\left[\operatorname{Bin}(\Delta, \frac{qz}{A})=k\right]\,,
\]
which is exactly the probability of another binomial distribution.
Taking the sum is equivalent to querying its CDF.
Accordingly, we present in Lemma~\ref{lemma:simple} a faster computation of the amplification. 

\begin{lemma}[Simplified Computation]\label{lemma:simple}
Algorithm~\ref{algo:sample-truncate} satisfies $\varepsilon'$-user-level DP for 
\begin{align*}
\varepsilon'(\varepsilon,\tau,\Delta,q)
& = \max\bigg\{(\operatorname{logcdf}(\tau,\Delta,q_1) + \Delta\cdot  \log A_1) \oplus (\operatorname{logsf}(\tau,\Delta,q)+\varepsilon),\\
& \qquad \qquad -\bigg((\operatorname{logcdf}(\tau,\Delta,q_2)+\Delta\cdot \log A_2) \oplus(
\operatorname{logsf}(\tau,\Delta,q)-\varepsilon)\bigg)\bigg\}\,,
\end{align*}
where $\oplus$ denotes the $\operatorname{LogSumExp}$ function, $\operatorname{logcdf}$ and $\operatorname{logsf}$ are the cumulative distribution function and survival function for the binomial distribution respectively, $A_1=1-q+q\cdot\exp(\varepsilon/\tau)$, $q_1=q\cdot \exp(\varepsilon/\tau)/A_1$; $A_2=1-q+q\cdot\exp(-\varepsilon/\tau)$, $q_2=q\cdot \exp(-\varepsilon/\tau)/A_2$.
\end{lemma}


\subsection{Truncation Threshold Selection} \label{sec:sca-choice}

The Sample-and-Truncate mechanism presented in Algorithm~\ref{algo:sample-truncate} requires the knowledge of the truncation threshold $\tau$. As explained in Section~\ref{sec:pre-trunc}, choosing an appropriate $\tau$ under user-level DP is highly non-trivial, and it is a major contribution of \textsf{R2T}~\cite{DBLP:conf/sigmod/DongF0TM22}. In this subsection, we carefully combine the doubling research solution in \cite{DBLP:conf/sigmod/DongF0TM22} to Sample-and-Truncate, leading to the complete \textsf{DP-S4S} for scalar SJA query with pure DP in Algorithm~\ref{algo:sample-r2t}.

\begin{algorithm}[t]
\caption{\textsf{DP-S4S for scalar SJA and pure DP}}\label{algo:sample-r2t}
\begin{algorithmic}[1]
\Require  Join tuples $T$, query $f$, sample rate $q$, privacy parameter $\varepsilon$, tuple bound $\Delta$, probability $\beta$
\Ensure Private query result $\hat{f}(T)$
\State Compute sample set $S$ with Lines 1-3 of Algorithm~\ref{algo:sample-truncate}.
\State $\varepsilon_0\gets\varepsilon$
\State $L\gets \lfloor \log\Delta\rfloor+1$
\For{$i\gets L$ downto $1$}
    \State $\tau_i\gets 2^{i-1}$, $\varepsilon_i\gets \varepsilon_{0} / i$
    \State $\bar{f}(S, \tau_i)\gets$ Optimal target value to problem~\eqref{eq:r2t} on $S,\tau_i$
    \State $\hat{f}(S,\tau_i)\gets \bar{f}(S,\tau_i)+\Lap(\frac{\tau_i}{\varepsilon_i})$
    \State $\tilde{f}(S,\tau_i)\gets \hat{f}(S,\tau_i)-\frac{\tau_i}{\varepsilon_i}\ln\frac{3L}{\beta}$
    \State $\varepsilon_i'\gets \varepsilon'(\varepsilon_i, \tau_i,\Delta,q)$ by Lemma~\ref{lemma:simple}
    \State $\varepsilon_0\gets \varepsilon_0-\varepsilon_i'$
\EndFor
\State Report $\hat{f}(T)\gets\frac{1}{q}\cdot \max_i \tilde{f}(S,\tau_i)$
\end{algorithmic}
\end{algorithm}

Recall from Lemma~\ref{lemma:pure} that $\varepsilon'\leq \varepsilon$, and that a mechanism on the sampled join tuples preserves user-level pure DP. Therefore, a naive solution is to allocate $\varepsilon_i=\varepsilon/L$ to each iteration. However, this approach fails to exploit the amplification from combining a truncation mechanism with subsampling.
Ideally, one wish to adopt some $\varepsilon^*>\varepsilon$ such that the amplified $\varepsilon'(\varepsilon^*,\tau,\Delta,q)\leq \varepsilon$.
While the amplified $\varepsilon'(\varepsilon,\tau,\Delta,q)$ is easily computable for a given $\varepsilon$, finding the optimal $\varepsilon^*$ involves a binary search.
Observe from Lemma~\ref{lemma:pure} that the amplification effect is more significant for large $\tau$'s. Accordingly, we adopt an \textit{adaptive budget allocation} strategy in \textsf{DP-S4S}.
In particular, we run the $L$ subroutines in descending order of $\tau_i$ (Line 4 of Algorithm~\ref{algo:sample-r2t}). For each $\tau_i$, we allocate a target privacy consumption $\varepsilon_i$ inversely proportional to the number of remaining subproblems, i.e.~applying basic composition in Lemma~\ref{lemma:composition}.
After obtaining the private result, we only charge its actual consumption from the remaining budgets (Lines 9-10). Note that the algorithm also takes as input a probability $\beta$ (used in Line 8), which is important for its utility analysis, presented in the next subsection.


\medskip
\noindent 
{\bf Remark.} In \cite{DBLP:journals/pacmmod/FangY24}, the authors also propose a faster variant of \textsf{S\&E} for scalar queries that explicitly requires that a good truncation threshold $\tau$ be given in advance; in their experiments, $\tau$ is tuned from data, which might leak private information according to \cite{DBLP:conf/sigmod/DongF0TM22}. In this paper, we follow \cite{DBLP:conf/sigmod/DongF0TM22} to properly select $\tau$ with differential privacy, and do not consider this variant of \textsf{S\&E}.




\subsection{Utility Analysis} \label{sec:sca-analysis}

The following lemma establishes the correctness of \textsf{DP-S4S} for scalar SJA queries under user-level pure DP, as well as its result utility guarantee.
The error bound covers both sampling error and noise injected to satisfy DP.

\begin{lemma}\label{lemma:s4s-s}
Algorithm~\ref{algo:sample-r2t} satisfies $(\varepsilon,0)$-DP.
Let $n$ be the number of records, $m$ be the number of users and $q$ be the sample rate.
With probability $1-\beta$, its error is 
\[
|\hat{f}(T)-f(T)|=\tilde{O}\left(\sqrt{\frac{n}{q}}+\frac{1}{q}+\frac{1}{\varepsilon_{i^*}}\left(\tau^*(T)+\sqrt{\frac{\tau^*(T)}{q}}\right)\right)\,,
\]
where $\tilde{O}(\cdot)$ omits constants and polylogarithmic factors in $m$ and $1/\beta$; $\tau^*(T):=\max_{u\in U}\sum_{t\in T}\mathbf{1}[u\leadsto t]$ is the maximum number of tuples a user contribute to in $T$; and $\varepsilon_{i^*}\geq \varepsilon/L$ is the 
allocated noise‑calibration budget 
where $L=\lfloor\log\Delta\rfloor+1$.
\end{lemma}

The proof of Lemma~\ref{lemma:s4s-s} can be found in 
\ifthenelse{\boolean{fullver}}{
Appendix~\ref{sec:proof2}}{
Appendix~B.1 of the full version~\cite{full-ver}}.
As a sketch, main ideas of the proof include (i) we apply standard Bernstein bounds for the sampling error and (ii) for DP noise, conditioned on the high probability event that $\tilde{f}$'s are all underestimates, we show that the output error depends on $\tau^*(S)$, i.e., the maximum contribution of a user in the sample.
A key observation in the proof is that besides the amplification effect that reduces the DP noise scale (see Lemma~\ref{lemma:pure}), another benefit of sampling in \textsf{DP-S4S} is that it reduces the maximum contribution among the users. Roughly speaking, if the maximum contribution of a user is $\tau^*(T)$ for the raw join results, then we expect the maximum contribution to be $\tau^*(S)\approx q\cdot \tau^*(T)$ in the sample, subject to some lower-order terms. 
After scaling the results by $\frac{1}{q}$, we obtain $\tau^*(T)$, so that the DP error does not increase. 
\ifthenelse{\boolean{fullver}}{Appendix~\ref{sec:ablation}}{
Appendix C in~\cite{full-ver}} evaluates the contribution of both error sources.

\medskip
\noindent
\textbf{Comparison with \textsf{R2T}.} The utility guarantee of R2T is as follows. 
\begin{lemma}[Result Utility of R2T~\cite{DBLP:conf/sigmod/DongF0TM22}]\label{lemma:r2t}
With probability $1-\beta$, the output of the {\sf R2T} mechanism satisfies
\[
|\M_\textsf{R2T}(T)-f(T)|=\tilde{O}\left(\frac{L}{\varepsilon}\cdot \DS(\mathbf{I})\right)\,.
\]
\end{lemma}

Compared to R2T, \textsf{DP-S4S} introduces an additional error term $\sqrt{n/q}$ for $q=\Omega((\tau^*(T))^{-1})$, due to sampling error.
For counting queries, $\tau^*(T)=\DS(\mathbf{I})$ is the downward sensitivity of query $f$ (see Section~\ref{sec:pre-trunc}). Thus, the sampling-based mechanism achieves the same asymptotic error for a sufficiently large sample rate. In addition, we improve the $\frac{L}{\varepsilon}$ term in the DP error to $\frac{1}{\varepsilon_{i^*}}\leq \frac{L}{\varepsilon}$, which benefits from the privacy amplification.
To see why, note that the first iteration allocates exactly $\varepsilon_L=\varepsilon/L$.
Meanwhile, its amplified privacy guarantee $\varepsilon_L'\leq \varepsilon_L$, meaning the remaining budget $\varepsilon_0\geq \varepsilon-\varepsilon/L=\varepsilon\cdot (1-1/L)$. 
By induction, we have in each iteration the allocated $\varepsilon_i$ is at least as large as $\varepsilon/L$.

\begin{algorithm}[t]
\caption{Sample-and-Truncate for R\'enyi DP}\label{algo:sample-truncate-renyi}
\begin{algorithmic}[1]
\Require Join tuples $T$, query $f$, truncation threshold $\tau$, sample rate $q$, privacy parameters $\alpha$ and $\rho$
\Ensure Private query result $\hat{f}(T,\tau)$
\State Compute $\bar{f}(S,\tau)$ with Lines 1-4 in Algorithm~\ref{algo:sample-truncate}
\State Release $\hat{f}(S,\tau)\gets\bar{f}(S,\tau)+\mathcal{N}(0,\sigma^2)$ for $\sigma=\tau\cdot \sqrt{\alpha/2\rho}$
\State Report $\hat{f}(T,\tau)\gets\frac{1}{q}\cdot \hat{f}(S,\tau)$ as in Algorithm~\ref{algo:sample-truncate}
\end{algorithmic}
\end{algorithm}

\medskip
\noindent
\textbf{Comparison with \textsf{S\&E}.}
\citet{DBLP:journals/pacmmod/FangY24} do not include a closed-form error bound for \textsf{S\&E} under pure DP. In the following, we show that its error can be large even when the exact (i.e., non-private) best truncation threshold is used.

Specifically, the error analysis in \cite{DBLP:journals/pacmmod/FangY24} for running {\sf S\&E} with {\sf R2T} focuses on approximate DP. It is not hard to see that when each of the $k$ iterations is allocated with a privacy budget of $\varepsilon_i=O(\varepsilon/k)$, the mechanism satisfies pure DP, according to basic composition (Lemma~\ref{lemma:composition}).
When $\varepsilon_i$ is small, the amplified $\varepsilon_i'$ for each iteration can be approximated as
\[
\varepsilon_i'=\ln\left(1+\frac{m}{C}\left(e^{\varepsilon_i}-1\right)\right)=
O\left(\frac{m \varepsilon}{kC}\right)\,, 
\]
where $m$ is the inverse of the sample rate. The final estimator scales the sample average by $m$. Thus, its variance from DP noise is 
\[
\frac{m^2}{k^2} \cdot \left(\sum_{i=1}^k \frac{2\tau^2}{(\varepsilon_i')^2}\right)=
O\left(kC^2\cdot\frac{\tau^2}{\varepsilon^2}\right)\,.
\]
Thus, the injected error to satisfy pure DP grows with $\sqrt{k}$, which becomes high when we set a large $k$, which is important for reducing sampling error. More importantly, its DP error is also linear to $C$, the maximum number of collaborators for each user (see Section~\ref{sec:pre-trunc}), which is data-dependent and can be large, as our experiments show in Section~\ref{sec:exp}. The proposed approach \textsf{DP-S4S} does not suffer from either of these issues, according to Lemma~\ref{lemma:s4s-s}.


Finally, as discussed in \ifthenelse{\boolean{fullver}}{Appendix~\ref{section:sande}}{Appendix~D of the full version~\cite{full-ver}}, the sampling error of \textsf{S\&E} is also a factor of $l$ larger compared to our \textsf{DP-S4S}, since the former takes \textit{all} tuples associated with the sampled user, whereas the latter performs sampling on aggregation units, which automatically takes into account the correlation between users.

\subsection{\textsf{DP-S4S} for Approximate DP} \label{sec:rdp-sample}

So far, our description of \textsf{DP-S4S} aims at enforcing pure DP, i.e., $\delta=0$. For approximate DP, i.e., $\delta>0$, we adapt \textsf{DP-S4S} to employ the Gaussian mechanism (Lemma~\ref{lem:ec-gauss}) to satisfy R\'enyi-DP, and subsequently convert R\'enyi DP to ($\varepsilon$, $\delta$)-DP using Lemma~\ref{lem:vec-RDP-convert}. Algorithm~\ref{algo:sample-truncate-renyi} adapts Algorithm~\ref{algo:sample-truncate} to enforce R\'enyi DP.

Under R\'enyi-DP, we derive privacy amplification as follows.

\begin{lemma}[RDP Amplification]\label{lemma:sample-truncate}
Algorithm~\ref{algo:sample-truncate-renyi} satisfies $(\alpha,\rho')$-user-level RDP for
\[
\rho'(\alpha,\rho,\tau,\Delta,q)=\frac{1}{\alpha-1}\ln\bigg(\sum_{k=0}^{\lfloor\tau\rfloor}p_k\cdot\exp\left((\alpha-1)\frac{k^2}{\tau^2}\rho  \right)  + \Pr[\operatorname{Bin}(\Delta, q)>\tau]\cdot \exp((\alpha-1)\rho)\bigg)\,,
\]
where $p_k=\Pr[\operatorname{Bin}(\Delta,q)=k]=\binom{\Delta}{k}q^k(1-q)^{\Delta-k}$.
\end{lemma}

The proof can be found in \ifthenelse{\boolean{fullver}}{Appendix~\ref{section:rdp-amp}}{Appendix~B.2 of the full version~\cite{full-ver}},
which is similar to that for Lemma~\ref{lemma:pure}. The difference is that we apply properties of R\'enyi divergence instead of analyzing the ratio between probabilities. 
We also have $\rho'\leq \frac{1}{\alpha-1}\ln(\sum_{k=0}^\Delta p_k\cdot \exp((\alpha-1)\rho))=\rho$.\footnote{
While we show that sampling join tuples amplifies user-level DP under both pure and R\'enyi DP, the same proof does not easily extend to $(\varepsilon,\delta)$-DP mechanisms outside the R\'enyi-DP class.
It is left open whether an arbitrary approximate-DP mechanism on such a sample preserves the same level of user-level privacy.
}

Using Lemma~\ref{lemma:sample-truncate} as a building block, we can arrive at a variant of Algorithm~\ref{algo:sample-r2t} that achieves $(\alpha,\rho)$-RDP, which implies $(\varepsilon,\delta)$-DP: we only need to change the noise and noise bound to Gaussian, and allocate $\rho_i$ instead of $\varepsilon_i$ during each iteration.
For the utility analysis, note that the sampling error is independent of the noise, and we also have the same bound for the user contribution in the sample. Thus, its utility is given as follows.

\begin{corollary}[RDP-S4S for scalar SJA]\label{lemma:s4s-s-rdp}
There is an $(\alpha,\rho)$-RDP version of Algorithm~\ref{algo:sample-r2t}.
Let $n$ be the number of records, $m$ be the number of users, $q$ be the sample rate and $\tau^*(T)$ be the maximum contribution of a user.
With probability $1-\beta$, its error is 
\[
|\hat{f}(T)-f(T)|=\tilde{O}\left(\sqrt{\frac{n}{q}}+\frac{1}{q}+\sqrt{\frac{\alpha}{\rho_{i^*}}}\left(\tau^*(T)+\sqrt{\frac{\tau^*(T)}{q}}\right)\right)\,,
\]
where $\rho_{i^*}\geq \rho/(\lfloor\log\Delta\rfloor+1)$ is the amplified privacy budget.
\end{corollary}

In particular, this indicates an $(\varepsilon,\delta)$-DP variant by taking $\rho=\varepsilon-\frac{\ln(1/\delta)}{\alpha-1}$. Now, the DP error scales with $\sqrt{\frac{\alpha(\alpha-1)}{\varepsilon(\alpha-1)-\ln(1/\delta)}}$, and the order $\alpha$ can be optimized for given $\varepsilon$ and $\delta$.

%% file: fig/TriCnt.tex
\def \globalscale {0.9}
\begin{tikzpicture}[y=1pt, x=1pt, yscale=\globalscale,xscale=\globalscale, every node/.append style={scale=\globalscale}, inner sep=0pt, outer sep=0pt]
  \begin{scope}[shift={(-104.8, 160.7)}]
    \node[text=black,anchor=south west,line width=0.7pt] (text1) at (156.8, 
  -132.5){D};

    \node[text=black,anchor=south west,line width=0.7pt] (text2) at (128.9, 
  -118.2){B};

    \node[text=black,anchor=south west,line width=0.7pt] (text2-8) at (185.5, 
  -118.2){F};

    \node[text=black,anchor=south west,line width=0.7pt] (text2-8-8) at (194.2, 
  -160.7){E};

    \node[text=black,anchor=south west,line width=0.7pt] (text2-8-8-5) at 
  (208.1, -132.5){G};

    \node[text=black,anchor=south west,line width=0.7pt] (text2-5) at (106.3, 
  -132.5){A};

    \node[text=black,anchor=south west,line width=0.7pt] (text3) at (118.5, 
  -160.6){C};

    \path[draw=black,fill=white,line width=1.0pt] (132.0, -125.0) circle (3.7pt);

    \path[draw=black,fill=white,line width=1.0pt] (160.3, -139.5) circle (3.7pt);

    \path[draw=black,fill=white,line width=1.0pt] (188.6, -125.0) circle (3.7pt);

    \path[draw=black,fill=white,line width=1.0pt] (132.0, -153.4) circle (3.7pt);

    \path[draw=black,fill=white,line width=1.0pt] (108.9, -139.5) circle (3.7pt);

    \path[draw=black,fill=white,line width=1.0pt] (188.6, -153.4) circle (3.7pt);

    \path[draw=black,fill=white,line width=1.0pt] (211.6, -139.5) circle (3.7pt);

    \path[draw=black,fill=white,even odd rule,line cap=butt,line join=miter,line
   width=0.7pt] (129.0, -127.0) -- (112.0, -137.6);

    \path[draw=black,fill=white,even odd rule,line cap=butt,line join=miter,line
   width=0.7pt] (132.0, -128.7) -- (132.0, -149.7);

    \path[draw=black,fill=white,even odd rule,line cap=butt,line join=miter,line
   width=0.7pt] (112.1, -141.4) -- (128.9, -151.5);

    \path[draw=black,fill=white,even odd rule,line cap=butt,line join=miter,line
   width=0.7pt] (157.0, -141.1) -- (135.3, -151.8);

    \path[draw=black,fill=white,even odd rule,line cap=butt,line join=miter,line
   width=0.7pt] (163.5, -141.1) -- (185.3, -151.8);

    \path[draw=black,fill=white,even odd rule,line cap=butt,line join=miter,line
   width=0.7pt] (135.7, -153.4) -- (185.0, -153.4);

    \path[draw=black,fill=white,even odd rule,line cap=butt,line join=miter,line
   width=0.7pt] (188.6, -149.7) -- (188.6, -128.7);

    \path[draw=black,fill=white,even odd rule,line cap=butt,line join=miter,line
   width=0.7pt] (185.4, -126.7) -- (163.5, -137.9);

    \path[draw=black,fill=white,even odd rule,line cap=butt,line join=miter,line
   width=0.7pt] (135.3, -126.7) -- (157.0, -137.9);

    \path[draw=black,fill=white,even odd rule,line cap=butt,line join=miter,line
   width=0.7pt] (191.7, -127.0) -- (208.5, -137.6);

    \path[draw=black,fill=white,even odd rule,line cap=butt,line join=miter,line
   width=0.7pt] (208.5, -141.4) -- (191.7, -151.5);

    \node[text=black,anchor=south west,line width=0.7pt] (text18) at (119.1, 
  -157.9){};

  \end{scope}

\end{tikzpicture}

%% file: sections/sampling-for-multiple-queries.tex
\section{\textsf{DP-S4S} for Vector Queries} \label{sec:vec}

This section focuses on {\sf DP-S4S} for \textit{vector} queries that return $d>1$ scalars (e.g., a group-by count query). Observe that a vector query $F$ can be viewed as $d$ scalar queries $F=\{f_1,\dots,f_d\}$. 

As discussed in Section~\ref{sec:pre-trunc}, a naive idea for answering a vector SJA queries with user-level DP is to process each component scalar query separately, each with an allocated privacy budget according to Lemma~\ref{lemma:composition}.
The problem of this approach, however, is that its error depends on the sensitivity of each scalar, rather than the sensitivity of the whole vector, which can be significantly smaller. In particular, the $L_2$ sensivitity for a $d$-dimensional vector query is smaller than that of its component scalar queries by a factor of $\sqrt{d}$ \cite{DBLP:journals/pacmmod/0007S023}.  {\sf PMSJA}~\cite{DBLP:journals/pacmmod/0007S023} address this issue by considering the following quadratically constrained quadratic program (QCQP):

\begin{align}\label{eq:qcqp}
&\underset{y_u, z_{f,t}}{\text{maximize}} && \sum_{u\in U} y_u \\
&\text{subject to} && \sum_{f\in F}\left(\sum_{t\in T_f:u\leadsto t} w_{f,t}\cdot z_{f,t}\right)^2 \leq \tau^2, &&\forall  u\in U \nonumber \\
& && \sum_{u: u\leadsto t} (y_u-1) \leq z_{f,t}-1, && \forall f\in F, t\in T_f \nonumber \\
& && 0\leq y_u\leq 1, && \forall u\in U \nonumber \\
& && 0\leq z_{f,t}\leq  1, && \forall f\in F,t\in T_f \nonumber
\end{align}


For given $\tau$, let $y_u^*, z_{f,t}^*$ be the optimal solution to Program~\eqref{eq:qcqp} (with ties broken arbitrarily). The truncated query for $f\in F$ is then 
\begin{align} \label{eq:vec-trunc}
\textstyle \bar{f}(T_f):=\sum_{t\in T_f} w_{f,t}\cdot z_{f,t}^{*}.
\end{align}

Unlike in {\sf R2T}, the global sensitivity of the truncated queries $\bar{F}$ is still large.
To address this issue, {\sf PMSJA} considers the local sensitivity of $\bar{F}$ (explained in Section~\ref{sec:dp-definition}).
Analysis shows that the $L_2$ local sensitivity satisfies:
\begin{equation}\label{eq:pmsja_ls_bound}
\LS_{\bar{F}}(\mathbf{I})\leq 2\tau\cdot (E(\mathbf{I}) + 1),
\end{equation}
where $E(\mathbf{I}):=\sum_{u\in U}(1-y_u^*)$ is the number of truncated users.
Intuitively, $E(\mathbf{I})$ has global sensitivity $1$, since the optimal solution on $\mathbf{I}$ can be transferred into a feasible solution on $\mathbf{I}'$ which contains one more user, by additionally truncating the extra user and all his/her tuples.
The formal proof can be found in~\cite{DBLP:journals/pacmmod/0007S023}.
{\sf PMSJA} then invokes an $(\varepsilon, \delta)$-DP algorithm to release $\bar{F}$.
For choosing the threshold $\tau$, it runs the sparse-vector-technique (SVT) mechanism~\cite{DBLP:journals/fttcs/DworkR14} to privately find the smallest power of $2$ such that only $E(\mathbf{I})=\tilde{O}(1)$ users are truncated. 

To achieve scalability, the main idea of \textsf{DP-S4S} is to combine join-tuple sampling described in Section~\ref{sec:sca} with {\sf PMSJA}. Note that {\sf PMSJA} enforces approximate DP; hence, we adopt the approach in Section~\ref{sec:rdp-sample} that processes the query under R\'enyi-DP, and subsequently converts the privacy guarantee to the standard ($\varepsilon$, $\delta$)-DP. 

\vspace{3pt}
\noindent
\textbf{Challenges.}
There are two major challenges in realizing this idea. First, as mentioned above, {\sf PMSJA} involves the use of local sensitivity, which cannot be directly used to calibrate DP noise as described in Section~\ref{sec:dp-definition}. For ($\varepsilon$, $\delta$)-DP, we can inject noise according to the smooth sensitivity (Lemma~\ref{lemma:ss}), which provides a viable middle ground that satisfies differential privacy while reducing the noise scale compared to global sensitivity. To our knowledge, however, there is no existing fundamental mechanism that utilizes smooth sensitivities to achieve R\'enyi DP, which is more strict than $(\varepsilon, \delta)$-DP. Second, even with the necessary tools in place to enforce R\'enyi DP with smooth sensitivity, it is still highly non-trivial to derive a similar sample-and-truncate algorithm with privacy amplification for {\sf PMSJA}, which involves solving a far more complex QCQP than the linear program in Section~\ref{sec:sca-sample}.

In the following, Sections~\ref{sec:rdp-ls} and \ref{sec:s4s-v} address these two challenges, respectively.

\subsection{Smooth Sensitivity Mechanism for R\'enyi DP}\label{sec:rdp-ls}

In this subsection, we build a new fundamental R\'enyi DP mechanism for the scenario that the local sensitivity is far smaller than the global sensitivity, i.e., $\LS_{L_2} \ll \GS_{L_2}$.
This is common in practical applications~\cite{DBLP:journals/fttcs/DworkR14,DBLP:books/sp/17/Vadhan17,DBLP:conf/sigmod/DongF0TM22,DBLP:journals/pacmmod/0007S023,DBLP:journals/pacmmod/FangY24}.
Utilizing Lemma~\ref{lem:vec-diverge}, we have the following result.

\begin{lemma}[RDP Smooth Sensitivity Mechanism]\label{lemma:additive-smooth}
Given $F:\mathcal{I}\to\mathbb{R}^d$ and its smooth-sensitivity function $
\operatorname{SS}_F:\mathcal{I}\times\mathbb{R}\to\mathbb{R}$ as in Lemma~\ref{lemma:ss}, the RDP-SS mechanism defined as
\[
\M_{\operatorname{RDP-SS}}(\mathbf{I}):=F(\mathbf{I})+\sigma(\mathbf{I})\cdot \mathcal{N}(0,\mathbf{1}_d)
\]
satisfies $(\alpha,\rho)$-RDP for \[\sigma(\mathbf{I})=\frac{\operatorname{SS}_F(\mathbf{I},\gamma)}{1-\alpha+\alpha e^{-2\gamma}}\cdot \sqrt{\frac{\alpha}{\rho}}\,,\]
where $\gamma:=\frac{1}{2}\ln\min\{t_1^{-1}, t_2\}$, and $0<t_1<1<t_2$ are solutions to 
\[
h(t):=t^\alpha-\alpha e^\frac{(\alpha-1)\rho}{d} t +(\alpha-1)e^\frac{(\alpha-1)\rho}{d}=0\,.
\]
\end{lemma}

The proof can be found in \ifthenelse{\boolean{fullver}}{Appendix~\ref{sec:proof3}}{Appendix~B.3 of the full version~\cite{full-ver}}. 
The intuition is that to bound the R\'enyi-divergence between two Gaussian distributions, there are two constraints in Lemma~\ref{lem:vec-diverge}: the noise scales must be large enough to bound the first term, which depends on the difference between their means. In addition, the scales must also be multiplicatively close to ensure the second term is bounded: when $\sigma_1=\sigma_2$, the second term takes the minimum value of $0$; when the variances differ, we must bound their ratio, which can be achieved by properly setting the smoothness factor $\gamma$.
After deciding $\gamma$ (see the proof for details), we calibrate the noise scale according to both $\gamma$ and the $\gamma$-smooth sensitivity $\operatorname{SS}_F$.

\medskip
\noindent
\textbf{Implementation.} 
A major challenge in implementing smooth-sensitivity mechanisms is that the $\operatorname{SS}_F$ function is hard to compute in general~\cite{DBLP:conf/stoc/NissimRS07}. In the following, we make an assumption that there exists an upper bound $B$ on the $L_2$ local sensitivity $\LS_F$ such that $B(\mathbf{I})\geq \LS_F(\mathbf{I})$ and $|B(\mathbf{I})-B(\mathbf{I}')|\leq 1$ for neighboring datasets.
Under this assumption, we can compute $\operatorname{SS}_F$ in constant time.
We first argue that this is a feasible assumption, which was also made by existing work~\cite{DBLP:conf/tcc/KasiviswanathanNRS13,DBLP:journals/tods/KarwaRSY14,DBLP:books/sp/17/Vadhan17}.
A trivial bound that satisfies this constraint is a constant function $B(\cdot)\equiv \GS$.
But it does not utilize the fact that $\LS\ll \GS$.
To reduce noise, the goal is to set $B$ as close to $\LS$ as possible.
Note that due to normalization, requiring $|B(\mathbf{I})-B(\mathbf{I}')|\leq 1$ is equivalent to the more general form of $|B(\mathbf{I})-B(\mathbf{I}')|\leq \tau$.
This is exactly how we will apply this lemma in our mechanism for vector queries in Section~\ref{sec:s4s-v}, following Equation~\eqref{eq:pmsja_ls_bound}.

Algorithm~\ref{algo:rdp-ss} implements the $\M_{\operatorname{RDP-SS}}$ mechanism. Here, the equation $h(t)=0$ is solved using Brent's method~\cite{brent2013algorithms} (Line 1).
Under the assumption for $B$, for any $\mathbf{I}'\in \mathcal{I}$,
\[
\LS(\mathbf{I}')\leq  B(\mathbf{I}')\leq B(\mathbf{I})+\operatorname{dist}(\mathbf{I},\mathbf{I}')\,.
\]
Therefore, $G(\mathbf{I}):=\max_{k\in\mathbb{N}_{\geq 0}} (k+B(\mathbf{I}))\cdot e^{-\gamma k}$ is the $\gamma$-smooth sensitivity of $F$ by Lemma~\ref{lemma:ss}.
In addition, to compute $G(I)$, we only need to evaluate two values of $k$, as follows.
Consider $g(k)=(k+b)\cdot e^{-\gamma k}$, we have its derivative $g'(k)=(1-\gamma k-\gamma b)e^{-\gamma k}$, thus the maximum is taken at the closest integers to $k^*=\frac{1}{\gamma}-b$.
Therefore, each $G(\mathbf{I})$ can be obtained by computing two values at most (Lines 4 and 6) and is bounded by
\[
G(\mathbf{I})\leq \textstyle \frac{1}{\gamma}\cdot e^{\gamma\cdot B(\mathbf{I})-1}\,.
\]
We compare {\sf RDP-SS} with the smooth sensitivity mechanism in \ifthenelse{\boolean{fullver}}{Appendix~\ref{sec:ablation-ss}}{Appendix E of the full version~\cite{full-ver}}.

\begin{algorithm}[t]
\caption{RDP Smooth Sensitivity Mechanism $\M_{\operatorname{RDP-SS}}$}
\label{algo:rdp-ss}
\begin{algorithmic}[1]
\Require Instance $\mathbf{I}$, vector query $F:\mathcal{I}\to\mathbb{R}^d$, upper bound $B$ on $\LS_{L_2}(F)$ with $|B(\mathbf{I})-B(\mathbf{I}')|\leq 1$, privacy parameters $\alpha,\rho$
\Ensure Private estimate $\tilde{F}(\mathbf{I})$
\State $t_1,t_2\gets \operatorname{brentq}(t^\alpha-\alpha e^{(\alpha-1)\rho/d}\cdot t+(\alpha-1)e^{(\alpha-1)\rho/d}=0)$
\State $\gamma\gets \frac{1}{2}\ln\min\{t_1^{-1},t_2\}$
\State $k\gets \lfloor\frac{1}{\gamma}-B(\mathbf{I})\rfloor$, $G(\mathbf{I})\gets (k+B(\mathbf{I}))\cdot e^{-\gamma k}$
\State $k\gets k+1$, $G(\mathbf{I})\gets \max\{G(\mathbf{I}), (k+B(\mathbf{I}))\cdot e^{-\gamma k}\}$
\State $\sigma\gets \frac{G(\mathbf{I})}{1-\alpha+\alpha e^{-2\gamma}}\cdot \sqrt{\frac{\alpha}{\rho}}$
\State \Return $F(\mathbf{I})+\sigma\cdot \mathcal{N}(0,\mathbf{1}_d)$
\end{algorithmic}
\end{algorithm}

\subsection{DP-S4S for Vector Queries}\label{sec:s4s-v}


Equipped with the smooth sensitivity mechanism $\M_{\operatorname{RDP-SS}}$ in Lemma~\ref{lemma:additive-smooth}, we are now ready to construct our \textsf{DP-S4S} mechanism for answering vector SJA queries under user-level DP, presented in Algorithm~\ref{algo:sample-pmsja}. Similar to Algorithm~\ref{algo:sample-r2t} for scalar queries, Algorithm~\ref{algo:sample-pmsja} consists of three stages, for computing the sample set $S$ (Lines 2-7), choosing the truncation threshold $\tau$ (Lines 9-16), and releasing the result of the truncated query (Lines 18-22), respectively. 

Specifically, in the first stage, for each scalar query $f \in F$, \textsf{DP-S4S} samples every join tuple with probability $q$, obtaining the sample set $S_f$ for $f$. The combined sample set $S$ is then the union of the samples $S_f$ for all $f$ (Line 7). The second stage then follows the PMSJA framework, which applies the SVT mechanism to choose $\tau$, solving QCQPs on the sample set $S$ in each iteration with a $\tau$ that is a power of 2, as mentioned at the beginning of Section~\ref{sec:vec}.

It remains to clarify Stage 3, which receives from Stage 2 the chosen truncation threshold $\tau$, as well as the truncated vector query $\bar{F}$ (defined in Eq.~\eqref{eq:vec-trunc}). The goal is to release $\bar{F}$ under R\'enyi DP so that the overall mechanism attains $(\varepsilon, \delta)$-DP via Lemma~\ref{lem:vec-RDP-convert}. Stage 3 first selects the R\'enyi order $\alpha$ (discussed below), and then the divergence parameter $\rho$ such that the resulting $(\alpha,\rho)$-RDP guarantee can be converted to $(\varepsilon,\delta)$-DP. It then defines the $L_2$ local-sensitivity bound $B$ for function $\bar{F}/(2\tau)$ according to Inequality~\eqref{eq:pmsja_ls_bound}, and invokes mechanism $\M_{\operatorname{RDP-SS}}$ (Algorithm~\ref{algo:rdp-ss}) to compute a noisy version of $\hat{F}(S)$ under DP, before reporting the scaled version of $\hat{F}(S)$ as the final result (Line 21).



Finally, we clarify the choice of the order $\alpha$ in R\'enyi divergence (Line 17). While our mechanism is private for any $\alpha>1+\frac{\ln(1/\delta)}{\varepsilon_2}$, the error depends on the choice of $\alpha$. 
To choose an appropriate $\alpha$, we perform a doubling search: from the lower bound of $\alpha$ (which leads to an infinite error), we repeatedly double it until we find some $\alpha$ that the data-independent error term $\sqrt{\alpha/\rho}/(1-\alpha+\alpha e^{-2\gamma})$ is smaller compared to choosing $\alpha/2$ or $2\alpha$. Then we iterate through possible choices within the interval $(\alpha/2,2\alpha)$ to locate an $\alpha$ that minimizes the term, and set $\rho=\varepsilon_2-\frac{\ln(1/\delta)}{\alpha-1}$ accordingly.

The following lemma formalizes the privacy guarantee of Algorithm~\ref{algo:sample-pmsja}, whose proof appears in \ifthenelse{\boolean{fullver}}{Appendix~\ref{sec:proof4}}{Appendix~B.4 of the full version~\cite{full-ver}}.
For this lemma, amplification only accounts for the probability that no record is sampled from the witness user.
Note that since $q'\leq 1$, we have $\varepsilon'\leq \frac{\varepsilon}{5}+\frac{4\varepsilon}{5}=\varepsilon$.

\begin{lemma}\label{lemma:s4s-v}
For any chosen $\alpha$, algorithm~\ref{algo:sample-pmsja} satisfies $(\varepsilon',\delta)$-DP for 
\[
\varepsilon'=\ln(1+q'(e^{\varepsilon/5}-1))+\frac{1}{\alpha-1}\left(\ln\left(1+q'\left(e^{4(\alpha-1)\varepsilon/5}\cdot \delta-1\right)\right)+\ln\frac{1}{\delta}\right)\,,
\]
where $q'=1-(1-q)^\Delta$.
\end{lemma}

\noindent\textbf{Discussion.} Unlike the case for scalar queries, where we derive a strong result for sampling amplification in Lemma~\ref{lemma:pure}, here the amplification effect is limited for a large $\Delta$, i.e., when each user may contribute to many join tuples. The reason is that it is unclear how the optimal solution to Program~\eqref{eq:qcqp} changes even when only $1$ tuple is sampled from the witness user. 
Note that this is the first mechanism that achieves user-level RDP with sampling of join tuples.
In particular, it remains unclear whether the original {\sf PMSJA} mechanism under $(\varepsilon,\delta)$-DP achieves the same privacy guarantee when run on sample join tuples.
We further evaluate the error of Algorithm~\ref{algo:sample-pmsja} in \ifthenelse{\boolean{fullver}}{Appendix~\ref{section:vector-error}}{Appendix F of the full version~\cite{full-ver}}.


\begin{algorithm}[t]
\caption{\textsf{DP-S4S} for vector SJA and approximate DP}\label{algo:sample-pmsja}
\begin{algorithmic}[1]
\Require Join tuples $T_1,\dots, T_d$, queries $F=\{f_1,\dots,f_d\}$, privacy parameters $\varepsilon,\delta$, sample rate $q$, probability $\beta$
\Ensure Private query results $\hat{F}(\mathbf{I})$
\State // \textit{Stage 1: compute the sample set $S$}
\State $S\gets \varnothing$
\ForEach{$f\in F$}
    \State $S_f\gets\varnothing$
    \ForEach{$t\in T_f$}
        \State Sample $t$ into $S_f$ with probability $q$
    \EndFor
    \State $S\gets S\cup\{S_f\}$
\EndFor
\State // \textit{Stage 2: choose truncation threshold $\tau$ with SVT~\cite{DBLP:journals/pacmmod/0007S023, DBLP:journals/fttcs/DworkR14}}
\State $\varepsilon_1\gets \varepsilon/5$, $\varepsilon_2\gets4\varepsilon/5$, $\theta\gets \frac{6}{\varepsilon_1}\ln\frac{6}{\beta}+\Lap(\frac{2}{\varepsilon_1})$
\ForEach{$\tau\gets 1,2,4,8,\dots$}
    \State $y_u^*, z_{f,t}^*\gets$ Optimal solution to problem~\eqref{eq:qcqp} on $S$, $\tau$ 
    \State $E(S)\gets \sum_u (1-y_u^*)$
    \If{$E(S)+\Lap(\frac{4}{\varepsilon_1})\leq\theta$}
        \State $\bar{F}(S)\gets\{\bar{f}(S):=\sum_{t\in S_f} w_{f,t}\cdot z_{f,t}\mid f\in F\}$
        \State \textbf{break}
    \EndIf
\EndFor
\State // \textit{Stage 3: release query result via Algorithm~\ref{algo:rdp-ss}}
\State $\alpha\gets\operatorname{find\_alpha}(\varepsilon_2,\delta,d)$
\State $\rho\gets \varepsilon_2-\frac{\ln(1/\delta)}{\alpha-1}$
\State Define LS upper bound $B(S):=E(S) + 1$
\State $\tilde{F}(S)\gets 2\tau\cdot \M_{\operatorname{RDP-SS}}(S, \bar{F}, B, \alpha,\rho)$
\State Report $\hat{F}(\mathbf{I}):=\frac{1}{q}\cdot \tilde{F}(S)$
\end{algorithmic}
\end{algorithm}

%% file: sections/experiments.tex
\section{Experiments} \label{sec:exp}

\newlength{\W}
\setlength{\W}{0.25\linewidth}
\newlength{\Hgt}
\setlength{\Hgt}{0.9\W}

This section presents extensive experiment evaluations comparing \textsf{DP-S4S} with the state-of-the-art algorithms on real benchmark datasets, for both relational queries and graphlet counting ones.

\subsection{Setup}

\noindent
\textbf{Mechanisms.} We consider the following mechanisms: \textsf{DP-S4S}, \textsf{R2T}~\cite{DBLP:conf/sigmod/DongF0TM22},
\textsf{PMSJA}~\cite{DBLP:journals/pacmmod/0007S023}, and \textsf{S\&E}~\cite{DBLP:journals/pacmmod/FangY24}. As mentioned in Section~\ref{sec:sca-choice}, we do not consider the variant of \textsf{S\&E} that requires a pre-known good truncation threshold $\tau$ (used in the experiments in ~\cite{DBLP:journals/pacmmod/FangY24}), which might leak private information and, thus, would be unfair to other mechanisms. However, we also observe that if we simply allocate a privacy budget in each of the $k$ independent iterations of \textsf{S\&E} for truncation threshold computation (e.g., by invoking {\sf R2T}), its overall noise becomes overwhelming that \textsf{S\&E}'s error reaches almost $100\%$ error for scalar queries across multiple datasets.
To make \textsf{S\&E} competitive, we improve the mechanism so that it samples multiple users at the same time, which avoids allocating a privacy budget in each iteration, leading to empirical performance compared to the implementation based on their original description. 
\ifthenelse{\boolean{fullver}}{
This is formally described in Corollary~\ref{cor:se} in Appendix~\ref{section:sande}}{
This is formally described in Corollary D.2 in Appendix D~\cite{full-ver}}.
In addition, we adopt an assumption favorable for \textsf{S\&E} that the number of users $m$ is public. 

Since all the mechanisms work on precomputed join results, we separate the time for computing the join using PostgreSQL from the time consumed by the DP mechanisms. 
We implement all algorithms in a single-threaded setting to enable a fair evaluation of CPU time.
Each mechanism is run 100 times independently, and we report the average error and running time, dropping the 20 largest and 20 smallest ones.


\medskip
\noindent
\textbf{Datasets.} We use the same datasets as existing work~\cite{DBLP:conf/sigmod/DongF0TM22,DBLP:journals/pacmmod/0007S023,DBLP:journals/pacmmod/FangY24} for evaluation.
For graph pattern counting queries, we use 4 real datasets:
\textit{Deezer} is a friendship network.
\textit{Amazon} is a co-purchasing graph.
Answer-to-question (\textit{A2Q}) and comment-to-answer (\textit{C2A}) are the interactions between users in the Stack Overflow network\footnote{We perform a similar preprocessing as in~\cite{DBLP:journals/pacmmod/0007S023}.}, with edges labeled by date.
All of them can be obtained from SNAP~\cite{snapnets}.
Their statistics are summarized in \ifthenelse{\boolean{fullver}}{Table~\ref{tab:graph}}{Table 4 in~\cite{full-ver}}.

For relational database SQL queries, we use the TPC-H benchmark dataset.
There are 8 relations: \texttt{Region(\underline{RK})}, \texttt{Nation(\underline{NK},RK)}, \texttt{Customer(\underline{CK},NK)}, \texttt{Orders(\underline{OK},CK)}, \texttt{Supplier(\underline{SK},NK)}, \texttt{Part(\underline{PK})}, \texttt{PartSupp(\underline{PSID},PK,SK)}, \texttt{Lineitem(\underline{LID},OK,PSID)}.
By default, we use data scale factor 4, where the largest table (\texttt{Lineitem}) contains roughly 24 million records.

For the definition of user-level DP, we regard each vertex as a user for the graph datasets.
In TPC-H, depending on the query, we may treat each customer (\texttt{CK}), supplier (\texttt{SK}), part-supp (\texttt{PID}), or order (\texttt{OK}) as a user, where a join result is contributed by a user if it is related to these primary keys through foreign-key references.


\medskip
\noindent
\textbf{Queries.} For graph datasets, we consider common graphlet counting queries, including edges $Q_{1-}$, line-2 paths $Q_{2-}$, line-3 paths $Q_{3-}$, triangles $Q_{\triangle}$ and rectangles $Q_{\square}$.
For directed graphs, we also consider the fanout-2 pattern $Q_{<}$, which differs from $Q_{2-}$.
For all these queries, an aggregation unit is contributed by $l\geq 2$ users.

We evaluate the vector query mechanisms on the TPC-H benchmark using the same queries as in~\cite{DBLP:journals/pacmmod/0007S023}.
All the queries are derived from the original $Q5$ or $Q7$ in TPC-H, with different aggregation functions \{\texttt{count(1)}, \texttt{sum(l\_extendedprice*(1-l\_discount))}\}, different group by attributes \{\texttt{n\_name}, \texttt{l\_shipdate}, \texttt{extract(year from l\_shipdate)}\}, and different definitions of private users \{(\texttt{SK},\texttt{CK}),  (\texttt{PSID},\texttt{CK}), (\texttt{PSID},\texttt{OK})\}.


\medskip
\noindent
\textbf{Parameters and environment.}
Unless otherwise specified, we consider pure DP with $\varepsilon=1$ for a scalar query, and approximate DP with $\varepsilon=4$ and $\delta=10^{-7}$ for vector queries. The default failure probability $\beta$ is set to $0.1$.
For {\sf S\&E}, the number of sampled users is chosen as $k=\lfloor q\cdot m\rfloor$, so that the expected number of records is the same as our mechanism.

Following prior works, we use data-independent degree upper bounds $D$ for graph datasets as shown in \ifthenelse{\boolean{fullver}}{Table~\ref{tab:graph}}{Table 4 in \cite{full-ver}}.
The global sensitivity is computed accordingly.
For example, $\GS_{1-}=D$, $\GS_{2-}=\GS_{\triangle}=D^2$, $\GS_{\square}=D^3$, etc.
The upper bound on the number of collaborators $C$ for {\sf S\&E} can also be computed from $D$.
To be more precise, for edge counting and triangle counting, we have $C=D$ since two collaborators must share an edge. For line-2 paths and rectangles, $C=D+D^2$ since a collaborator is either a neighbor or a distance-2 neighbor. In the experiments, we always use $C=D=1024$ for {\sf S\&E} for simplicity, which is favorable to {\sf S\&E}.

For TPC-H, none of the mechanisms requires the global sensitivity $\GS$ as input. {\sf S\&E} still requires setting $C$. Following their paper~\cite{DBLP:journals/pacmmod/FangY24}, we still set $C=1024$ (again in its favor), though this can be insufficient, especially for vector queries. 

\medskip

All experiments were conducted on a Linux server equipped with an Intel(R) Xeon(R) Gold 6252N 2.30GHz CPU with 96 cores and 1.5TB of memory.
Each experiment was run in a separate Docker container. For scalar queries, we allocated 32GB of memory; for vector queries, we allocated 256GB of memory. The experiment code can be found at \url{https://github.com/QiuTedYuan/DP-S4S/}.

\subsection{Comparison on Scalar Queries}\label{sec:exp-scalar}

\input{fig/deezer}

\input{fig/amazon2}

Figures~\ref{fig:deezer-q} and~\ref{fig:amazon2-q} compare the empirical performance of {\sf R2T}, {\sf S\&E}, the proposed {\sf DP-S4S} over graphlet counting queries, with varying sample rate from $1/64$ to $1/4$ for the latter two.
Horizontal axes correspond to execution times, while vertical axes correspond to error (in percentage) relative to the actual query result, both in log scale.
Each data point in the plot corresponds to the result for a given sample rate; a high sample rate leads to longer execution time.
In the figure captions, {\it base query time} refers to the join computation time measured by PostgreSQL.

Comparing {\sf DP-S4S} with {\sf R2T}, we observe that {\sf DP-S4S} achieves a similar error on almost every dataset at moderate sample rates.
As analyzed in Section~\ref{sec:sca-analysis}, there are two sources of error in {\sf DP-S4S}: sampling error and DP noise. While it is typical for sampling error to increase when lowering the sample rate, the amplification through Lemma~\ref{algo:sample-truncate} also becomes more effective, leading to a larger $\varepsilon_i^*$ and, correspondingly, lower noise to satisfy DP.
The overall error depends on which is dominating.
Recall that a lower sample rate also reduces the maximum contribution of a user to the sample, which offsets the increase in the scaling-back factor.
Therefore, despite having an additional sampling error, it is still possible for {\sf DP-S4S} to be slightly more accurate than {\sf R2T} under certain sample rates thanks to the privacy amplification.
The optimal sample rate, however, depends on both the data size (which affects the sampling error) and the downward sensitivity of the original query (which impacts the DP error).
Under differential privacy, we cannot use the exact best value for $q$ that balances both.
On the other hand, the experimental results show that the accuracy of {\sf DP-S4S} is stable across various datasets and a range of sample rates, suggesting that $q$ can be chosen in a data-inpdenent manner.
Notably, on the largest set of join tuples (Deezer $Q_{2-}$), {\sf DP-S4S} reduces the running time of {\sf R2T} from 1855 seconds (or 30 minutes) to 36.7 seconds at sample rate $q=1/64$, with only a mild error increase from $6.30\%$ to $6.72\%$.

Compared with {\sf S\&E}, {\sf DP-S4S} is consistently more accurate given the same execution time. Their accuracy gap is larger for small $q$, and is more than 10x for $q=1/64$. This is because {\sf DP-S4S} benefits from having a better sample amplification when $q$ is small, whereas {\sf S\&E} has no amplification effect unless $q\leq 1/(C+1)\approx10^{-3}$ as per \ifthenelse{\boolean{fullver}}{Corollary~\ref{cor:se}}{Corollary D.2 in~\cite{full-ver}}.
Also, as shown in Section~\ref{sec:sca-analysis}, the error of {\sf S\&E} grows with the sample rate (and, thus, the number of iterations $k$) under pure DP.
Even when relaxed to approximate DP, the overhead due to DP noise, $C\sqrt{\log(1/\delta)}$, causes its error to be more than 6000x larger than in the non-sampling setting.

We also observe that {\sf S\&E} has a slightly higher running time than our mechanism when given the same sample rate. This is mainly because of the variance in its sample size. While both mechanisms have the same expected sample size, {\sf DS-S4S} uses Poisson sampling, whose size is stable. For {\sf S\&E}, the sampled tuples are correlated since all the records from the sampled users are taken. When high-degree nodes are sampled, more join tuples are produced, leading to a higher running time.

\subsection{Comparison on Vector Queries}

\noindent
\textbf{Relational SJA queries.}
We adopt the same query set as~\cite{DBLP:journals/pacmmod/0007S023}: $Q5$ aggregates the revenue grouped by nation names, where each customer or supplier is a private user.
In the three variants of $Q7$, both $7a$ and $7b$ compute group record counts where each customer or partsupp is a private user: $7a$ groups the join results by nation names, while $7b$ uses the year of the shipdate as an additional group-by attribute.
$7c$ computes the revenue grouped by shipdate, where each user is a partsupp or an order.

Figure~\ref{fig:tpch} presents the results on a TPC-H benchmark. 
Clearly {\sf PMSJA} consumes a lot of memory and CPU time on such datasets. In particular, {\sf PMSJA} took more than $10,000$ seconds (about 2.7 hours) to produce an output for $Q5$. 
We have also experimented with $Q7$ by considering customers and suppliers as private users, and {\sf PMSJA} ran out of memory even after we increased its RAM limit to 512GB.
Note that the CPU time and memory consumption of a QCQP solver are super-linear: a typical solver such as MOSEK~\cite{mosek} requires over $O(n^3)$ time and $O(n^2)$ space. Accordingly, doubling the data size may cause the mechanism to be slower by more than 8x, and may increase the memory usage by 4x.

Following the same logic, sampling provides a significant efficiency boost for large datasets. For $Q5$, even with a relatively high sample rate of $1/4$, we can reduce the running time to $38.11s$, which speeds up {\sf PMSJA} by more than 260x.
The error, on the other hand, increases from $13.9\%$ to $23.5\%$. This is mainly due to the fact that the amplification effect in Lemma~\ref{lemma:s4s-v} is small, meaning that the theoretical result of \textsf{DP-S4S} only guarantees the sampled instance is at least as private with respect to the original instance, but it does not allow using a significantly larger $\rho$ (in R\'enyi DP) or $\varepsilon$ (in ($\varepsilon$, $\delta$)-DP) when running the mechanism on sub-sampled data, and reducing the sample rate will increase the scale-up factor $1/q$. It is an interesting future direction to investigate better amplification bounds for combining sampling with the {\sf PMSJA} mechanism.
Nevertheless, \textsf{DP-S4S} helps reduce the noise level to achieve consistently and significantly improved efficiency-utility trade-offs compared to {\sf S\&E}, whose error exceeds $100\%$ even at a sample rate of $1/4$ on $Q5$.

\input{fig/tpch}

\input{fig/c2a}


\medskip
\noindent
\textbf{Graph queries.}
In addition to the queries in~\cite{DBLP:journals/pacmmod/0007S023}, we include another query that counts the fanout-2 pattern $a\to b \wedge a\to c$, which is different from a length-2 path on a directed graph. The results are presented in Figure 
\ref{fig:c2a}.
Similar to the performance on TPC-H, we observe that {\sf DP-S4S} is consistently more accurate than {\sf S\&E} by an order of magnitude. The actual improvement depends on both the sample rate and the characteristics of the data.

For the fanout-2 query on C2A, we have to allocate {\sf PMSJA} with 512GB memory since under the default setup, the QCQP solver MOSEK runs out of memory. In addition to high memory consumption, {\sf PMSJA} is also slow: it took more than 4 hours to finish a single iteration. Our sampling mechanism reduces the running time to within a minute at a sample rate around $10\%$, without severely affecting the error, whereas {\sf S\&E} offers the same level of efficiency improvement at the cost of rather low result utility.

\subsection{Parameter Study}\label{sec:exp:param_study}

\noindent
\textbf{Varying privacy parameter $\varepsilon$.}
Figure~\ref{fig:amazon2-eps} presents the effect of varying the privacy parameter $\varepsilon$ on all the mechanisms over scalar queries.
Observe that even under a relatively small sample rate of $1\%$, {\sf DP-S4S} still matches the error of {\sf R2T}, whereas {\sf S\&E} has an overhead that can reach 14x. In all the scenarios other than {\sf DP-S4S} on $Q_{1-}$, DP noise dominates sampling error; thus, the accuracy of all mechanisms improves with increasing $\varepsilon$. {\sf DP-S4S} clearly achieves a better privacy-utility trade-off compared to {\sf S\&E} in all settings.
The privacy amplification effect may also help {\sf DP-S4S} to outperform {\sf R2T} when DP noise dominates (as for $Q_{2-}$ and $Q_{\square}$), in addition to the benefit of reducing the query time by over 100x. 

Figure~\ref{fig:c2a-eps} shows the results for scalar queries. Since the amplification from Lemma~\ref{lemma:s4s-v} is limited, we run the sampling mechanisms with $q=10\%$ to avoid large scale-back factors.
We observe a 2x-3x error increase for {\sf DP-S4S}, compared to a more than 15x increase in error for {\sf S\&E}.
Nevertheless, since QCQPs are far more expensive to solve than LPs, even with a sample rate of $10\%$, we achieve 2-3 orders of magnitude improvement in running time, according to Figure~\ref{fig:c2a}.
This leads to a favorable privacy-accuracy trade-off, and {\sf DP-S4S} consistently outperforms {\sf S\&E}.

\input{fig/amazon2-eps}

\input{fig/c2a-eps}

\input{fig/amazon2-d}


\medskip
\noindent
\textbf{Degree upper bound $D$ (equivalently $C$ and $\Delta$)}. All the mechanisms considered above achieve instance-specific error, which is practical for queries with a large global sensitivity. For graph queries, there are two reasons that account for a large $\GS$: on the one hand, the degree upper bound $D$ should be a data-independent worst-case upper bound. Under-estimating $D$ can lead to DP violations. In practice, this can be obtained from domain knowledge, e.g., Facebook limits each user to have up to 5,000 friends. On the other hand, for queries involving joins, we need to adopt worst-case join size bounds~\cite{DBLP:conf/focs/AtseriasGM08}, whereas in practice the join size can be much smaller. The same logic applies to the number of collaborators bound $C$ in {\sf S\&E}, and the number of tuples bound $\Delta$ in our mechanism. This has been studied in existing work~\cite{DBLP:conf/sigmod/DongF0TM22,DBLP:journals/pacmmod/0007S023}. 

Figure~\ref{fig:amazon2-d} shows the results with varying $D$ and correspondingly $C$ and $\Delta$.
For all queries, we vary the degree bound $D$ from 512 to 4096. Note that the maximum degree in the dataset is $549>512$; therefore, setting $D=512$ actually violates the privacy constraint.
As expected,
$\GS$ mechanism scales linearly with $\Delta$, which is $D$ for $Q_{1-}$ and $\Omega(D^2)$ for the other queries.
Instance-specific mechanisms like {\sf R2T} is better when $\GS$ is large, as the dependency is only $O(\log D)$. For sampling mechanisms, besides the overhead introduced by {\sf R2T}, a larger $D$ also leads to a larger $C$ for {\sf S\&E} and a larger $\Delta$ for {\sf DP-S4S}, weakening the privacy amplification effects in both.
Accordingly, the result utility of {\sf DP-S4S} is slightly worse than {\sf R2T}, and consistently better than {\sf S\&E}.

For vector queries, both {\sf PMSJA} and {\sf DP-S4S} show performance independent of the choice of $D$. Note that $D$ (or $\Delta$) is not an input to Algorithm~\ref{algo:sample-pmsja}, as do not apply privacy amplification in this case. 

%% file: fig/deezer.tex
\begin{figure}[t]
\centering

\begin{tikzpicture}
  \node{
    \begin{tabular}{@{}c@{\quad}c@{\qquad}c@{\quad}c@{\qquad}c@{\quad}c@{}}
    \tikz{\draw[thick, draw=r2tCol]     plot[mark=asterisk, mark size=3pt] (0,0);} & R2T &
    \tikz{\draw[thick, draw=spCol]      plot[mark=triangle*, mark size=3pt, mark options={fill=white, line width=1pt}] (0,0);} & S\&E &
    \tikz{\draw[thick, draw=dpsCol]     plot[mark=o, mark size=3pt] (0,0);} & DP-S4S
    \end{tabular}
  };
\end{tikzpicture}

\makebox[\textwidth][c]{%
   \resizebox{\linewidth}{!}{%
\begin{tabular}{@{}c@{}c@{}c@{}c@{}}

\hspace{-2mm}
\subcaptionbox{$\boldsymbol{Q}_{1-}$ (base query time: 0.18s).}[ \W ]{%
\begin{tikzpicture}
\begin{axis}[
  width=\W, height=\Hgt,
  xmode=log, ymode=log,
  log ticks with fixed point, scaled ticks=false,
  minor tick style={draw=none},
  grid=major, major tick length=2pt,
  minor grid style={draw=none},  
  xtick={0.1,1,10,100}, xticklabels={0.1,1,10,100},
  ytick={0.1,1,10,100}, yticklabels={0.1,1,10,100},
  tick align=outside, tick style={black},
  tick label style={font=\scriptsize}, label style={font=\scriptsize},
  xlabel={time (s)}, ylabel={relative error (\%)},
  xmin=0.1, xmax=100, ymin=0.1, ymax=100,
  ylabel style={yshift=-4mm}
]
  \addplot[only marks, mark=asterisk, mark size=3pt, draw=r2tCol, thick] coordinates {
(35.00, 0.5505)}; 
  \addplot[mark=triangle, mark size=3pt, mark options={fill=white, line width=1pt}, draw=spCol, thick] coordinates { 
(12.68,1.615) (6.313, 3.1618) (3.16, 5.73) (1.94, 7.38) (0.807,9.72)
 }; 
  \addplot[mark=o, mark size=3pt, draw=dpsCol, thick] coordinates {
(10.69,0.789) (5.713, 0.447) (2.850, 0.722) (1.537, 0.777) (0.642,0.716)
}; 
\end{axis}
\end{tikzpicture}
} & \hspace{-5mm}
\subcaptionbox{$\boldsymbol{Q}_{2-}$ (base query time: 57.4s).}[ \W ]{%
\begin{tikzpicture}
\begin{axis}[
  width=\W, height=\Hgt,
  xmode=log, ymode=log,
  log ticks with fixed point, scaled ticks=false,
  minor tick style={draw=none},
  grid=major, major tick length=2pt,
  minor grid style={draw=none}, 
  xtick={1,10,100,1000,10000}, xticklabels={1,10,100,1000,10000},
  ytick={0.1,1,10,100}, yticklabels={0.1,1,10,100},
  tick align=outside, tick style={black},
  tick label style={font=\scriptsize}, label style={font=\scriptsize},
  xlabel={time (s)}, ylabel={},
  xmin=10, xmax=10000, ymin=0.1, ymax=100
]
  \addplot[only marks, mark=asterisk, mark size=3pt, thick, draw=r2tCol] coordinates {
  (1855.77, 6.30)
}; 
  \addplot[mark=triangle, mark size=3pt, mark options={fill=white, line width=1pt}, thick, draw=spCol] coordinates {
(548.31, 6.5427) (257.32, 8.714) (128.12, 16.70) (59.57,22.25) (32.73,27.52)
}; 
  \addplot[mark=o, mark size=3pt, thick, draw=dpsCol] coordinates {
(536.65,5.509) (271.69, 7.255) (134.84, 5.76) (77.83, 5.75) (36.70, 6.72)
  }; 
\end{axis}
\end{tikzpicture}
} & \hspace{-6mm}
\subcaptionbox{$\boldsymbol{Q}_{\triangle}$ (base query time: 8.87s).}[ \W ]{%
\begin{tikzpicture}
\begin{axis}[
  width=\W, height=\Hgt,
  xmode=log, ymode=log,
  log ticks with fixed point, scaled ticks=false,
  minor tick style={draw=none},
  grid=major, major tick length=2pt,
  minor grid style={draw=none}, 
  xtick={0.1,1,10,100}, xticklabels={0.1,1,10,100},
  ytick={0.1,1,10,100}, yticklabels={0.1,1,10,100},
  tick align=outside, tick style={black},
  tick label style={font=\scriptsize}, label style={font=\scriptsize},
  xlabel={time (s)}, ylabel={},
  xmin=0.1, xmax=100, ymin=0.1, ymax=100
]
  \addplot[only marks, mark=asterisk, mark size=3pt, thick, draw=r2tCol] coordinates {(50.82, 5.1124)}; 
  \addplot[mark=triangle, mark size=3pt, mark options={fill=white, line width=1pt}, thick, draw=spCol] coordinates {
  (16.99,10.874) (8.568,13.318) (4.401, 20.0604) (3.190, 33.660) (1.456,48.284)
}; 
  \addplot[mark=o, mark size=3pt, thick, draw=dpsCol] coordinates {
  (15.67,4.45) (8.316, 5.950) (3.6115, 5.290) (2.592, 4.031) (0.823, 6.546)
  }; 
\end{axis}
\end{tikzpicture}
} & \hspace{-6mm}
\subcaptionbox{$\boldsymbol{Q}_{\square}$ (base query time: 296s).}[ \W ]{%
\begin{tikzpicture}
\begin{axis}[
  width=\W, height=\Hgt,
  xmode=log, ymode=log,
  log ticks with fixed point, scaled ticks=false,
  minor tick style={draw=none},
  grid=major, major tick length=2pt,
  minor grid style={draw=none}, 
  xtick={0.1,1,10,100,1000,10000}, xticklabels={0.1,1,10,100,1000,10000},
  ytick={0.1,1,10,100,1000}, yticklabels={0.1,1,10,100,1000},
  tick align=outside, tick style={black},
  tick label style={font=\scriptsize}, label style={font=\scriptsize},
  xlabel={time (s)}, ylabel={},
  xmin=10, xmax=10000, ymin=0.1, ymax=100
]
  \addplot[only marks, mark=asterisk, mark size=3pt, thick, draw=r2tCol] coordinates {(944.23,16.2)}; 
  \addplot[mark=triangle, mark size=3pt, mark options={fill=white, line width=1pt}, thick, draw=spCol] coordinates {
  (287.36,22.49) (140.50,32.27) (67.04,50.59) (37.06, 60.36) (20.22, 75.626)
}; 
  \addplot[mark=o, mark size=3pt, thick, draw=dpsCol] coordinates {
(269.43, 12.79) (140.32, 17.59) (63.48, 11.78) (45.57,13.77) (18.89, 13.51)
}; 
\end{axis}
\end{tikzpicture}
} 
\end{tabular}
}}
\vspace{-3mm}
\caption{Query error vs.\ time on Deezer with $\varepsilon = 1$. (Sampling rates $q\in\left\{ \frac{1}{64},\tfrac{1}{32}, \tfrac{1}{16}, \tfrac{1}{8}, \tfrac{1}{4} \right\}$)}
\label{fig:deezer-q}
\end{figure}

%% file: fig/amazon2.tex
\begin{figure}[t]
\centering

\makebox[\textwidth][c]{%
   \resizebox{\linewidth}{!}{%
\begin{tabular}{@{}c@{}c@{}c@{}c@{}}

\hspace{-2mm}
\subcaptionbox{$\boldsymbol{Q}_{1-}$ (base query time: 0.21s).}[ \W ]{%
\begin{tikzpicture}
\begin{axis}[
  width=\W, height=\Hgt,
  xmode=log, ymode=log,
  log ticks with fixed point, scaled ticks=false,
  minor tick style={draw=none},
  grid=major, major tick length=2pt,
  minor grid style={draw=none},  
  xtick={0.1,1,10,100}, xticklabels={0.1,1,10,100},
  ytick={0.1,1,10,100}, yticklabels={0.1,1,10,100},
  tick align=outside, tick style={black},
  tick label style={font=\scriptsize}, label style={font=\scriptsize},
  xlabel={time (s)}, ylabel={relative error (\%)},
  xmin=0.1, xmax=100, ymin=0.1, ymax=100,
  ylabel style={yshift=-4mm}
]
  \addplot[only marks, mark=asterisk, mark size=3pt, thick, draw=r2tCol] coordinates {
(39.67,0.4282)}; 
  \addplot[mark=triangle, mark size=3pt, mark options={fill=white, line width=1pt}, thick, draw=spCol] coordinates { 
(12.815,1.153) (6.603, 1.712) (2.946,2.233) (1.557, 3.506) (0.717,5.696)
 }; 
  \addplot[mark=o, mark size=3pt, thick, draw=dpsCol] coordinates {
(11.46,0.475) (5.924,0.446) (2.932,0.721) (1.144,0.6419) (0.341,0.958)
}; 
\end{axis}
\end{tikzpicture}
} & \hspace{-5mm}
\subcaptionbox{$\boldsymbol{Q}_{2-}$ (base query time: 35.1s).}[ \W ]{%
\begin{tikzpicture}
\begin{axis}[
  width=\W, height=\Hgt,
  xmode=log, ymode=log,
  log ticks with fixed point, scaled ticks=false,
  minor tick style={draw=none},
  grid=major, major tick length=2pt,
  minor grid style={draw=none}, 
  xtick={1,10,100,1000}, xticklabels={1,10,100,1000},
  ytick={0.1,1,10,100}, yticklabels={0.1,1,10,100},
  tick align=outside, tick style={black},
  tick label style={font=\scriptsize}, label style={font=\scriptsize},
  xlabel={time (s)}, ylabel={},
  xmin=1, xmax=1000, ymin=0.1, ymax=100
]
  \addplot[only marks, mark=asterisk, mark size=3pt, thick, draw=r2tCol] coordinates {
(782.00,9.113)}; 
  \addplot[mark=triangle, mark size=3pt, mark options={fill=white, line width=1pt}, thick, draw=spCol] coordinates {
(241.67,9.879) (107.19,8.760)  (53.88,9.36) (24.93, 12.14) (12.20,13.06)
}; 
  \addplot[mark=o, mark size=3pt, thick, draw=dpsCol] coordinates {
(243.69, 8.714)  (113.87, 7.758) (57.58,7.69) (26.964,8.393) (13.702,6.874)
 }; 
\end{axis}
\end{tikzpicture}
} & \hspace{-6mm}
\subcaptionbox{$\boldsymbol{Q}_{\triangle}$ (base query time: 7.27s).}[ \W ]{%
\begin{tikzpicture}
\begin{axis}[
  width=\W, height=\Hgt,
  xmode=log, ymode=log,
  log ticks with fixed point, scaled ticks=false,
  minor tick style={draw=none},
  grid=major, major tick length=2pt,
  minor grid style={draw=none}, 
  xtick={0.1,1,10,100}, xticklabels={0.1,1,10,100},
  ytick={0.1,1,10,100}, yticklabels={0.1,1,10,100},
  tick align=outside, tick style={black},
  tick label style={font=\scriptsize}, label style={font=\scriptsize},
  xlabel={time (s)}, ylabel={},
  xmin=0.1, xmax=100, ymin=0.1, ymax=100
]
  \addplot[only marks, mark=asterisk, mark size=3pt, thick, draw=r2tCol] coordinates {(40.41, 1.796)}; 
  \addplot[mark=triangle, mark size=3pt, mark options={fill=white, line width=1pt}, thick, draw=spCol] coordinates {
(13.35,3.904) (6.419,5.668) (3.073,6.868) (1.582,10.491) (0.640,16.63)
}; 
  \addplot[mark=o, mark size=3pt, thick, draw=dpsCol] coordinates {
(11.015,2.104) (5.467,1.564) (2.579,2.239) (1.141,2.810) (0.295,2.860)
  }; 
\end{axis}
\end{tikzpicture}
} & \hspace{-6mm}
\subcaptionbox{$\boldsymbol{Q}_{\square}$ (base query time: 36.2s).}[ \W ]{%
\begin{tikzpicture}
\begin{axis}[
  width=\W, height=\Hgt,
  xmode=log, ymode=log,
  log ticks with fixed point, scaled ticks=false,
  minor tick style={draw=none},
  grid=major, major tick length=2pt,
  minor grid style={draw=none}, 
  xtick={0.1,1,10,100,1000}, xticklabels={0.1,1,10,100,1000},
  ytick={0.1,1,10,100,1000}, yticklabels={0.1,1,10,100,1000},
  tick align=outside, tick style={black},
  tick label style={font=\scriptsize}, label style={font=\scriptsize},
  xlabel={time (s)}, ylabel={},
  xmin=1, xmax=1000, ymin=0.1, ymax=100
]
  \addplot[only marks, mark=asterisk, mark size=3pt, thick, draw=r2tCol] coordinates {(211.94, 10.213)}; 
  \addplot[mark=triangle, mark size=3pt, mark options={fill=white, line width=1pt}, thick, draw=spCol] coordinates {
(68.574,11.013)(34.93,16.34)(16.96,21.17)(8.415,25.100) (4.422,33.569)
}; 
  \addplot[mark=o, mark size=3pt, thick, draw=dpsCol] coordinates {
(63.15,8.26) (33.04,8.54) (16.08, 8.12) (8.16,9.71) (3.47,7.76)
}; 
\end{axis}
\end{tikzpicture}
} 
\end{tabular}
}}
\vspace{-3mm}
\caption{Query error vs.\ time on Amazon with $\varepsilon = 1$. (Sampling rates $q\in\left\{ \frac{1}{64},\tfrac{1}{32}, \tfrac{1}{16}, \tfrac{1}{8}, \tfrac{1}{4} \right\}$)}
\label{fig:amazon2-q}
\end{figure}

%% file: fig/tpch.tex
\begin{figure}[t]
\centering

\begin{tikzpicture}
  \node{
    \begin{tabular}{@{}c@{\quad}c@{\qquad}c@{\quad}c@{\qquad}c@{\quad}c@{}}
    \tikz{\draw[thick, draw=pmsjaCol]      plot[mark=pentagon, mark size=3pt, mark options={fill=white, line width=1pt}] (0,0);} & PMSJA &
    \tikz{\draw[thick, draw=spCol]      plot[mark=triangle*, mark size=3pt, mark options={fill=white, line width=1pt}] (0,0);} & S\&E &
    \tikz{\draw[thick, draw=dpsCol]     plot[mark=o, mark size=3pt] (0,0);} & DP-S4S
    \end{tabular}
  };
\end{tikzpicture}

\makebox[\textwidth][c]{%
   \resizebox{\linewidth}{!}{%
\begin{tabular}{@{}c@{}c@{}c@{}c@{}}

\hspace{-2mm}
\subcaptionbox{$\boldsymbol{Q}_{5}$ (base query time: 14.2s).}[ \W ]{%
\begin{tikzpicture}
\begin{axis}[
  width=\W, height=\Hgt,
  xmode=log, ymode=log,
  log ticks with fixed point, scaled ticks=false,
  minor tick style={draw=none},
  grid=major, major tick length=2pt,
  minor grid style={draw=none},  
  xtick={0.1,1,10,100,1000,10000,100000}, xticklabels={10$^{-1}$,1,10,10$^2$,10$^3$,10$^4$,10$^5$},
  xticklabel style={font=\scriptsize, text height=1.5ex, text depth=.25ex}, 
  ytick={0.1,1,10,100,1000,10000}, 
  yticklabels={10$^{-1}$,1,10,10$^2$,10$^3$,10$^4$},
  yticklabel style={font=\scriptsize, text width=1.5em, align=right, text height=1.5ex, text depth=.25ex},   
  tick align=outside, tick style={black},
  label style={font=\scriptsize},
  xlabel={time (s)}, ylabel={relative error (\%)},
  xmin=1, xmax=100000, ymin=1, ymax=10000,
  ylabel style={yshift=-4mm}
]
  \addplot[only marks, mark=pentagon, mark size=3pt, thick, draw=pmsjaCol] coordinates {
(10145.24, 13.90)
}; 
  \addplot[mark=triangle, mark size=3pt, mark options={fill=white, line width=1pt}, thick, draw=spCol] coordinates { 
(3.51,1849.81) (6.69,931.53) (15.46,469.70) (59.97,241.40)
 }; 
  \addplot[mark=o, mark size=3pt, thick, draw=dpsCol] coordinates {
(4.42,77.66) (7.73, 41.80) (16.65,34.50) (38.11,23.52)
}; 
\end{axis}
\end{tikzpicture}
} & \hspace{-5mm}
\subcaptionbox{$\boldsymbol{Q}_{7a}$ (base query time: 6.46s).}[ \W ]{%
\begin{tikzpicture}
\begin{axis}[
  width=\W, height=\Hgt,
  xmode=log, ymode=log,
  log ticks with fixed point, scaled ticks=false,
  minor tick style={draw=none},
  grid=major, major tick length=2pt,
  minor grid style={draw=none}, 
  xtick={0.1,1,10,100,1000,10000,100000}, xticklabels={10$^{-1}$,1,10,10$^2$,10$^3$,10$^4$,10$^5$},
  xticklabel style={font=\scriptsize, text height=1.5ex, text depth=.25ex}, 
  ytick={0.1,1,10,100,1000,10000}, 
  yticklabels={10$^{-1}$,1,10,10$^2$,10$^3$,10$^4$},
  yticklabel style={font=\scriptsize, text width=1.5em, align=right, text height=1.5ex, text depth=.25ex},
  tick align=outside, tick style={black},
  label style={font=\scriptsize},
  xlabel={time (s)}, ylabel={},
  xmin=0.1, xmax=100, ymin=1, ymax=1000
]
  \addplot[only marks, mark=pentagon, mark size=3pt, thick, draw=pmsjaCol] coordinates {
  (79.42, 8.276)
}; 
  \addplot[mark=triangle, mark size=3pt, mark options={fill=white, line width=1pt}, thick, draw=spCol] coordinates {
(0.969, 170.77) (2.916,108.12) (8.021, 67.20) (16.69, 57.18)
}; 
  \addplot[mark=o, mark size=3pt, thick, draw=dpsCol] coordinates {
(2.75,47.96) (3.29, 32.87) (5.81, 16.60) (11.21, 16.45)
  }; 
\end{axis}
\end{tikzpicture}
} & \hspace{-6mm}
\subcaptionbox{$\boldsymbol{Q}_{7b}$ (base query time: 14.2s).}[ \W ]{%
\begin{tikzpicture}
\begin{axis}[
  width=\W, height=\Hgt,
  xmode=log, ymode=log,
  log ticks with fixed point, scaled ticks=false,
  minor tick style={draw=none},
  grid=major, major tick length=2pt,
  minor grid style={draw=none}, 
  xtick={0.1,1,10,100,1000,10000,100000}, xticklabels={10$^{-1}$,1,10,10$^2$,10$^3$,10$^4$,10$^5$},
  xticklabel style={font=\scriptsize, text height=1.5ex, text depth=.25ex}, 
  ytick={0.1,1,10,100,1000,10000}, 
  yticklabels={10$^{-1}$,1,10,10$^2$,10$^3$,10$^4$},
  yticklabel style={font=\scriptsize, text width=1.5em, align=right, text height=1.5ex, text depth=.25ex},
  tick align=outside, tick style={black},
  label style={font=\scriptsize},
  xlabel={time (s)}, ylabel={},
  xmin=0.1, xmax=100, ymin=1, ymax=1000
]
  \addplot[only marks, mark=pentagon, mark size=3pt, thick, draw=pmsjaCol] coordinates {
(51.71,12.36)
}; 
  \addplot[mark=triangle, mark size=3pt, mark options={fill=white, line width=1pt}, thick, draw=spCol] coordinates {
(1.04,236.08) (2.70,196.77) (9.54, 104.93) (18.87, 94.12)
}; 
  \addplot[mark=o, mark size=3pt, thick, draw=dpsCol] coordinates {
(2.60,65.92) (3.86, 61.22) (6.55, 28.83) (13.31, 16.30)
  }; 
\end{axis}
\end{tikzpicture}
} & \hspace{-6mm}
\subcaptionbox{$\boldsymbol{Q}_{7c}$ (base query time: 10.3s).}[ \W ]{%
\begin{tikzpicture}
\begin{axis}[
  width=\W, height=\Hgt,
  xmode=log, ymode=log,
  log ticks with fixed point, scaled ticks=false,
  minor tick style={draw=none},
  grid=major,  major tick length=2pt,
  minor grid style={draw=none}, 
  xtick={0.1,1,10,100,1000,10000,100000}, xticklabels={10$^{-1}$,1,10,10$^2$,10$^3$,10$^4$,10$^5$},
  xticklabel style={font=\scriptsize, text height=1.5ex, text depth=.25ex}, 
  ytick={0.1,1,10,100,1000,10000}, 
  yticklabels={10$^{-1}$,1,10,10$^2$,10$^3$,10$^4$},
  yticklabel style={font=\scriptsize, text width=1.5em, align=right, text height=1.5ex, text depth=.25ex},
  tick align=outside, tick style={black},
  label style={font=\scriptsize},
  xlabel={time (s)}, ylabel={},
  xmin=1, xmax=1000, ymin=1, ymax=1000
]
  \addplot[only marks, mark=pentagon, mark size=3pt, thick, draw=pmsjaCol] coordinates {
(143.36, 7.112)
}; 
  \addplot[mark=triangle, mark size=3pt, mark options={fill=white, line width=1pt}, thick, draw=spCol] coordinates {
(3.74,460.35) (8.88,227.49) (21.46, 127.54) (81.95, 88.55)
}; 
  \addplot[mark=o, mark size=3pt, thick, draw=dpsCol] coordinates {
(3.98,136.59) (6.22,67.95) (12.30, 34.03) (28.12, 17.25)
}; 
\end{axis}
\end{tikzpicture}
} 
\end{tabular}
}}
\vspace{-3mm}
\caption{Query error vs.\ time on TPC-H SF4 with $\varepsilon = 4$, $\delta=10^{-7}$. (Sampling rates $q\in\left\{ \frac{1}{32},\tfrac{1}{16}, \tfrac{1}{8}, \tfrac{1}{4}\right\}$)}
\label{fig:tpch}
\vspace{-2mm}
\end{figure}

%% file: fig/c2a.tex
\begin{figure}[t]
\centering

\makebox[\textwidth][c]{%
   \resizebox{\linewidth}{!}{%
\begin{tabular}{@{}c@{}c@{}c@{}c@{}}
\hspace{-2mm}
\subcaptionbox{$\boldsymbol{Q}_{1-}$ (base query time: 1.46s).}[ \W ]{%
\begin{tikzpicture}
\begin{axis}[
  width=\W, height=\Hgt,
  xmode=log, ymode=log,
  log ticks with fixed point, scaled ticks=false,
  minor tick style={draw=none},
  grid=major, major tick length=2pt,
  minor grid style={draw=none},  
  xtick={0.1,1,10,100,1000,10000,100000}, xticklabels={10$^{-1}$,1,10,10$^2$,10$^3$,10$^4$,10$^5$},
  xticklabel style={font=\scriptsize, text height=1.5ex, text depth=.25ex}, 
  ytick={0.1,1,10,100,1000,10000}, 
  yticklabels={10$^{-1}$,1,10,10$^2$,10$^3$,10$^4$},
  yticklabel style={font=\scriptsize, text width=1.5em, align=right, text height=1.5ex, text depth=.25ex},   
  tick align=outside, tick style={black},
  label style={font=\scriptsize},
  xlabel={time (s)}, ylabel={relative error (\%)},
  xmin=0.1, xmax=100, ymin=1, ymax=1000,
  ylabel style={yshift=-4mm}
]
  \addplot[only marks, mark=pentagon, mark size=3pt, thick, draw=pmsjaCol] coordinates {
(13.28,17.07)
}; 
  \addplot[mark=triangle, mark size=3pt, mark options={fill=white, line width=1pt}, thick, draw=spCol] coordinates { 
(0.3151,617.65) (0.6120,355.57) (1.298,267.83) (2.677,136.89)
 }; 
  \addplot[mark=o, mark size=3pt, thick, draw=dpsCol] coordinates {
  (0.483,87.81) (0.6737,46.55) (0.957, 42.23) (1.897,23.00)
}; 
\end{axis}
\end{tikzpicture}
} & \hspace{-5mm}
\subcaptionbox{$\boldsymbol{Q}_{2-}$ (base query time: 1.65s).}[ \W ]{%
\begin{tikzpicture}
\begin{axis}[
  width=\W, height=\Hgt,
  xmode=log, ymode=log,
  log ticks with fixed point, scaled ticks=false,
  minor tick style={draw=none},
  grid=major, major tick length=2pt,
  minor grid style={draw=none}, 
  xtick={0.1,1,10,100,1000,10000,100000}, xticklabels={10$^{-1}$,1,10,10$^2$,10$^3$,10$^4$,10$^5$},
  xticklabel style={font=\scriptsize, text height=1.5ex, text depth=.25ex}, 
  ytick={0.1,1,10,100,1000,10000}, 
  yticklabels={10$^{-1}$,1,10,10$^2$,10$^3$,10$^4$},
  yticklabel style={font=\scriptsize, text width=1.5em, align=right, text height=1.5ex, text depth=.25ex},
  tick align=outside, tick style={black},
  label style={font=\scriptsize},
  xlabel={time (s)}, ylabel={},
  xmin=0.1, xmax=100, ymin=1, ymax=10000
]
  \addplot[only marks, mark=pentagon, mark size=3pt, thick, draw=pmsjaCol] coordinates {
(53.17,28.67)
}; 
  \addplot[mark=triangle, mark size=3pt, mark options={fill=white, line width=1pt}, thick, draw=spCol] coordinates {
  (1.007,1233) (2.549,691.53) (5.334,380.1) (11.24,247.72)
}; 
  \addplot[mark=o, mark size=3pt, thick, draw=dpsCol] coordinates {
(1.070,111.204) (1.658,61.51) (3.314, 71.319) (8.896,59.30)
  }; 
\end{axis}
\end{tikzpicture}
} & \hspace{-6mm}
\subcaptionbox{$\boldsymbol{Q}_{3-}$ (base query time: 1.91s).}[ \W ]{%
\begin{tikzpicture}
\begin{axis}[
  width=\W, height=\Hgt,
  xmode=log, ymode=log,
  log ticks with fixed point, scaled ticks=false,
  minor tick style={draw=none},
  grid=major, major tick length=2pt,
  minor grid style={draw=none}, 
  xtick={0.1,1,10,100,1000,10000,100000}, xticklabels={10$^{-1}$,1,10,10$^2$,10$^3$,10$^4$,10$^5$},
  xticklabel style={font=\scriptsize, text height=1.5ex, text depth=.25ex}, 
  ytick={0.1,1,10,100,1000,10000}, 
  yticklabels={10$^{-1}$,1,10,10$^2$,10$^3$,10$^4$},
  yticklabel style={font=\scriptsize, text width=1.5em, align=right, text height=1.5ex, text depth=.25ex},
  tick align=outside, tick style={black},
  label style={font=\scriptsize},
  xlabel={time (s)}, ylabel={},
  xmin=1, xmax=10000, ymin=1, ymax=10000
]
  \addplot[only marks, mark=pentagon, mark size=3pt, thick, draw=pmsjaCol] coordinates {
  (2955.33,39.62)
}; 
  \addplot[mark=triangle, mark size=3pt, mark options={fill=white, line width=1pt}, thick, draw=spCol] coordinates {
(3.174,1039.12) (7.083,679.99) (16.70,511.46) (55.29,314.29)
}; 
  \addplot[mark=o, mark size=3pt, thick, draw=dpsCol] coordinates {
(4.17,82.48) (6.92, 56.46) (12.95,56.35) (53.88,55.82)
}; 
\end{axis}
\end{tikzpicture}
} & \hspace{-6mm}
\subcaptionbox{$\boldsymbol{Q}_{<}$ (base query time: 1.28s).}[ \W ]{%
\begin{tikzpicture}
\begin{axis}[
  width=\W, height=\Hgt,
  xmode=log, ymode=log,
  log ticks with fixed point, scaled ticks=false,
  minor tick style={draw=none},
  grid=major, major tick length=2pt,
  minor grid style={draw=none}, 
  xtick={0.1,1,10,100,1000,10000,100000}, xticklabels={10$^{-1}$,1,10,10$^2$,10$^3$,10$^4$,10$^5$},
  xticklabel style={font=\scriptsize, text height=1.5ex, text depth=.25ex}, 
  ytick={0.1,1,10,100,1000,10000}, 
  yticklabels={10$^{-1}$,1,10,10$^2$,10$^3$,10$^4$},
  yticklabel style={font=\scriptsize, text width=1.5em, align=right, text height=1.5ex, text depth=.25ex},
  tick align=outside, tick style={black},
  label style={font=\scriptsize},
  xlabel={time (s)}, ylabel={},
  xmin=1, xmax=100000, ymin=1, ymax=1000
]
  \addplot[only marks, mark=pentagon, mark size=3pt, thick, draw=pmsjaCol] coordinates {
  (16719,12.76)
}; 
  \addplot[mark=triangle, mark size=3pt, mark options={fill=white, line width=1pt}, thick, draw=spCol] coordinates {
(3.95,765.24) (9.14,457.25)(21.52,314.27)(43.39,214.56)
}; 
  \addplot[mark=o, mark size=3pt, thick, draw=dpsCol] coordinates {
(3.44,48.43)(6.08,27.36)(15.57,19.44) (43.64,24.69)
}; 
\end{axis}
\end{tikzpicture}
} 

\end{tabular}
}}
\vspace{-3mm}
\caption{Query error vs.\ time on C2A with $\varepsilon = 4$, $\delta=10^{-7}$. (Sampling rates $q\in\left\{ \frac{1}{32},\tfrac{1}{16}, \tfrac{1}{8}, \tfrac{1}{4}\right\}$)}
\label{fig:c2a}

\end{figure}

%% file: fig/amazon2-eps.tex
\begin{figure}[t]
\centering

\begin{tikzpicture}
  \node{
    \begin{tabular}{@{}c@{\quad}c@{\qquad}c@{\quad}c@{\qquad}c@{\quad}c@{\qquad}c@{\quad}c@{\qquad}c@{\quad}c@{}}\    \tikz{\draw[thick, draw=r2tCol]     plot[mark=asterisk, mark size=3pt] (0,0);} & R2T &
    \tikz{\draw[thick, draw=pmsjaCol]      plot[mark=pentagon, mark size=3pt, mark options={fill=white, line width=1pt}] (0,0);} & PMSJA &
    \tikz{\draw[thick, draw=red]     plot[mark=10-pointed star, mark size=3pt] (0,0);} & GS &
    \tikz{\draw[thick, draw=spCol]      plot[mark=triangle*, mark size=3pt, mark options={fill=white, line width=1pt}] (0,0);} & S\&E &
    \tikz{\draw[thick, draw=dpsCol]     plot[mark=o, mark size=3pt] (0,0);} & DP-S4S
    \end{tabular}
  };
\end{tikzpicture}

\makebox[\textwidth][c]{%
   \resizebox{\linewidth}{!}{%

\begin{tabular}{@{}c@{}c@{}c@{}c@{}}

\hspace{-2mm}
\subcaptionbox{$\boldsymbol{Q}_{1-}$}[ \W ]{%
\begin{tikzpicture}
\begin{axis}[
  width=\W, height=\Hgt,
  xmode=log, ymode=log,
  log ticks with fixed point, scaled ticks=false,
  minor tick style={draw=none},
  grid=major, major tick length=2pt,
  minor grid style={draw=none},  
  xtick={0.1,0.5,1,5,10}, xticklabels={0.1,,1,,10},
  ytick={0.1,1,10,100}, yticklabels={0.1,1,10,100},
  tick align=outside, tick style={black},
  tick label style={font=\scriptsize}, label style={font=\scriptsize},
  xlabel={$\varepsilon$}, ylabel={relative error (\%)},
  xmin=0.1, xmax=10, ymin=0.1, ymax=100,
  ylabel style={yshift=-4mm}
]
  \addplot[mark=asterisk, mark size=3pt, thick, draw=r2tCol] coordinates {
(0.25,1.3656) (0.5,0.7407) (1,0.4282) (2,0.2720) (4,0.1892)}; 
  \addplot[mark=triangle, mark size=3pt, mark options={fill=white, line width=1pt}, thick, draw=spCol] coordinates {
    (0.25,21.3352) (0.5,13.6337) (1,7.4417) (2,4.1208) (4,2.5884)
 }; 
  \addplot[mark=o, mark size=3pt, thick, draw=dpsCol] coordinates {
  (0.25,2.2924) (0.5,1.2434) (1,0.8714) (2,0.7602) (4,0.7189)
}; 
\end{axis}
\end{tikzpicture}
} & \hspace{-5mm}
\subcaptionbox{$\boldsymbol{Q}_{2-}$}[ \W ]{%
\begin{tikzpicture}
\begin{axis}[
  width=\W, height=\Hgt,
  xmode=log, ymode=log,
  log ticks with fixed point, scaled ticks=false,
  minor tick style={draw=none},
  grid=major, major tick length=2pt,
  minor grid style={draw=none},  
  xtick={0.1,0.5,1,5,10}, xticklabels={0.1,,1,,10},
  ytick={0.1,1,10,100}, yticklabels={0.1,1,10,100},
  tick align=outside, tick style={black},
  tick label style={font=\scriptsize}, label style={font=\scriptsize},
  xlabel={$\varepsilon$}, ylabel={},
  xmin=0.1, xmax=10, ymin=1, ymax=100,
  ylabel style={yshift=-4mm}
]
  \addplot[mark=asterisk, mark size=3pt, thick, draw=r2tCol] coordinates {
  (0.25,15.36) (0.5,11.1936)
  (1,9.1137)
  (2,7.1574)
  (4,4.4236)
}; 
  \addplot[mark=triangle, mark size=3pt, mark options={fill=white, line width=1pt}, thick, draw=spCol] coordinates {
      (0.25,36.8274) (0.5,22.0766) (1,14.72) (2,10.7443) (4,7.6804)
 }; 
  \addplot[mark=o, mark size=3pt, thick, draw=dpsCol] coordinates {
    (0.25,14.0751) (0.5,10.8735) (1,7.0815) (2,4.8275) (4,3.7027)
}; 
\end{axis}
\end{tikzpicture}
} & \hspace{-6mm}
\subcaptionbox{$\boldsymbol{Q}_{\triangle}$}[ \W ]{%
\begin{tikzpicture}
\begin{axis}[
  width=\W, height=\Hgt,
  xmode=log, ymode=log,
  log ticks with fixed point, scaled ticks=false,
  minor tick style={draw=none},
  grid=major, major tick length=2pt,
  minor grid style={draw=none},  
  xtick={0.1,0.5,1,5,10}, xticklabels={0.1,,1,,10},
  ytick={0.1,1,10,100}, yticklabels={0.1,1,10,100},
  tick align=outside, tick style={black},
  tick label style={font=\scriptsize}, label style={font=\scriptsize},
  xlabel={$\varepsilon$}, ylabel={},
  xmin=0.1, xmax=10, ymin=0.1, ymax=100,
  ylabel style={yshift=-4mm}
]
  \addplot[mark=asterisk, mark size=3pt, thick, draw=r2tCol] coordinates {
    (0.25,5.0036) (0.5,2.8686)  (1,1.7967) (2,1.2624) (4, 0.9980)
  }; 
  \addplot[mark=triangle, mark size=3pt, mark options={fill=white, line width=1pt}, thick, draw=spCol] coordinates {
    (0.25,48.2453) (0.5,39.6622)  (1,22.7302) (2,14.1549) (4, 9.6151)
}; 
  \addplot[mark=o, mark size=3pt, thick, draw=dpsCol] coordinates {
    (0.25,12.7301) (0.5,6.6038)  (1,3.6093) (2,2.1041) (4, 1.5598)
  }; 
\end{axis}
\end{tikzpicture}
} & \hspace{-6mm}
\subcaptionbox{$\boldsymbol{Q}_{\square}$}[ \W ]{%
\begin{tikzpicture}
\begin{axis}[
  width=\W, height=\Hgt,
  xmode=log, ymode=log,
  log ticks with fixed point, scaled ticks=false,
  minor tick style={draw=none},
  grid=major, major tick length=2pt,
  minor grid style={draw=none},  
  xtick={0.1,0.5,1,5,10}, xticklabels={0.1,,1,,10},
  ytick={0.1,1,10,100}, yticklabels={0.1,1,10,100},
  tick align=outside, tick style={black},
  tick label style={font=\scriptsize}, label style={font=\scriptsize},
  xlabel={$\varepsilon$}, ylabel={},
  xmin=0.1, xmax=10, ymin=1, ymax=100,
  ylabel style={yshift=-4mm}
]
  \addplot[mark=asterisk, mark size=3pt, thick, draw=r2tCol] coordinates {
    (0.25,16.3366) (0.5, 12.2547) (1, 10.2137)
    (2, 6.69259)
    (4, 3.94380)
  }; 
  \addplot[mark=triangle, mark size=3pt, mark options={fill=white, line width=1pt}, thick, draw=spCol] coordinates {
    (0.25,70.7358) (0.5,58.5833)  (1,45.1040) (2,27.9551) (4, 19.6332)
}; 
  \addplot[mark=o, mark size=3pt, thick, draw=dpsCol] coordinates {
    (0.25,18.8890) (0.5,12.9277)  (1,7.8862) (2,5.3939) (4, 4.0788)
}; 
\end{axis}
\end{tikzpicture}
} 
\end{tabular}
}}
\vspace{-3mm}
\caption{Query error vs.\ $\varepsilon$ on Amazon with $q = 1\%$ for sampling mechanisms.}
\label{fig:amazon2-eps}
\vspace{-2mm}
\end{figure}
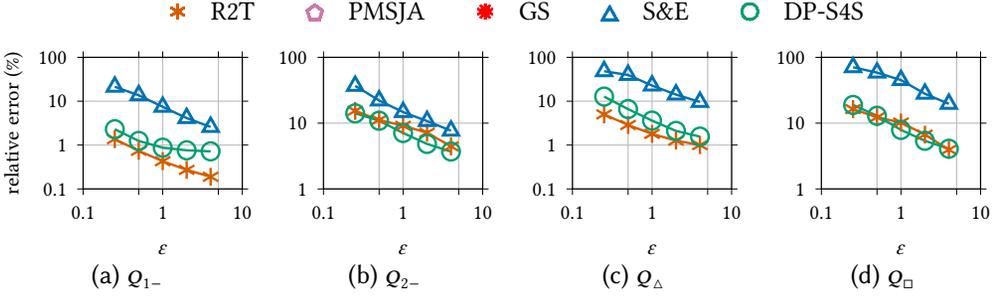

%% file: fig/c2a-eps.tex
\begin{figure}[t]
\centering

\makebox[\textwidth][c]{%
   \resizebox{\linewidth}{!}{%
\begin{tabular}{@{}c@{}c@{}c@{}c@{}}

\hspace{-2mm}
\subcaptionbox{$\boldsymbol{Q}_{1-}$}[ \W ]{%
\begin{tikzpicture}
\begin{axis}[
  width=\W, height=\Hgt,
  xmode=log, ymode=log,
  log ticks with fixed point, scaled ticks=false,
  minor tick style={draw=none},
  grid=major, major tick length=2pt,
  minor grid style={draw=none},  
  xtick={1,2,4,8}, xticklabels={1,2,4,8},
  xticklabel style={font=\scriptsize, text height=1.5ex, text depth=.25ex}, 
  ytick={0.1,1,10,100,1000,10000}, 
  yticklabels={10$^{-1}$,1,10,10$^2$,10$^3$,10$^4$},
  yticklabel style={font=\scriptsize, text width=1.5em, align=right, text height=1.5ex, text depth=.25ex},   
  tick align=outside, tick style={black},
  label style={font=\scriptsize},
  xlabel={$\varepsilon$}, ylabel={relative error (\%)},
  xmin=0.8, xmax=10, ymin=1, ymax=10000,
  ylabel style={yshift=-4mm}
]
  \addplot[mark=pentagon, mark size=3pt, thick, draw=pmsjaCol] coordinates {
(1, 83.4040) (2, 32.0811) (4,9.6077) (8,2.7825)
}; 
  \addplot[mark=triangle, mark size=3pt, mark options={fill=white, line width=1pt}, thick, draw=spCol] coordinates { 
(1, 1672.786) (2, 569.2037) (4, 190.1409) (8, 75.4504)
 }; 
  \addplot[mark=o, mark size=3pt, thick, draw=dpsCol] coordinates {
(1, 230.3819) (2, 88.6889) (4,42.9156) (8,13.2193)
}; 
\end{axis}
\end{tikzpicture}
} & \hspace{-5mm}
\hspace{-2mm}
\subcaptionbox{$\boldsymbol{Q}_{2-}$}[ \W ]{%
\begin{tikzpicture}
\begin{axis}[
  width=\W, height=\Hgt,
  xmode=log, ymode=log,
  log ticks with fixed point, scaled ticks=false,
  minor tick style={draw=none},
  grid=major, major tick length=2pt,
  minor grid style={draw=none},  
  xtick={1,2,4,8}, xticklabels={1,2,4,8},
  xticklabel style={font=\scriptsize, text height=1.5ex, text depth=.25ex}, 
  ytick={0.1,1,10,100,1000,10000}, 
  yticklabels={10$^{-1}$,1,10,10$^2$,10$^3$,10$^4$},
  yticklabel style={font=\scriptsize, text width=1.5em, align=right, text height=1.5ex, text depth=.25ex},   
  tick align=outside, tick style={black},
  label style={font=\scriptsize},
  xlabel={$\varepsilon$}, ylabel={},
  xmin=0.8, xmax=10, ymin=1, ymax=10000,
  ylabel style={yshift=-4mm}
]
  \addplot[mark=pentagon, mark size=3pt, thick, draw=pmsjaCol] coordinates {
(1,217.3099) (2,55.74) (4, 20.3542) (8, 8.5423)
}; 
  \addplot[mark=triangle, mark size=3pt, mark options={fill=white, line width=1pt}, thick, draw=spCol] coordinates {
(1,3316.1705) (2,974.9048) (4,347.11053) (8,132.0732)
}; 
  \addplot[mark=o, mark size=3pt, thick, draw=dpsCol] coordinates {
(1,353.7621) (2,129.9661) (4,44.8931) (8,13.8727)
  }; 
\end{axis}
\end{tikzpicture}
} & \hspace{-6mm}
\subcaptionbox{$\boldsymbol{Q}_{3-}$}[ \W ]{%
\begin{tikzpicture}
\begin{axis}[
  width=\W, height=\Hgt,
  xmode=log, ymode=log,
  log ticks with fixed point, scaled ticks=false,
  minor tick style={draw=none},
  grid=major, major tick length=2pt,
  minor grid style={draw=none},  
  xtick={1,2,4,8}, xticklabels={1,2,4,8},
  xticklabel style={font=\scriptsize, text height=1.5ex, text depth=.25ex}, 
  ytick={0.1,1,10,100,1000,10000}, 
  yticklabels={10$^{-1}$,1,10,10$^2$,10$^3$,10$^4$},
  yticklabel style={font=\scriptsize, text width=1.5em, align=right, text height=1.5ex, text depth=.25ex},   
  tick align=outside, tick style={black},
  label style={font=\scriptsize},
  xlabel={$\varepsilon$}, ylabel={},
  xmin=0.8, xmax=10, ymin=1, ymax=10000,
  ylabel style={yshift=-4mm}
]
  \addplot[mark=pentagon, mark size=3pt, thick, draw=pmsjaCol] coordinates {
  (1,247.9240) (2,90.1030) (4, 23.1737) (8,8.6722)
}; 
  \addplot[mark=triangle, mark size=3pt, mark options={fill=white, line width=1pt}, thick, draw=spCol] coordinates {
  (1, 3456.5575)  (2, 1224.2231) (4, 424.4094) (8, 147.4067)
}; 
  \addplot[mark=o, mark size=3pt, thick, draw=dpsCol] coordinates {
  (1, 490.088)  (2, 163.6385)  (4, 65.7314)  (8, 18.4947)  
}; 
\end{axis}
\end{tikzpicture}
} & \hspace{-6mm}
\subcaptionbox{$\boldsymbol{Q}_{<}$}[ \W ]{%
\begin{tikzpicture}
\begin{axis}[
  width=\W, height=\Hgt,
  xmode=log, ymode=log,
  log ticks with fixed point, scaled ticks=false,
  minor tick style={draw=none},
  grid=major, major tick length=2pt,
  minor grid style={draw=none},  
  xtick={1,2,4,8}, xticklabels={1,2,4,8},
  xticklabel style={font=\scriptsize, text height=1.5ex, text depth=.25ex}, 
  ytick={0.1,1,10,100,1000,10000}, 
  yticklabels={10$^{-1}$,1,10,10$^2$,10$^3$,10$^4$},
  yticklabel style={font=\scriptsize, text width=1.5em, align=right, text height=1.5ex, text depth=.25ex},   
  tick align=outside, tick style={black},
  label style={font=\scriptsize},
  xlabel={$\varepsilon$}, ylabel={},
  xmin=0.8, xmax=10, ymin=1, ymax=10000,
  ylabel style={yshift=-4mm}
]
  \addplot[mark=pentagon, mark size=3pt, thick, draw=pmsjaCol] coordinates {
  (1, 112.8944) (2, 35.9237) (4,12.7411) (8,5.3398)
}; 
  \addplot[mark=triangle, mark size=3pt, mark options={fill=white, line width=1pt}, thick, draw=spCol] coordinates {
(1, 2820.0685) (2, 666.9916) (4, 215.3008) (8, 91.5674)
}; 
  \addplot[mark=o, mark size=3pt, thick, draw=dpsCol] coordinates {
(1, 176.4841) (2, 67.8083) (4, 28.8187) (8, 4.9877)
}; 
\end{axis}
\end{tikzpicture}
} 

\end{tabular}
}}
\vspace{-3mm}
\caption{Query error vs. $\varepsilon$ on C2A with $\delta=10^{-5}$ and $q=10\%$ for sampling mechanisms.}
\label{fig:c2a-eps}
\vspace{-2mm}
\end{figure}
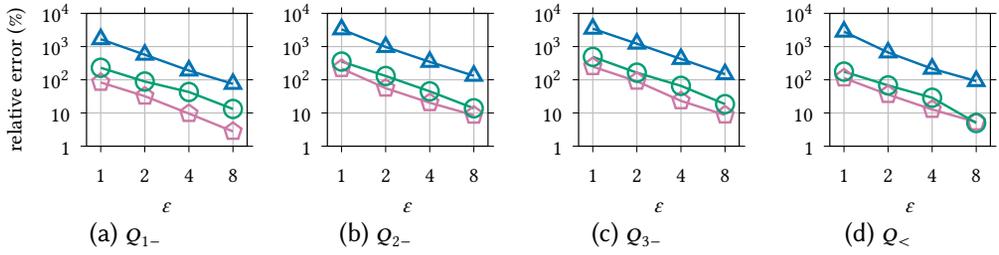

%% file: fig/amazon2-d.tex
\begin{figure}[t]
\centering

\makebox[\textwidth][c]{%
   \resizebox{\linewidth}{!}{%
\begin{tabular}{@{}c@{}c@{}c@{}c@{}}

\hspace{-2mm}
\subcaptionbox{$\boldsymbol{Q}_{1-}$}[ \W ]{%
\begin{tikzpicture}
\begin{axis}[
  width=\W, height=\Hgt,
  xmode=log, ymode=log,
  log ticks with fixed point, scaled ticks=false,
  minor tick style={draw=none},
  grid=major, major tick length=2pt,
  minor grid style={draw=none},  
  xtick={512,1024,2048,4096}, xticklabels={512,1024,2048,4096},
  ytick={0.01,0.1,1,10}, yticklabels={0.01,0.1,1,10},
  tick align=outside, tick style={black},
  tick label style={font=\scriptsize}, label style={font=\scriptsize},
  xlabel={$D$}, ylabel={relative error (\%)},
  xmin=450, xmax=4600, ymin=0.01, ymax=10,
  ylabel style={yshift=-4mm}
]
  \addplot[mark=10-pointed star, mark size=3pt, thick, draw=red] coordinates {
(512,0.0546) (1024,0.1092) (2048,0.2184) (4096, 0.4368)}; 
  \addplot[mark=asterisk, mark size=3pt, thick, draw=r2tCol] coordinates {
(512,0.3755) (1024,0.4282) (2048,0.4282) (4096,0.4282)}; 
  \addplot[mark=triangle, mark size=3pt, mark options={fill=white, line width=1pt}, thick, draw=spCol] coordinates {
(512,6.1785) (1024,7.4417) (2048,7.4417) (4096,7.4417)
 }; 
  \addplot[mark=o, mark size=3pt, thick, draw=dpsCol] coordinates {
(512,0.8499) (1024,0.8714) (2048,0.8714) (4096,0.8714)
}; 
\end{axis}
\end{tikzpicture}
} & \hspace{-5mm}
\subcaptionbox{$\boldsymbol{Q}_{2-}$}[ \W ]{%
\begin{tikzpicture}
\begin{axis}[
  width=\W, height=\Hgt,
  xmode=log, ymode=log,
  log ticks with fixed point, scaled ticks=false,
  minor tick style={draw=none},
  grid=major, major tick length=2pt,
  minor grid style={draw=none},  
  xtick={512,1024,2048,4096}, xticklabels={512,1024,2048,4096},
  ytick={0.1,1,10,100,1000}, yticklabels={0.1,1,10,100,1000},
  tick align=outside, tick style={black},
  tick label style={font=\scriptsize}, label style={font=\scriptsize},
  xlabel={$D$}, ylabel={},
  xmin=450, xmax=4600, ymin=1, ymax=1000,
  ylabel style={yshift=-4mm}
]
  \addplot[mark=10-pointed star, mark size=3pt, thick, draw=red] coordinates {
(512,2.6542) (1024,10.6169) (2048,42.4677) (4096, 169.8711)}; 
  \addplot[mark=asterisk, mark size=3pt, thick, draw=r2tCol] coordinates {
(512,8.7018) (1024,9.1137) (2048,9.1137) (4096,9.5695)
}; 
  \addplot[mark=triangle, mark size=3pt, mark options={fill=white, line width=1pt}, thick, draw=spCol] coordinates {
(512,13.3234) (1024,14.7296) (2048,14.7296) (4096,14.8525)
 }; 
  \addplot[mark=o, mark size=3pt, thick, draw=dpsCol] coordinates {
(512,6.4269) (1024,7.0815) (2048,7.0815) (4096,9.3050)
}; 
\end{axis}
\end{tikzpicture}
} & \hspace{-6mm}
\subcaptionbox{$\boldsymbol{Q}_{\triangle}$}[ \W ]{%
\begin{tikzpicture}
\begin{axis}[
  width=\W, height=\Hgt,
  xmode=log, ymode=log,
  log ticks with fixed point, scaled ticks=false,
  minor tick style={draw=none},
  grid=major, major tick length=2pt,
  minor grid style={draw=none},  
  xtick={512,1024,2048,4096}, xticklabels={512,1024,2048,4096},
  ytick={0.1,1,10,100,1000}, yticklabels={0.1,1,10,100,1000},
  tick align=outside, tick style={black},
  tick label style={font=\scriptsize}, label style={font=\scriptsize},
  xlabel={$D$}, ylabel={},
  xmin=450, xmax=4600, ymin=1, ymax=1000,
  ylabel style={yshift=-4mm}
]
  \addplot[mark=10-pointed star, mark size=3pt, thick, draw=red] coordinates {
(512,38.7999) (1024,155.1999) (2048,620.7999) (4096, 2483.1998)}; 
  \addplot[mark=asterisk, mark size=3pt, thick, draw=r2tCol] coordinates {
  (512,1.6694) (1024,1.7967) (2048,1.7967) (4096, 1.9357)
}; 
  \addplot[mark=triangle, mark size=3pt, mark options={fill=white, line width=1pt}, thick, draw=spCol] coordinates {
(512,19.5289) (1024,22.7302) (2048,22.7302) (4096,22.6242)
 }; 
  \addplot[mark=o, mark size=3pt, thick, draw=dpsCol] coordinates {
(512,2.2261) (1024,3.6093) (2048,3.6093) (4096,3.4815)
}; 
\end{axis}
\end{tikzpicture}
} & \hspace{-6mm}
\subcaptionbox{$\boldsymbol{Q}_{\square}$}[ \W ]{%
\begin{tikzpicture}
\begin{axis}[
  width=\W, height=\Hgt,
  xmode=log, ymode=log,
  log ticks with fixed point, scaled ticks=false,
  minor tick style={draw=none},
  grid=major, major tick length=2pt,
  minor grid style={draw=none},  
  xtick={512,1024,2048,4096}, xticklabels={512,1024,2048,4096},
  ytick={0.1,1,10,100,1000,10000}, yticklabels={0.1,1,10,100,1000,10000},
  tick align=outside, tick style={black},
  tick label style={font=\scriptsize}, label style={font=\scriptsize},
  xlabel={$D$}, ylabel={},
  xmin=450, xmax=4600, ymin=1, ymax=10000,
  ylabel style={yshift=-4mm}
]
  \addplot[mark=10-pointed star, mark size=3pt, thick, draw=red] coordinates {
(512,4240.4951) (1024,33923.96) (2048,271391.6923) (4096, 2171133.53855)}; 
  \addplot[mark=asterisk, mark size=3pt, thick, draw=r2tCol] coordinates {
(512,9.3316) (1024,10.2137)  (2048,10.2393) (4096,10.4149)
}; 
  \addplot[mark=triangle, mark size=3pt, mark options={fill=white, line width=1pt}, thick, draw=spCol] coordinates {
(512,39.2683) (1024,45.1040) (2048,46.5591) (4096,53.5821)
 }; 
  \addplot[mark=o, mark size=3pt, thick, draw=dpsCol] coordinates {
(512,6.0798) (1024,7.8862) (2048,8.3935) (4096,8.9581)
}; 
\end{axis}
\end{tikzpicture}
} 
\end{tabular}
}}
\vspace{-3mm}
\caption{Query error vs.\ $D$ on Amazon with $\varepsilon=1$ and $q = 1\%$ for sampling mechanisms.}
\label{fig:amazon2-d}
\end{figure}
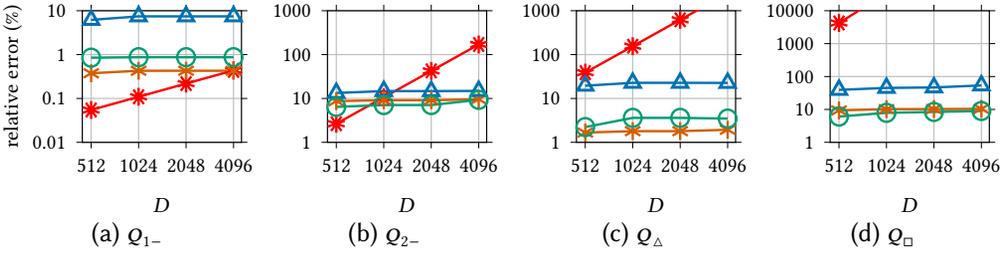

%% file: sections/conclusion.tex
\section{Conclusion}\label{sec:conclusion}

We have designed and implemented a novel sampling-based mechanism {\sf DP-S4S} for efficiently and accurately answering both scalar and vector SJA queries under user-level differential privacy. The main ideas are to (i) sample join tuples and (ii) employ R\'enyi DP as an intermediate step to reach the required ($\varepsilon$, $\delta$)-DP. 
For scalar queries, we have derived a strong sampling amplification result, leading to {\sf DP-S4S} matching (and sometimes even exceeding) the accuracy of the non-sampling solution \textsf{R2T}, at a fraction of the latter's cost. For vector queries, we have additionally built the necessary infrastructure: the first RDP smooth sensitivity mechanism, which can be of independent interest in tasks such as releasing graph statistics or computing gradient vectors in machine learning. Extensive experiments confirm the consistent and significant (often 10x) accuracy improvement over the previous state of the art \textsf{S\&E}.

For future work, we plan to investigate whether the amplification from sampling aggregation units extends to arbitrary $(\varepsilon,\delta)$-DP mechanisms beyond R\'enyi DP. 
In addition, enforcing DP on streaming data~\cite{DBLP:conf/stoc/DworkNPR10,DBLP:conf/pet/ChanLSX12,DBLP:conf/dbsec/CaoCPNY23,DBLP:journals/pacmmod/0007CLS024} via dynamic sampling is another exciting future topic.

%% file: sections/useful-lemmas.tex
\section{Useful Lemmas}

\begin{lemma}[Quasi-convexity of R\'enyi Divergence~\cite{DBLP:journals/tit/ErvenH14}]\label{lemma:quasi}
For any pair of distributions $(P, Q)$ and $(P', Q')$ and any $0<\beta<1$, we have 
\[D_\alpha(\beta P+(1-\beta) P'\| \beta Q+(1-\beta) Q')\leq \max\{D_\alpha(P\|Q), D_\alpha(P'\|Q')\}\,.
\]
\end{lemma}

\begin{lemma}[Convexity of R\'enyi Divergence in its Second Term~\cite{DBLP:journals/tit/ErvenH14}]\label{lemma:convex}
For any distributions $P$, $Q$, $Q'$ and any $0<\beta<1$, we have 
\[D_\alpha(P\|\beta Q+(1-\beta) Q')\leq \beta D_\alpha(P\|Q)+(1-\beta) D_\alpha(P'\|Q')\,.
\]
\end{lemma}

\begin{lemma}[Bernstein inequality]\label{lemma:bernstein}
Let $X_1,\dots,X_k$ be independent zero-mean random variables such that $|X_i|\leq 1$ for all $i$, then
\[
\Pr\left[\left|\sum_{i=1}^k X_i\right|\geq t\right]\leq 2\exp\left(-\frac{t^2/2}{\sum_{i=1}^k \mathbf{E}[X_i^2]+t/3}\right)\,.
\]
\end{lemma}

\begin{lemma}[Chernoff bound]\label{lemma:chernoff}
Let $X_1,\dots,X_k$ be independent Bernoulli random variables, where $\mathbf{E}[\sum_{i=1}^k X_i]=\mu$. For $0\leq t\leq 4$,
\[
\Pr\left[\sum_{i=1}^k X_i\geq(1+t) \mu\right]\leq \exp\left(-\frac{t^2\mu}{4}\right)\,.
\]
\end{lemma}

%% file: sections/proofs.tex
\section{Proofs}

\subsection{Proof for Lemma~\ref{lemma:s4s-s}} \label{sec:proof2}

\begin{proof}
Denote by $f(S):=\sum_{t\in S} w_t$. Note that $f(S)=\sum_{t\in T}\mathbf{1}[t\text{ is sampled}]\cdot w_t$. Therefore
\[
\mathbf{E}[f(S)]=q\cdot f(T),\text{ and }\mathbf{Var}[f(S)]=q(1-q)\sum_{t\in T} w_t^2\leq qn\,.
\]
We have $\Pr\left[\left|\Lap(\frac{\tau_i}{\varepsilon_i})\right|\leq \frac{\tau_i}{\varepsilon_i}\ln\frac{3L}{\beta}\right]>1-\frac{\beta}{3L}$.
By a union bound, this holds for all iterations with probability at least $1-\beta/3$.
Conditioned on this happening, $\tilde{f}(S,\tau_i)$ is always an underestimate since $\tilde{f}(S,\tau_i)\leq \bar{f}(S,\tau_i)\leq f(S)$ and $\tilde{f}(S,\tau_{i^*})\geq f(S) - 2(\tau_{i^*} / \varepsilon_{i^*})\ln(3L/\beta)$.
Let $\tau^*(S):=\max_{u\in U} \sum_{t\in S} \mathbf{1}[u\leadsto t
]$ be the maximum number of records contributed by a user in the sample.
For any $\tau\geq  \tau^*(S)$, the optimal solution to problem \eqref{eq:r2t} assigns $x_t=w_t$ so that $\bar{f}(S,\tau)=f(S)$.
In particular, this holds for a power of 2 where $\tau_{i^*} \geq \tau^*(S) > \tau_{{i^*}-1}$. Thus, we have both $\tilde{f}(S,\tau_{i^*})\leq\max_i\tilde{f}(S,\tau_i)\leq  f(S)$, and
\[
\tilde{f}(S,\tau_{i^*})\geq f(S)-\frac{2\tau_{i^*}}{\varepsilon_{i^*}}\ln\frac{3L}{\beta}>f(S)-\frac{4\tau^*(S)}{\varepsilon_{i^*}}\ln\frac{3L}{\beta}\,.
\]

By Bernstein inequality (Lemma~\ref{lemma:bernstein}), we have
\[
\Pr\left[\left|f(S)-q\cdot f(T)\right|>\sqrt{2nq\ln\frac{6}{\beta}+\frac{1}{9}\left(\ln\frac{6}{\beta}\right)^2}+\frac{1}{3}\ln\frac{6}{\beta}\right]<
\frac{\beta}{3}\,.
\]
Therefore with total probability at least $1-2\beta/3$, we have
\begin{align*}
|\hat{f}(T)-f(T)|&\leq \frac{1}{q}\left(|\max_i\tilde{f}(S,\tau_i)-f(S)|+|f(S)-q\cdot f(T)|\right) \\
&=\tilde{O}\left( \frac{\tau^*(S)}{q\varepsilon_{i^*}} + \sqrt{\frac{n}{q}}+\frac{1}{q}\right)\,,
\end{align*}
where $\tilde{O}(\cdot)$ omits constants and polylogarithmic factors.

Our final step is upper bound $\tau^*(S)$, the contribution of each user in the sample.
Define $\tau^*(T):\max_{u\in U}\sum_{t\in T}\mathbf{1}[u\leadsto t]$, so that each user has at most $\tau^*(T)$ records in the dataset.
By Chernoff bound (Lemma~\ref{lemma:chernoff}), we have for any user $u\in U$, 
\[
\Pr\left[|S_u|>q\cdot \tau^*(T)+\sqrt{4q\cdot \tau^*(T)\ln\frac{3m}{\beta}}\right] \leq \frac{\beta}{3m}\,,
\]
where $m$ is the number of users.
With a union bound, we have $\tau^*(S)=\tilde{O}(q\cdot \tau^*(T))$ with probability at least $1-\beta/3$.

In summary, with probability $1-\beta$, the error of {\sf DP-S4S} is
\[
\sqrt{\frac{2n}{q}\ln\frac{6}{\beta}
}+\frac{1}{3q}\ln\frac{6}{\beta}+\frac{1}{\varepsilon_{i^*}}\left(\tau^*(T)+\sqrt{\frac{4\tau^*(T)}{q}\ln\frac{3m}{\beta}}\right)\,.
\]
\end{proof}

\input{sections/rdp-amplification}

\subsection{Proof for Lemma~\ref{lemma:additive-smooth}}\label{sec:proof3}

\begin{proof}
Consider $\mathbf{I}\sim \mathbf{I}'$, let $\sigma(\mathbf{I})=\sigma_1$ and $\sigma(\mathbf{I}')=\sigma_2$.
By the smoothness of $\operatorname{SS}_F$, we have
\[
e^{-2\gamma}\leq \frac{\sigma_2^2}{\sigma_1^2}=\frac{\operatorname{SS}_F(\mathbf{I}',\gamma)^2}{\operatorname{SS}_F(\mathbf{I},\gamma)^2}\leq e^{2\gamma}
\]
According to Lemma~\ref{lem:vec-diverge}, the divergence between the outputs is
\[
D_\alpha(\M(\mathbf{I})\|\M(\mathbf{I}'))=\frac{\alpha}{2} \cdot \frac{\|F(\mathbf{I})-F(\mathbf{I}')\|_2^2}{(1-\alpha)\sigma_1^2+\alpha\sigma_2^2}+\frac{d}{2(\alpha-1)}\ln\frac{(\sigma_2/\sigma_1)^{2\alpha}}{\alpha(\sigma_2/\sigma_1)^2+1-\alpha}
\]
when $\alpha(\sigma_2/\sigma_1)^2+1-\alpha>0$, which can be satisfied when
\[
\frac{\sigma_2^2}{\sigma_1^2}\geq e^{-2\gamma}>\frac{\alpha-1}{\alpha}\,,
\]
or $\gamma<\frac{1}{2}\ln\frac{\alpha}{\alpha-1}$.
For the first term, we have
\begin{align*}
\frac{\alpha}{2} \cdot \frac{\|F(\mathbf{I})-F(\mathbf{I}')\|_2^2}{(1-\alpha)\sigma_1^2+\alpha\sigma_2^2}&\leq \frac{\alpha}{2}\cdot\frac{\LS(\mathbf{I})^2}{\sigma_1^2\cdot (\alpha(\sigma_2/\sigma_1)^2+1-\alpha)}\\
&\leq \frac{\rho}{2}\cdot \frac{1-\alpha+\alpha e^{-2\gamma}}{1-\alpha+\alpha (\sigma_2^2/\sigma_1^2)}\\ 
&\leq \frac{\rho}{2}\,.
\end{align*}
For the second term, we require that 
\[
\frac{d}{2(\alpha-1)}\ln\frac{(\sigma_2/\sigma_1)^{2\alpha}}{\alpha(\sigma_2/\sigma_1)^2+1-\alpha}\leq \frac{\rho}{2}\,,
\]
which is equivalent to
\[
h(t):=t^{\alpha}-  \alpha At+(\alpha-1)A\leq 0\,,
\]
for $t=\sigma_2^2/\sigma_1^2$ and $A= e^{\frac{\rho(\alpha-1)}{d}}$.
Note that $h'(t)=\alpha t^{\alpha-1}-\alpha A$, therefore $h(t)$ decreases for $t<e^\frac{\rho}{d}$, and increases for $t>e^\frac{\rho}{d}$. Given that $h(\frac{\alpha-1}{\alpha})=t^\alpha>0$, $h((\alpha A)^\frac{1}{\alpha-1})=(\alpha-1)A>0$ and $h(1)=1-\alpha A+\alpha A-A<0$,  the equation $h(t)=0$ has two roots $\frac{\alpha-1}{\alpha}<t_1<1<t_2<(\alpha A)^\frac{1}{\alpha-1}$. The second condition is satisfied when
\[
t_1\leq e^{-2\gamma} \leq \frac{\sigma_2^2}{\sigma_1^2}\leq e^{2\gamma}\leq t_2\,.
\]
This gives $\gamma\leq\frac{1}{2} \min\{-\ln t_1, \ln t_2\}$.
Note that $t_1>\frac{\alpha-1}{\alpha}$, and, thus, the condition that $\gamma<\frac{1}{2}\ln\frac{\alpha}{\alpha-1}$ can be omitted.
The analysis for $D_\alpha(\M(\mathbf{I}')\|\M(\mathbf{I}))$ is symmetric.
\end{proof}

\subsection{Proof for Lemma~\ref{lemma:s4s-v}}\label{sec:proof4}

\begin{proof}
Similar to PMSJA~\cite{DBLP:journals/pacmmod/0007S023}, the output of Algorithm~\ref{algo:sample-pmsja} consists of two parts: an SVT mechanism $\M_1$ that searches for the truncation threshold $\tau$, and a smooth sensitivity mechanism $\M_2$ to release the truncated query for the given $\tau$. By the analysis in~\cite{DBLP:journals/pacmmod/0007S023},  $\M_1(S)$ is $(\varepsilon_1,0)$-DP with respect to the sample $S$.
To apply our Lemma~\ref{lemma:additive-smooth}, note that the local sensitivity of function $\bar{F}/2\tau$ is $\LS_{\bar{F}/2\tau}(S)=\LS_{\bar{F}}(S)/2\tau\leq E(S)+1$ by Inequality~\eqref{eq:pmsja_ls_bound}, and $B(S)=E(S)+1$ has global sensitivity $1$, thus $\M_2(S)$ is $(\alpha,\rho)$-RDP with repect to $S$ following Lemma~\ref{lemma:additive-smooth}. We will first argue about their privacy with respect to the original instance $\mathbf{I}$, and finally convert the guarantee back to $(\varepsilon,\delta)$-DP.

Recall the proof for Lemma~\ref{lemma:pure}. Note that $\M$ is not limited to any specific mechanism until we plug in the exact distribution for $\hat{f}$.
In particular, let $\mathbf{I}\sim \mathbf{I}'$ be neighboring datasets, where $\mathbf{I}'\supseteq \mathbf{I}$
 contains extra records from user $u$. Let $\mathcal{S}$ be the set of all possible samples from $\mathbf{I}$, and $\mathcal{S}_u$ be all possible samples contributed by user $u$  from $\mathbf{I}'$. We have
 \[
\Pr[(\M_1\circ\lambda)(\mathbf{I})=z]=\sum_{S\in\mathcal{S}}\Pr[\lambda(\mathbf{I})=S]\cdot \Pr[\mathcal{M}_1(S)=z]\,.
 \]

With similar steps, we can derive that
\[
\frac{\Pr[(\M_1\circ\lambda)(\mathbf{I})=z]}{\Pr[(\M_1\circ\lambda)(\mathbf{I}')=z]}
\leq \max_{S\in \mathcal{S}}\frac{1}{\sum_{S_u\in \mathcal{S}_u}\Pr[\lambda(T_u)=S_u]\frac{\Pr[\M_1(S\cup S_u)=z]}{\Pr[\M_1(S)=z]}} \,,
\]
where $T_u$ are all the join results  contributed by $u$. 
While it is hard to give a tighter bound for $\frac{\Pr[\M_1(S\cup S_u)=z]}{\Pr[\M_1(S)=z]}$ when $\M_1$ is the SVT mechanism, we can adopt a worst-case argument: when $S_u=\varnothing$, which happens with probability $(1-q)^\Delta$, the ratio is $1$; otherwise, we have $\frac{\Pr[\M_1(S\cup S_u)=z]}{\Pr[\M_1(S)=z]}\geq e^{-\varepsilon_1}$ due to the DP guarantee. The final privacy guarantee is then 
\[
\frac{\Pr[(\M_1\circ\lambda)(\mathbf{I})=z]}{\Pr[(\M_1\circ\lambda)(\mathbf{I}')=z]}\leq \frac{1}{(1-q)^\Delta+(1-(1-q)^\Delta)e^{-\varepsilon_1}}\leq e^{\varepsilon_1}\,.
\]
A similar argument can be made on the inverse, proving that the mechanism $\mathcal{M}_1$ is $\ln(1+q'(e^{\varepsilon_1}-1))$-DP on the original dataset for $q'=1-(1-q)^\Delta$. 

For the RDP mechanism $\M_2$, we follow the proof in Section~\ref{section:rdp-amp}. In particular, we can upper bound
\begin{align*}
& \exp((\alpha-1)\mathrm{D}_\alpha((\M_2\circ\lambda)(\mathbf{I})\|(M_2\circ\lambda)(\mathbf{I}'))) \\
\leq &\max_{S\in\mathcal{S}} \sum_{S_u\subseteq T_u} \Pr[\lambda(T_u)=S_u]\exp\left( (\alpha-1) \mathrm{D}_\alpha(\M_2(S)\|\M_2(S\cup S_u))\right)\\
\leq & (1-q)^\Delta+(1-(1-q)^\Delta)\exp((\alpha-1)\rho)
\end{align*}
Thus $\M_2$ is also $(\alpha,\rho)$-RDP on $\mathbf{I}$. Note that we set $\rho=\varepsilon_2-\frac{\ln(1/\delta)}{\alpha-1}$, by the conversion, $\M_2$ is also $(\varepsilon_2',\delta)$-DP for 
\[
\varepsilon_2'=\frac{1}{\alpha-1}\left(\ln\left(1+q'\left(e^{(\alpha-1)\varepsilon_2}\cdot \delta-1\right)\right)+\ln\frac{1}{\delta}\right)
\]
for the same $q'=1-(1-q)^\Delta$.
Finally, applying basic composition, the whole mechanism is $(\varepsilon_1'+\varepsilon_2',\delta)$-DP.
\end{proof}

%% file: sections/rdp-amplification.tex
\subsection{Proof of Lemma~\ref{lemma:sample-truncate}}\label{section:rdp-amp}

\begin{proof}
Following the proof for pure-DP, now we need to bound $\mathrm{D}_\alpha(\M(T)\|\M(T'))$ and $\mathrm{D}_\alpha(\M(T')\|\M(T))$ for R\'enyi-DP.
By the joint quasi-convexity of R\'enyi divergence (Lemma~\ref{lemma:quasi}), we have
\[
\mathrm{D}_\alpha(\M(T')\|\M(T)) \leq \max_{S\subseteq T} \mathrm{D}_\alpha\left(\sum_{S_u\subseteq T_u}\hat{f}(S\cup S_u,\tau)\middle\|\hat{f}(S,\tau)\right)\,.
\]

Denote $S^*=\arg\max_{S\subseteq T} \mathrm{D}_\alpha\left(\sum_{S_u\subseteq T_u}\hat{f}(S\cup S_u,\tau)\middle\|\hat{f}(S,\tau)\right)$, then similar to~\cite{DBLP:journals/pvldb/JiangLYX24}, we have
\begin{align}\label{eq:div1}
&\exp((\alpha-1)\mathrm{D}_\alpha(\M(T')\|\M(T))) \\
\leq &\sum_{S_u\subseteq T_u} \Pr[\lambda(T_u)=S_u]\cdot\exp((\alpha-1)\mathrm{D}_\alpha(\hat{f}(S^*\cup S_u,\tau)\| \hat{f}(S^*,\tau)))\,. \nonumber
\end{align}

By Lemma~\ref{lem:ec-gauss}, we have
\[
\mathrm{D}_\alpha(\hat{f}(S^*\cup S_u)\|\hat{f}(S^*)) = \frac{\alpha(\bar{f}(S^*\cup S_u,\tau)-\bar{f}(S^*))^2}{2\sigma^2}\leq \left(\frac{\min\{\tau, |S_u|\}}{\tau}\right)^2\cdot \rho\,.
\]
Plugging this into \eqref{eq:div1}, we have
\begin{align*}
&\exp((\alpha-1)\mathrm{D}_\alpha(M(T')\| M(T))) \\
\leq &\sum_{k=0}^\Delta \Pr[|\lambda(T_u)|=k]\cdot \exp\left((\alpha-1) \left(\frac{\min\{\tau,k\}}{\tau}\right)^2\rho\right) \\
=&\sum_{k=0}^{\lfloor\tau\rfloor} p_k \exp\left((\alpha-1)\frac{k^2}{\tau^2}\rho\right)  + \Pr[\operatorname{Bin}(\Delta,q)>\tau] \cdot \exp((\alpha-1)\rho)\,.
\end{align*}

Similarly, we can also upper bound 
\[
\mathrm{D}_\alpha(\M(T)\|\M(T')) \leq \max_{S\subseteq T} \mathrm{D}_\alpha\left(\hat{f}(S,\tau)\middle\|\sum_{S_u\subseteq T_u}\hat{f}(S\cup S_u,\tau)\right)\,.
\]
Let $S^\circ=\arg\max_{S\subseteq T} \mathrm{D}_\alpha\left(\hat{f}(S,\tau)\middle\|\sum_{S_u\subseteq T_u}\hat{f}(S\cup S_u,\tau)\right)$ be where the max is taken. By the convexity of R\'enyi-divergence in its second term (Lemma~\ref{lemma:convex}), we further have
\begin{align*}
&(\alpha-1)\mathrm{D}_\alpha(\M(T)\|\M(T'))\\
\leq & \sum_{S_u\subseteq T_u} \Pr[\lambda(T_u)=S_u]\cdot (\alpha-1)\mathrm{D}_\alpha(\hat{f}(S^\circ,\tau)\|\hat{f}(S^\circ\cup S_u,\tau))\\
=&\sum_{S_u\subseteq T_u} \Pr[\lambda(T_u)=S_u]\ln\left(\exp\left( (\alpha-1) \mathrm{D}_\alpha(\hat{f}(S^\circ,\tau)\|\hat{f}(S^\circ\cup S_u,\tau))\right)\right)\\
\leq& \ln\left(\sum_{S_u\subseteq T_u} \Pr[\lambda(T_u)=S_u]\exp\left( (\alpha-1) \mathrm{D}_\alpha(\hat{f}(S^\circ,\tau)\|\hat{f}(S^\circ\cup S_u,\tau))\right)\right)
\end{align*}
by the concavity of the $\ln(\cdot)$.
Following a similar analysis to equation~\eqref{eq:div1}, we have
\begin{align*}
&\exp((\alpha-1)\mathrm{D}_\alpha(\M(T)\|\M(T')))\\
\leq &\sum_{k=0}^{\lfloor\tau\rfloor} p_k \exp\left((\alpha-1)\frac{k^2}{\tau^2}\rho\right)  + \Pr[\operatorname{Bin}(\Delta,q)>\tau] \cdot \exp((\alpha-1)\rho)\,,
\end{align*}
which concludes the proof.
\end{proof}

%% file: sections/ablation.tex
\section{Evaluations on DP Noise vs. Sampling Error for Scalar Queries}\label{sec:ablation}
Recall that there are two sources of error in Algorithm~\ref{algo:sample-r2t}, sampling error and noise injected to satisfy DP. In this section, we investigate their respective contribution to the overall query error. We start by evaluating the magnitude of amplification in Lemma~\ref{lemma:pure}. 

\begin{table}[htbp]
    \centering
    \caption{DP-S4S Amplification ($\varepsilon=1$, $\Delta=1024$).}     \label{tab:amp}
    
    \begin{tabular}{ccccccc}
    \toprule
       $q\backslash \tau$& 1 & 4 & 16 & 64 & 256 & 1024 \\
       \midrule
       0.1\% & 0.7426 & 0.2878 & 0.0660 & 0.0161 & 0.0040 & 0.0010 \\ 
       1\% & 0.9999 & 0.9976 & 0.6538 & 0.1612 & 0.0400 & 0.0100\\
       10\% & 1.0000 & 1.0000 & 1.0000 & 0.9999 & 0.4007 & 0.1000\\
    \bottomrule
    \end{tabular}
\end{table}

Table~\ref{tab:amp} presents the effective $\varepsilon'(\varepsilon,\tau,\Delta,q)$ in Lemma~\ref{lemma:pure} with default $\varepsilon=1$ and $\Delta=1024$. Clearly, the amplification is significant if $q$ is small or $\tau$ is large, intuitively when the noise scale is over-killing.
This results in smaller noises in Algorithm~\ref{algo:sample-r2t} compared to {\sf R2T}, as we plot the standard deviations of Laplace noises $(\frac{\sqrt{2}\tau_i}{\varepsilon_i})$ from all mechanisms in Figure~\ref{fig:ablation}, with total $\varepsilon=1$ and $C=\Delta=1024$.
Note that {\sf S\&E} does not show clear benefits of privacy amplification even when $q=0.1\%$, whereas {\sf DP-S4S} achieves up to 5x improvement in accuracy. 

\setlength{\W}{0.33\linewidth}
\setlength{\Hgt}{0.9\W}

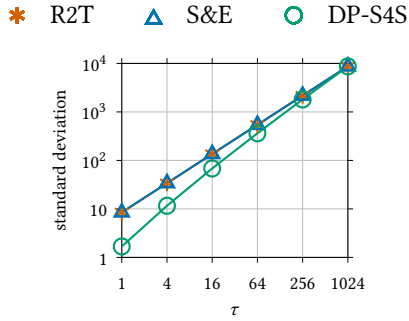
\begin{figure}[htbp]
\centering


\begin{tikzpicture}
  \node{
    \begin{tabular}{@{}c@{\quad}c@{\qquad}c@{\quad}c@{\qquad}c@{\quad}c@{}}
    \tikz{\draw[thick, draw=r2tCol]      plot[mark=asterisk, mark size=3pt, mark options={fill=white, line width=1pt}] (0,0);} & R2T &
    \tikz{\draw[thick, draw=spCol]      plot[mark=triangle*, mark size=3pt, mark options={fill=white, line width=1pt}] (0,0);} & S\&E &
    \tikz{\draw[thick, draw=dpsCol]     plot[mark=o, mark size=3pt] (0,0);} & DP-S4S
    \end{tabular}
  };
\end{tikzpicture}

\begin{tikzpicture}
\begin{axis}[
  width=\W, height=\Hgt,
  xmode=log, ymode=log,
  log ticks with fixed point, scaled ticks=false,
  minor tick style={draw=none},
  grid=major, major tick length=2pt,
  minor grid style={draw=none},  
  xtick={1,4,16,64,256,1024}, xticklabels={1,4,16,64,256,1024},
  xticklabel style={font=\scriptsize, text height=1.5ex, text depth=.25ex}, 
  ytick={0.1,1,10,100,1000,10000}, 
  yticklabels={10$^{-1}$,1,10,10$^2$,10$^3$,10$^4$},
  yticklabel style={font=\scriptsize, text width=1.5em, align=right, text height=1.5ex, text depth=.25ex},   
  tick align=outside, tick style={black},
  label style={font=\scriptsize},
  xlabel={$\tau$}, ylabel={standard deviation},
  xmin=1, xmax=1024, ymin=1, ymax=10000,
  ylabel style={yshift=-4mm}
]
  \addplot[mark=asterisk, mark size=3pt, thick, draw=r2tCol] coordinates {
( 1 ,  8.485281374238571 )
( 4 ,  33.941125496954285 )
( 16 ,  135.76450198781714 )
( 64 ,  543.0580079512686 )
( 256 ,  2172.2320318050743 )
( 1024 ,  8688.928127220297 )
}; 
  \addplot[mark=triangle, mark size=3pt, mark options={fill=white, line width=1pt}, thick, draw=spCol] coordinates { 
( 1 ,  8.485281374238571 )
( 4 ,  33.941125496954285 )
( 16 ,  135.76450198781714 )
( 64 ,  543.0580079512686 )
( 256 ,  2172.2320318050743 )
( 1024 ,  8688.928127220297 )
 }; 
  \addplot[mark=o, mark size=3pt, thick, draw=dpsCol] coordinates {
( 1 ,  1.6796356815613196 )
( 4 ,  11.620721012601106 )
( 16 ,  68.22136516414885 )
( 64 ,  362.3890510966337 )
( 256 ,  1810.495133559503 )
( 1024 ,  8688.928127220297 )
}; 
\end{axis}
\end{tikzpicture}

\vspace{-3mm}
\caption{Noise scale with $q=0.1\%$ for sampling mechanisms.}
\label{fig:ablation}
\end{figure}

While {\sf DP-S4S} requires smaller noise than {\sf R2T} in each iteration, the results must be scaled back by $1/q$ for a sampling mechanism. This is also why {\sf S\&E} has larger error despite adding the same amount of noises.
Fortunately, for our new mechanism, this is offset by the fact that a smaller $\tau^*(S)\approx q\cdot \tau^*(T)$ can be used in the sample compared to the optimal threshold for {\sf R2T}. As a concrete example, for edge counting on the Amazon dataset, the average $\tau$ that optimizes {\sf R2T} is 166, which shrinks to 30 for {\sf S\&E} and 5 for {\sf DP-S4S} with $q=1\%$.
Ignoring the amplification, the DP noise scales linearly with $\tau$ as in Figure~\ref{fig:ablation}, so the smaller DP noise in {\sf DP-S4S} after scaling back by $1/q$ matches the error of {\sf R2T}.

So far we have shown how the DP error matches and may possibly improve {\sf R2T} due to amplification. 
In addition to the DP error, our mechanism also has a sampling error, which by standard analysis is $\tilde{O}(\sqrt{n})$ for constant sampling rate. Informally, the DP error is dominating when $\tau^*(T)\gg \sqrt{n}$, i.e., the user with maximum contribution contributes to a $1/\sqrt{n}$ fraction of the records. Sampling error dominates in the opposite case.
When reducing the sample rate $q$, the mechanism becomes more efficient, where sampling error increases while DP error reduces due to better amplification. The overall error depends on the sum of both.

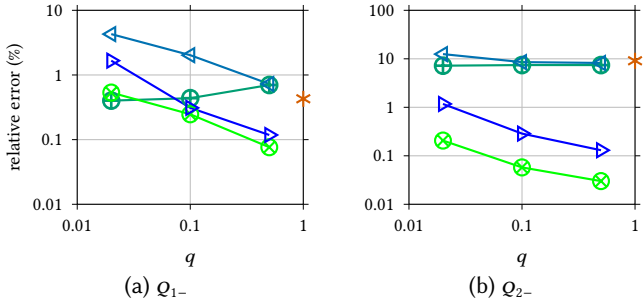
\begin{figure}[htbp]
\centering

\begin{tikzpicture}
  \node{
    \begin{tabular}{@{}c@{ }c@{\quad}c@{ }c@{\quad}c@{ }c@{\quad}c@{ }c@{\quad}c@{ }c@{}}
    \tikz{\draw[thick, draw=r2tCol]     plot[mark=asterisk, mark size=3pt] (0,0);} & R2T &
    \tikz{\draw[thick, draw=dpsCol]     plot[mark=oplus, mark size=3pt] (0,0);} & DP-S4S (DP Err.) &
    \tikz{\draw[thick, draw=green]     plot[mark=otimes, mark size=3pt] (0,0);} & DP-S4S (Sampling Err.) &
    \tikz{\draw[thick, draw=spCol]      plot[mark=triangle*, mark size=3pt, mark options={fill=white, line width=1pt, rotate=90}] (0,0);} & S\&E (DP Err.) &
    \tikz{\draw[thick, draw=blue]      plot[mark=triangle*, mark size=3pt, mark options={fill=white, line width=1pt, rotate=270}] (0,0);} & S\&E (Sampling Err.) 
    \end{tabular}
  };
\end{tikzpicture}

\begin{tabular}{@{}c@{}c@{}}

\hspace{-2mm}
\subcaptionbox{$\boldsymbol{Q}_{1-}$}[ \W ]{%
\begin{tikzpicture}
\begin{axis}[
  width=\W, height=\Hgt,
  xmode=log, ymode=log,
  log ticks with fixed point, scaled ticks=false,
  minor tick style={draw=none},
  grid=major, major tick length=2pt,
  minor grid style={draw=none},  
  xtick={0.01,0.1,1,10}, xticklabels={0.01,0.1,1,10},
  ytick={0.01,0.1,1,10}, yticklabels={0.01,0.1,1,10},
  tick align=outside, tick style={black},
  tick label style={font=\scriptsize}, label style={font=\scriptsize},
  xlabel={$q$}, ylabel={relative error (\%)},
  xmin=0.01, xmax=1, ymin=0.01, ymax=10,
  ylabel style={yshift=-4mm}
]
  \addplot[only marks, mark=asterisk, mark size=3pt, thick, draw=r2tCol] coordinates {
(1,0.4282)}; 
  \addplot[mark=oplus, mark size=3pt, mark options={fill=white, line width=1pt}, thick, draw=dpsCol] coordinates {
(0.5, 0.7001) (0.1, 0.4376) (0.02,0.4011)
 };
  \addplot[mark=otimes, mark size=3pt, thick, draw=green] coordinates {
(0.5,0.07632) (0.1, 0.2467) (0.02,0.5344)
}; 
  \addplot[mark=triangle, mark size=3pt, thick, draw=spCol, mark options={rotate=90}] coordinates {
(0.5,0.7258) (0.1, 2.0240) (0.02,4.2871)
}; 
  \addplot[mark=triangle, mark size=3pt, thick, draw=blue, mark options={rotate=270}] coordinates {
(0.5,0.1180) (0.1, 0.3103) (0.02,1.6647)
}; 
\end{axis}
\end{tikzpicture}
} & 
\subcaptionbox{$\boldsymbol{Q}_{2-}$}[ \W ]{%
\begin{tikzpicture}
\begin{axis}[
  width=\W, height=\Hgt,
  xmode=log, ymode=log,
  log ticks with fixed point, scaled ticks=false,
  minor tick style={draw=none},
  grid=major, major tick length=2pt,
  minor grid style={draw=none},  
  xtick={0.01,0.1,1}, xticklabels={0.01,0.1,1},
  ytick={0.01,0.1,1,10,100}, yticklabels={0.01,0.1,1,10,100},
  tick align=outside, tick style={black},
  tick label style={font=\scriptsize}, label style={font=\scriptsize},
  xlabel={$q$}, ylabel={},
  xmin=0.01, xmax=1, ymin=0.01, ymax=100,
  ylabel style={yshift=-4mm}
]
  \addplot[only marks, mark=asterisk, mark size=3pt, thick, draw=r2tCol] coordinates {
(1,9.113)}; 
  \addplot[mark=oplus, mark size=3pt, mark options={fill=white, line width=1pt}, thick, draw=dpsCol] coordinates {
(0.02, 7.214) (0.1,7.4348) (0.5,7.4250)
 };
  \addplot[mark=otimes, mark size=3pt, thick, draw=green] coordinates {
(0.02, 0.2061) (0.1,0.0576)(0.5, 0.03008)
}; 
  \addplot[mark=triangle, mark size=3pt, thick, draw=spCol, mark options={rotate=90}] coordinates {
  (0.02, 12.5001) (0.1, 8.5497) (0.5, 8.2112)
}; 
  \addplot[mark=triangle, mark size=3pt, thick, draw=blue, mark options={rotate=270}] coordinates {
(0.02, 1.1719) (0.1, 0.2871) (0.5,0.1302)
}; 
\end{axis}
\end{tikzpicture}
} 
\end{tabular}
\vspace{-3mm}
\caption{DP vs. Sampling Error on Amazon ($\varepsilon=1$).}
\label{fig:ablation2}
\end{figure}

We illustrate this with two contrasting examples in Figure~\ref{fig:ablation2}.
For sample $S$, sampling error refers to the difference between $f(T)$ and $f(S)/q$, whereas DP error refers to the difference between $f(S)/q$ and the output $\hat{f}(T)$. 
On the Amazon dataset, for the edge counting query $Q_{1-}$, the query size is $n=925872$ while $\tau^*=549$, so the dominating term depends on the choice of $q$. But for 2-path counting query $Q_{2-}$, we have $\tau^*=151870\gg \sqrt{n}=3122.85$, thus DP error dominates sampling error for $Q_{2-}$ for all $q$'s we consider.

There are several observations. By reducing $q$, the sampling error increases for both mechanisms. For DP error, it increases in {\sf S\&E} due to a larger scale-back factor $1/q$, but decreases in {\sf DP-S4S} following our discussion.
The amplification is better on $Q_{1-}$ due to a smaller $\Delta$.
It is also clear that both our sampling and DP error improve over {\sf S\&E}, where the major improvement is on DP error for small sampling rates.
Note that a major contribution of this work is to show sampling join tuples preserves user-level DP.

Regarding the sample rate $q$, its best value depends upon the data size (which determines sampling error) and the downward sensitivity (which impacts DP error). Under DP, it is challenging to find the best $q$ that balances these two sources of error. Further, $q$ also affects the efficiency of the linear program solver, which is also an important consideration in practice. A heuristic way is to choose a $q$ that balances  efficiency and the (non-DP) sampling error. As shown above, the DP error reduces with smaller $q$.

%% file: sections/sande.tex
\section{More on \textsf{S\&E}}\label{section:sande}

\textbf{Implementation.} The main result in \textsf{S\&E} is as follows.

\begin{lemma}[Sample-and-Explore Amplification~\cite{DBLP:journals/pacmmod/FangY24}]
Assuming each user can have at most $C$ collaborators, let $\lambda:\mathcal{T}\to\mathcal{T}$ be the sample-and-explore process, and let $\M:\mathcal{T}\to \mathcal{O}$ be an $(\varepsilon,\delta)$-user-level DP mechanism. Then, $\M\circ\lambda$ satisfies $(\varepsilon',\delta')$-user-level DP for $\varepsilon'=\ln(1+\frac{C}{m}(e^{2\varepsilon}-1))$ and $\delta'=\frac{C}{m}(e^\varepsilon+1)\delta$.
When $\delta=0$, it suffices to take $\varepsilon'=\ln\left(1+\frac{C}{m}(e^{\varepsilon}-1)\right)$.
\end{lemma}

The analysis in~\cite{DBLP:journals/pacmmod/FangY24} depends only on the event that the sampled user is not a collaborator of the witness user, which happens with probability $1-C/m$. In our implementation of \textsf{S\&E}, we extend it to a more general setting when we jointly sample $k$ out of $m$ users.
In this case, the sampled users have $kC$ collaborators in the worst case, and the probability that the witness user is not one of them is $\binom{m-kC}{k}/\binom{m}{k}$, which is $0$ for $k>m/(C+1)$.
Correspondingly, the explored instance contains tuples where the first contributing user is one of the $k$ users, thus each tuple is sampled with probability $k/m$.
We can run an $(\varepsilon,\delta)$-DP mechanism on the explored instance, and scale its results by $m/k$ for an unbiased estimator. This leads to the following result, which is the theoretical basis our our implementation of \textsf{S\&E} used in Section~\ref{sec:exp}.

\begin{corollary}\label{cor:se}
Assuming each user can have at most $C$ collaborators, let $\lambda_k:\mathcal{T}\to\mathcal{T}$ be the Sample-and-Explore process for $k$ sampled users, and let $\M:\mathcal{T}\to \mathcal{O}$ be an $(\varepsilon,\delta)$-user-level DP mechanism. Then, $\M\circ\lambda_k$ satisfies $(\varepsilon',\delta')$-user-level DP for $\varepsilon'=\ln(1+(1-P)\cdot(e^{2\varepsilon}-1))$ and $\delta'=(1-P)(e^\varepsilon+1)\delta$, where $P=0$ if $k>m/(C+1)$ and $P=\binom{m-kC}{k}/\binom{m}{k}$ otherwise. When $\delta=0$, it suffices to take $\varepsilon'=\ln(1+(1-P)\cdot(e^\varepsilon-1))$.
\end{corollary}

\vspace{3pt}
\noindent
\textbf{Sampling error analysis.} \textsf{S\&E} can also incur a large sampling error.
Note that in {\sf S\&E}, a tuple is sampled only if its {\it first} contributing user is sampled. Consider a typical case where only $m/l$ users can be the first contributing user (e.g.~a join query where the first user is always the supplier but not the customer), and each of them collaborate with the other users to contribute $l\cdot n/m$ records. Then the sample variance is
\[
\frac{m^2}{k^2}\cdot k\cdot \left(\frac{m}{l} \cdot  \frac{l^2n^2}{m^2}\cdot \frac{1}{m}\left(1-\frac{1}{m}\right)\right)\approx \frac{l\cdot n^2}{k}\,.
\]
Under the same sample rate $k=qn$, this is larger than the sampling variance of \textsf{DP-S4S} by an additional $l$ factor.

Since \textsf{DP-S4S} performs sampling on aggregation units, it automatically takes into account the correlation between users.
Moreover, our tuple-level sampling intuitively reduces the contribution of each user, which brings down the noise required for DP on the sample, and cancels with the effect of scaling back by the inverse sample rate.
This cannot be achieved by {\sf S\&E}, since the algorithm takes \textit{all} tuples associated with the sampled user, and the contribution is not reduced.


%% file: sections/ablation-ss.tex
\section{Experimental Evaluation on RDP-SS}\label{sec:ablation-ss}

In this section, we perform an empirical study on the proposed RDP-SS mechanism.
The original smooth sensitivity mechanism proposed by \citet{DBLP:conf/stoc/NissimRS07} is as follows.

\begin{lemma}[Smooth-sensitivity Mechanism~\cite{DBLP:conf/stoc/NissimRS07}]\label{lemma:ss-mechnism}
Given $F:\mathcal{I}\to\mathbb{R}^d$ and its smooth-sensitivity function $
\operatorname{SS}_F:\mathcal{I}\times\mathbb{R}\to\mathbb{R}$ as in Lemma~\ref{lemma:ss}, the smooth-sensitivity mechanism is defined as
\[
\M_{\operatorname{SS}}(\mathbf{I}):=F(\mathbf{I})+ \eta\cdot \operatorname{SS}_F(\mathbf{I},\gamma)\cdot Z\,.
\]
Its privacy guarantee is as follows.
\begin{enumerate}
\item If $\gamma=\frac{\varepsilon}{6d}$, $\eta=\frac{6}{\varepsilon}$ and $Z\sim\operatorname{Cauchy}^d(0,1)$, then $\M_{\operatorname{SS-C}}$ satisfies $(\varepsilon,0)$-DP.
\item If $\gamma =\frac{\varepsilon}{4(d+\ln(2/\delta))}$, $\eta=\frac{2}{\varepsilon}$ and $Z\sim\Lap^d(0,1)$, then $\M_{\operatorname{SS-L}}$ satisfies $(\varepsilon,\delta)$-DP.
\item If $\gamma = \frac{\varepsilon}{4(d+\ln(2/\delta))}$, $\eta=\frac{5\sqrt{2\ln(2/\delta)}}{\varepsilon}$ and $Z\sim \mathcal{N}(0, \mathbf{1}_d)$, then $\M_{\operatorname{SS-G}}$ satisfies $(\varepsilon,\delta)$-DP.
\end{enumerate}
\end{lemma}

Comparing {\sf SS-G} with our {\sf RDP-SS} in Lemma~\ref{lemma:additive-smooth}, given the same $\operatorname{SS}_F$ function, the difference is in both the smoothness requirement $\gamma$ and the noise scale factor $\eta$. In general, smaller $\gamma$ leads to a more strict smoothness requirement, and correspondingly a larger value of $\operatorname{SS}_F(\mathbf{I}, \gamma)$.

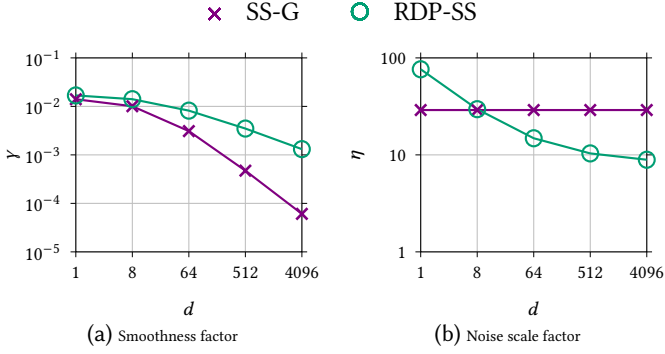
\begin{figure}[htbp]
\centering

\begin{tikzpicture}
  \node{
    \begin{tabular}{@{}c@{\quad}c@{\qquad}c@{\quad}c@{}}
    \tikz{\draw[thick, draw=violet]     plot[mark=x, mark size=3pt] (0,0);} & SS-G &
    \tikz{\draw[thick, draw=dpsCol]     plot[mark=o, mark size=3pt] (0,0);} & RDP-SS
    \end{tabular}
  };
\end{tikzpicture}

\begin{tabular}{@{}c@{}c@{}}

\hspace{-2mm}
\subcaptionbox{Smoothness factor}[ \W ]{%
\begin{tikzpicture}
\begin{axis}[
  width=\W, height=\Hgt,
  xmode=log, ymode=log,
  log ticks with fixed point, scaled ticks=false,
  minor tick style={draw=none},
  grid=major,  major tick length=2pt,
  minor grid style={draw=none},  
  xtick={1,8,64,512,4096}, xticklabels={1,8,64,512,4096},
  ytick={0.00001,0.0001,0.001,0.01,0.1}, yticklabels={$10^{-5}$,$10^{-4}$,$10^{-3}$,$10^{-2}$,$10^{-1}$},
  tick align=outside, tick style={black},
  tick label style={font=\scriptsize}, label style={font=\scriptsize},
  xlabel={$d$}, ylabel={$\gamma$},
  xmin=1, xmax=4096, ymin=0.00001, ymax=0.1,
  ylabel style={yshift=-4mm}
]
  \addplot[mark=x, mark size=3pt, mark options={line width=1pt}, thick, draw=violet] coordinates {
(1, 0.014036078355947581) (8, 0.010076077272615282) (64, 0.003093628946175472) (512, 0.0004727584811952479) (4096, 0.00006078567316594969)
 };
  \addplot[mark=o, mark size=3pt, thick, draw=dpsCol] coordinates {
(1,0.01682655077703489) (8, 0.014079456351998252) (64,0.008172202396820403) (512,0.0034905052865101837) (4096,0.0013134104845140434)
}; 
\end{axis}
\end{tikzpicture}
} & 
\subcaptionbox{Noise scale factor}[ \W ]{%
\begin{tikzpicture}
\begin{axis}[
  width=\W, height=\Hgt,
  xmode=log, ymode=log,
  log ticks with fixed point, scaled ticks=false,
  minor tick style={draw=none},
  grid=major, major tick length=2pt,
  minor grid style={draw=none},  
  xtick={1,8,64,512,4096}, xticklabels={1,8,64,512,4096},
  ytick={1,10,100}, yticklabels={1,10,100},
  tick align=outside, tick style={black},
  tick label style={font=\scriptsize}, label style={font=\scriptsize},
  xlabel={$d$}, ylabel={$\eta$},
  xmin=1, xmax=4096, ymin=1, ymax=100,
  ylabel style={yshift=-4mm}
]
  \addplot[mark=x, mark size=3pt, mark options={line width=1pt}, thick, draw=violet] coordinates {
(1, 28.992449733955098) (8, 28.992449733955098) (64, 28.992449733955098) (512, 28.992449733955098) (4096, 28.992449733955098)
 };
  \addplot[mark=o, mark size=3pt, thick, draw=dpsCol] coordinates {
(1,76.31063109723108) (8, 29.59474671755659) (64,14.83684515268077) (512,10.35459369348862) (4096, 8.909230878384308)
}; 
\end{axis}
\end{tikzpicture}
} 
\end{tabular}
\vspace{-3mm}
\caption{Comparison between {\sf SS-G} and {\sf RDP-SS}.}
\label{fig:ss}
\end{figure}

Figure~\ref{fig:ss} compares the proposed {\sf RDP-SS} against the smooth-sensitivity mechanism {\sf SS-G} implemented using Gaussian noise, under the same privacy parameters $\varepsilon=1$ and $\delta=10^{-7}$.
For high-dimensional vector queries ($d>8$), our {\sf RDP-SS} mechanism has a smaller noise scale factor $\eta$ in addition to enabling a less strict smoothness requirement (i.e., larger $\gamma$), both leading to smaller noises when compared to {\sf SS-G}.
The overall improvement in accuracy is influenced by how $\operatorname{SS}_F$ depends on $\gamma$.
For the function $G$ that we use in Section~\ref{sec:rdp-ls}, we plot its dependency on $\gamma$ and the standard deviation of the final noises in Figure~\ref{fig:ss-error}. Since $G(\mathbf{I})$ also depends on $B(\mathbf{I})$, we show two cases for $B(\mathbf{I})=1$ and $100$ respectively. 

\begin{figure}[htbp]
\centering

\begin{tikzpicture}
  \node{
    \begin{tabular}{@{}c@{ }c@{\qquad}c@{ }c@{\qquad}c@{ }c@{}}
    \tikz{\draw[thick, draw=blue]     plot[mark=|, mark size=3pt] (0,0);} & $B(\mathbf{I})=1$ &
    \tikz{\draw[thick, draw=pink]     plot[mark=x, mark size=3pt] (0,0);} & SS-G, $B(\mathbf{I})=1$ &
    \tikz{\draw[thick, draw=green]     plot[mark=o, mark size=3pt] (0,0);} & RDP-SS, $B(\mathbf{I})=1$\\
    \tikz{\draw[thick, draw=cyan]     plot[mark=+, mark size=3pt] (0,0);} & $B(\mathbf{I})=100$ &
    \tikz{\draw[thick, draw=violet]     plot[mark=x, mark size=3pt] (0,0);} & SS-G, $B(\mathbf{I})=100$ &
    \tikz{\draw[thick, draw=dpsCol]     plot[mark=o, mark size=3pt] (0,0);} & RDP-SS , $B(\mathbf{I})=100$
    \end{tabular}
  };
\end{tikzpicture}

\begin{tabular}{@{}c@{}c@{}}

\hspace{-2mm}
\subcaptionbox{Value of Smooth Sensitivity $G$}[ \W ]{%
\begin{tikzpicture}
\begin{axis}[
  width=\W, height=\Hgt,
  xmode=log, ymode=log,
  log ticks with fixed point, scaled ticks=false,
  minor tick style={draw=none},
  grid=major, major tick length=2pt,
  minor grid style={draw=none},  
  xtick={0.00001,0.0001,0.001,0.01,0.1}, xticklabels={$10^{-5}$,$10^{-4}$,$10^{-3}$,$10^{-2}$,$10^{-1}$},
  ytick={1,10,100,1000,10000}, yticklabels={1,10,$10^2$,$10^3$,$10^4$},
  tick align=outside, tick style={black},
  tick label style={font=\scriptsize}, label style={font=\scriptsize},
  xlabel={$\gamma$}, ylabel={$G$},
  xmin=0.00001, xmax=0.1, ymin=1, ymax=10000,
  ylabel style={yshift=-4mm}
]
  \addplot[mark=|, mark size=4pt, mark options={line width=1pt}, thick, draw=blue] coordinates {
(0.05, 7.734820469090024) (0.01, 37.15766910220457) (0.002, 184.30796815170942) (0.0004, 920.0665559554766) (0.00008, 4598.86090879977)
 };
  \addplot[mark=+, mark size=4pt, mark options={line width=1pt}, thick, draw=cyan] coordinates {
(0.05, 100) (0.01, 100) (0.002, 224.6644820586108) (0.0004, 957.23221493778) (0.00008, 4635.428503727446)
 }; 
\end{axis}
\end{tikzpicture}
} &
\subcaptionbox{Final Noise}[ \W ]{%
\begin{tikzpicture}
\begin{axis}[
  width=\W, height=\Hgt,
  xmode=log, ymode=log,
  log ticks with fixed point, scaled ticks=false,
  minor tick style={draw=none},
  grid=major, major tick length=2pt,
  minor grid style={draw=none},  
  xtick={1,8,64,512,4096}, xticklabels={1,8,64,512,4096},
  ytick={100,1000,10000,100000,1000000}, yticklabels={$10^2$,$10^3$,$10^4$,$10^5$,$10^6$},
  tick align=outside, tick style={black},
  tick label style={font=\scriptsize}, label style={font=\scriptsize},
  xlabel={$d$}, ylabel={standard deviation},
  xmin=1, xmax=4096, ymin=100, ymax=1000000,
  ylabel style={yshift=-4mm}
]
  \addplot[mark=x, mark size=3.5pt, thick, draw=pink] coordinates {
(1,770.6157219568884) (8, 1069.2360703275247) (64,3458.323609240973) (512,22571.291819835704) (4096, 175475.1406481584)
}; 
  \addplot[mark=x, mark size=3.5pt, thick, draw=violet] coordinates {
(1, 2899.24497339551) (8, 2899.24497339551) (64, 4697.607005687947) (512, 23652.809787816088) (4096, 176534.29541315642)
 };
  \addplot[mark=o, mark size=3.5pt, thick, draw=green] coordinates {
(1,1696.6479006291231) (8, 784.2397696444912) (64,673.3721705287678) (512,1095.1296756935892) (4096, 2498.708572078539)
}; 
  \addplot[mark=o, mark size=3.5pt, thick, draw=dpsCol] coordinates {
(1,7631.063109723109) (8, 2959.4746717556586) (64,1512.238400126843) (512,1547.1782762194896) (4096, 2845.6788912317165)
};
\end{axis}
\end{tikzpicture}
} 
\end{tabular}
\vspace{-3mm}
\caption{Error of {\sf SS-G} and {\sf RDP-SS} for function $G$.}
\label{fig:ss-g}
\end{figure}
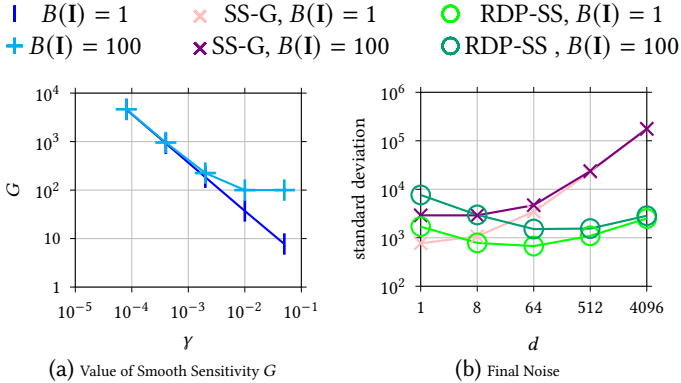

Recall that $G(\mathbf{I})$ is larger for smaller $\gamma$, so the noise of {\sf SS-G} increases with $d$ since $\eta$ is a constant. But for {\sf RDP-SS}, $\eta=\frac{\sqrt{\alpha/\rho}}{1-\alpha+\alpha e^{-2\gamma}}$ decreases\footnote{If we treat $\alpha$ and $\rho$ as constant, this is obvious. In the above settings, the optimal $\alpha$ ranges between $17$ and $32$, whose change is dominated by the decrease of $\gamma$.} when reducing $\gamma$, therefore the error first decreases due to a smaller $\eta$, and then increases due to a larger smooth sensitivity $G(\mathbf{I})$ for larger $d$. Again, for high dimensions ($d>8$), the {\sf RDP-SS} mechanism improves over {\sf SS}. For our experiments in Section~\ref{sec:exp}, we test for $d=10$ queries, where the improvement is around 1.5x in accuracy. For larger $d$, e.g. $d=4096$, the accuracy improvement can be as large as 70x. 

\begin{figure}[htbp]
\centering

\begin{tikzpicture}
  \node{
    \begin{tabular}{@{}c@{ }c@{\quad}c@{ }c@{\quad}c@{ }c@{\quad}c@{ }c@{\quad}c@{ }c@{}}
    \tikz{\draw[thick, draw=pmsjaCol]     plot[mark=pentagon, mark size=4pt] (0,0);} & PMSJA &
    \tikz{\draw[thick, draw=dpsCol]     plot[mark=oplus, mark size=3.5pt] (0,0);} & DP-S4S (DP Err.) &
    \tikz{\draw[thick, draw=green]     plot[mark=otimes, mark size=3.5pt] (0,0);} & DP-S4S (Sampling Err.) &   
    \tikz{\draw[thick, draw=spCol]      plot[mark=triangle*, mark size=4pt, mark options={fill=white, line width=1pt, rotate=90}] (0,0);} & S\&E (DP Err.) &
    \tikz{\draw[thick, draw=blue]      plot[mark=triangle*, mark size=4pt, mark options={fill=white, line width=1pt, rotate=270}] (0,0);} & S\&E (Sampling Err.) 
    \end{tabular}
  };
\end{tikzpicture}

\begin{tabular}{@{}c@{}c@{}}

\hspace{-2mm}
\subcaptionbox{$\boldsymbol{Q}_{1-}$}[ \W ]{%
\begin{tikzpicture}
\begin{axis}[
  width=\W, height=\Hgt,
  xmode=log, ymode=log,
  log ticks with fixed point, scaled ticks=false,
  minor tick style={draw=none},
  grid=major,
  minor grid style={draw=none},  
  xtick={0.01,0.1,1,10}, xticklabels={0.01,0.1,1,10,1000},
  ytick={0.01,0.1,1,10,100,1000}, yticklabels={0.01,0.1,1,10,100,1000},
  tick align=outside, tick style={black},
  tick label style={font=\scriptsize}, label style={font=\scriptsize},
  xlabel={$q$}, ylabel={relative error (\%)},
  xmin=0.01, xmax=1, ymin=0.1, ymax=1000,
  ylabel style={yshift=-4mm}
]
  \addplot[only marks, mark=pentagon, mark size=3pt, thick, draw=pmsjaCol] coordinates {
(1,17.07)}; 
  \addplot[mark=oplus, mark size=3pt, mark options={fill=white, line width=1pt}, thick, draw=dpsCol] coordinates {
(0.5, 18.24105517032763) (0.1, 67.55095641231468) (0.02,103.79818863945134)
 };
  \addplot[mark=otimes, mark size=3pt, thick, draw=green] coordinates {
(0.5,0.8969204891084145) (0.1, 2.7270061765186813) (0.02,7.0672274347102)
}; 
  \addplot[mark=triangle, mark size=3pt, thick, draw=spCol, mark options={rotate=90}] coordinates {
(0.5, 89.20622130280674) (0.1, 260.85139969438995) (0.02,930.6289361076966)
}; 
  \addplot[mark=triangle, mark size=3pt, thick, draw=blue, mark options={rotate=270}] coordinates {
(0.5,1.1081502182481646) (0.1, 3.1343125815493678) (0.02,8.333563210763254)
}; 
\end{axis}
\end{tikzpicture}
} & 
\subcaptionbox{$\boldsymbol{Q}_{2-}$}[ \W ]{%
\begin{tikzpicture}
\begin{axis}[
  width=\W, height=\Hgt,
  xmode=log, ymode=log,
  log ticks with fixed point, scaled ticks=false,
  minor tick style={draw=none},
  grid=major,
  minor grid style={draw=none},  
  xtick={0.01,0.1,1}, xticklabels={0.01,0.1,1},
  ytick={0.01,0.1,1,10,100,1000,10000}, yticklabels={0.01,0.1,1,10,100,1000,10000},
  tick align=outside, tick style={black},
  tick label style={font=\scriptsize}, label style={font=\scriptsize},
  xlabel={$q$}, ylabel={},
  xmin=0.01, xmax=1, ymin=0.1, ymax=10000,
  ylabel style={yshift=-4mm}
]
  \addplot[only marks, mark=pentagon, mark size=3pt, thick, draw=pmsjaCol] coordinates {
(1,28.67)}; 
  \addplot[mark=oplus, mark size=3pt, mark options={fill=white, line width=1pt}, thick, draw=dpsCol] coordinates {
(0.02, 88.59257562087096) (0.1,69.16302550209338) (0.5,45.79306440181432)
 };
  \addplot[mark=otimes, mark size=3pt, thick, draw=green] coordinates {
(0.02,5.73408696367044) (0.1,2.2284667978796238)(0.5,0.7883535100097815)
}; 
  \addplot[mark=triangle, mark size=3pt, thick, draw=spCol, mark options={rotate=90}] coordinates {
  (0.02, 1616.9207801567434) (0.1,  504.45683027109567) (0.5, 170.51013823447747)
}; 
  \addplot[mark=triangle, mark size=3pt, thick, draw=blue, mark options={rotate=270}] coordinates {
(0.02, 12.179365587314997) (0.1, 5.178300077777543) (0.5,1.8034427964841202)
}; 
\end{axis}
\end{tikzpicture}
} 
\end{tabular}
\vspace{-3mm}
\caption{DP vs. Sampling Error on C2A $(\varepsilon=1,\delta=10^{-7}$).}
\label{fig:ss-error}
\end{figure}
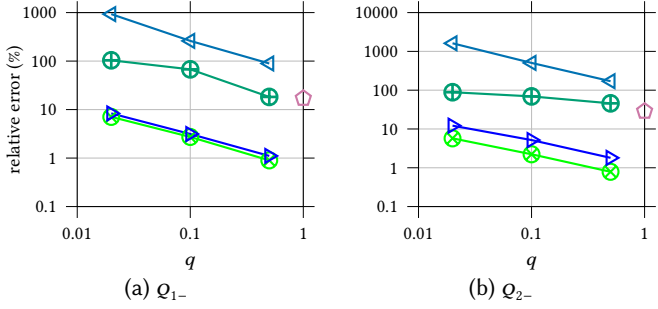

Similar to Appendix~\ref{sec:ablation}, we also study the relative contribution of sampling error vs error due to DP noise in Algorithm~\ref{algo:sample-pmsja}, using the C2A dataset as an example. For vector queries, the error is usually dominated by the injected noise required by DP, which increases with the dimensionality $d$. Therefore, although {\sf DP-S4S} achieve a smaller sampling error owing to independence between samples, this is less significant in our parameter ranges, when DP is enforced. Regarding DP noise, unlike the case in Figure~\ref{fig:ablation2}, here, both mechanisms have increasing DP noise when reducing sample rate $q$. This is because our privacy guarantee in Lemma~\ref{lemma:s4s-v} only ensures the mechanism on the sample is at least as private, but not strong enough to allow using larger privacy parameters as we did in Algorithm~\ref{algo:sample-r2t}.
Nevertheless, the DP error from {\sf DP-S4S} is still lower than that from {\sf S\&E}, which is mainly because our sampling strategy is more effective for reducing the maximum contribution of a user to the query results, which affects the DP error as we analyzed in Appendix~\ref{sec:ablation}.
A minor improvement is also from adopting {\sf RDP-SS}, which reduces the noise by around 1.5x in our setting with $d=10$. 
Overall, {\sf DP-S4S} consistently outperforms {\sf S\&E}.

Given that both error terms increase with respect to decreasing $q$, in practice, one may adopt the largest $q$ that can be handled by the available computation resources.

%% file: sections/vector-error.tex
\section{Error Analysis of DP-S4S for Vector Queries}\label{section:vector-error}

\setlength{\W}{0.25\linewidth}
\setlength{\Hgt}{0.9\W}

As discussed in Section~\ref{sec:s4s-v}, the amplification effect in Lemma~\ref{lemma:s4s-v} for vector queries is not a strong result: it only guarantees that the sampled mechanism is as least as private as the mechanism without sampling, but does not provide a better privacy guarantee for performance improvement. We first recall the error of {\sf PMSJA}.

\begin{lemma}
Let $\operatorname{Err}_{\mathcal{M}_{\LS}}(B(\mathbf{I});\varepsilon,\delta,\beta)$ be the error of a local-sensitivity mechanism $\mathcal{M}_{\LS}$ for function $F$ under $(\varepsilon,\delta)$-DP, given bound $B(\mathbf{I})\geq \LS_F(\mathbf{I})$ and $|B(\mathbf{I})-B(\mathbf{I}')|$ for neighboring datasets. 
The noise that {\sf PMSJA} adds to each query is
\[
\sigma_{\sf PMSJA}=\tilde{O}\left(\tau^*(\mathbf{I})\cdot \operatorname{Err}\left(\frac{1}{\varepsilon};\varepsilon,\delta,\beta\right)\right)
\]
for $\tau^*(\mathbf{I}):=\max_u \|\boldsymbol{x}_u\|_2$, where $\boldsymbol{x}_u=(x_{1,u},\dots,x_{|F|,u})$ is the contribution of user $u$ to all queries. 
Specially, when $\mathcal{M}_{\LS}$ is the privately bounding local sensitivity mechanism~\cite{DBLP:books/sp/17/Vadhan17}, its error  $\operatorname{Err}_{\mathcal{M}_{\LS}}(B(\mathbf{I}),\varepsilon,\delta,\beta)=\tilde{O}(\frac{B(\mathbf{I})}{\varepsilon})$, and the noise of {\sf PMSJA} is $\tilde{O}(\frac{\tau^*(\mathbf{I})}{\varepsilon^2})$.
\end{lemma}

Following similar analysis as in {\sf PMSJA}, the DP noise in {\sf DP-S4S} for vector queries after rescale can be bounded as
\[
\sigma_{\sf DP-S4S} = \tilde{O}\left(\frac{1}{q}\cdot \tau^*(S)\cdot \operatorname{Err}_{\mathcal{M}_{\sf RDP-SS}}\left(\frac{1}{\varepsilon};\varepsilon,\delta,\beta\right)\right)
\]
where $\tau^*(S)$ is the maximum contribution of a user in the sample; and $\operatorname{Err}(\cdot)$ is the error from our {\sf RDP-SS} mechanism.
In addition, there is a sampling error, which is $\tilde{O}(\sqrt{\frac{n_f}{q}}+\frac{1}{q})$ for query $f$ where $n_f$ is the query size.

However, it is difficult to provide an exact error bound. The major challenge is in bounding the error of smooth-sensitivity mechanisms which we discuss later.
Nevertheless, for vector queries, there is still a reduction in the contribution of each user to the sample, similar to the case of scalar queries.
Let $y_{f,u}$ be the records sampled for query $f$ from user $u$, then $y_{f,u}\sim \operatorname{Bin}(x_{f,u},q)$.
By triangle inequality, $\|y_u\|_2\leq q\|x_u\|_2+\|z_u\|_2$, where $z_{f,u}:=y_{f,u}-qx_{f,u}$.
By Hoeffding's inequality,
\[
\Pr\left[|y_{f,u}-qx_{f,u}|>\sqrt{\frac{x_{f,u}}{2}\ln\frac{2d}{\beta}}\right]\leq \frac{\beta}{d}
\]
Thus, with probability $1-\beta$, we can bound $\|z_u\|_2$ as 
\[
\|z_u\|_2=\sqrt{\sum_{f} z_{f,u}^2}=\tilde{O}\left(\sqrt{\sum_f  x_{f,u}}\right)=\tilde{O}(\sqrt{\|x_u\|_1})\,.
\]
With a union bound, conditioned on all the $\|z_u\|_2$'s are bounded,
\begin{align*}
\tau^*(S):&=\max_u \|y_u\|_2\leq \max_u q\|x_u\|_2 + \max_u \|z_u\|_2\\
&=\tilde{O}(q\cdot \tau^*(\mathbf{I})+d^\frac{1}{4}\cdot \tau^*(\mathbf{I}))\,,
\end{align*}
where we used $\|x_u\|_1\leq \sqrt{d}\cdot \|x_u\|_2$.
Thus compared to {\sf PMSJA}, the DP noise of {\sf DP-S4S} has an additional overhead factor $(1+\frac{d^\frac{1}{4}}{q})$, which still increases when reducing $q$, as Figure~\ref{fig:ss-error} shows.

Finally, we discuss the major challenge: bounding the error of {\sf RDP-SS} in Lemma~\ref{lemma:additive-smooth}.
Unlike the privately bounding local sensitivity mechanism \cite{DBLP:books/sp/17/Vadhan17}, it is hard to provide an exact error bound, since the value of $\gamma$ depends on solving a polynomial of degree $\alpha$. But as has been discussed in the empirical study in Appendix~\ref{sec:ablation-ss}, the {\sf RDP-SS} mechanism achieves better smoothness requirement $\gamma$ and smaller error compared to {\sf SS}.

%% file: sections/cdp.tex
\section{Connection to zCDP}\label{sec:cdp}
In this paper, we adopt R\'enyi-DP (Definition~\ref{def:rdp}) for Gaussian-based mechanisms besides the standard DP notion. A closely related definition is Zero-Concentrated Differential Privacy~\cite{DBLP:conf/tcc/BunS16}, formally defined as follows.

\begin{definition}[zCDP~\cite{DBLP:conf/tcc/BunS16}]
For $\rho>0$, a randomized mechanism $\M:\mathcal{I}\to\mathcal{O}$ is $\rho$-zCDP, if for any pair of neighboring datasets $\mathbf{I}\sim \mathbf{I}'$ and any $\alpha>1$, it holds that
\[
\mathrm{D}_\alpha(\M(\mathbf{I}) \| \M(\mathbf{I}'))\leq \rho\alpha\,,
\]
where $\mathrm{D}_\alpha(P\|Q):=\frac{1}{\alpha-1}\ln\int_{z\in \operatorname{supp}(Q)} P(z)^\alpha Q(z)^{1-\alpha}\,\mathrm{d}z$ is the R\'enyi divergence of order $\alpha$.
\end{definition}

It is easy to see that any $\rho$-zCDP mechanism is immediately $(\alpha,\rho\alpha)$-RDP for any $\alpha>1$; and any $(\alpha,\rho(\alpha))$-RDP mechanism is $\sup_{\alpha>1}\frac{\rho(\alpha)}{\alpha}$-zCDP. Besides, $\varepsilon$-DP implies $\frac{1}{2}\varepsilon^2$-zCDP; and $\rho$-zCDP implies $(\rho+2\sqrt{\rho\ln(1/\delta)},\delta)$-DP for any $\delta>0$.
For the standard Gaussian mechanism, the two definitions are equivalent when reduced to approximate-DP, as
\[
\varepsilon_{\text{RDP}}=\min_\alpha \frac{\alpha}{2\sigma^2}+\frac{\ln(1/\delta)}{\alpha-1} = \frac{1}{2\sigma^2}+\sqrt{\frac{2\ln(1/\delta)}{\sigma}}=\varepsilon_{\text{zCDP}}
\]
is optimized exactly when $\alpha=1+\sigma\cdot \sqrt{2\ln(1/\delta)}$.

Since our proof for RDP works by bounding the R\'enyi divergence for any $\alpha$, it can be naturally extended to zCDP through the general reduction.
However, since $\alpha$ appears in exponents in Lemma~\ref{lemma:sample-truncate}, the resulting guarantee does not have a trivial representation in closed form.

%% file: sections/data.tex
\section{Data Statistics and Additional Experiments} \label{section:data}

\input{fig/a2q}
\input{fig/roadnetca}

Table~\ref{tab:graph} and~\ref{tab:tpch} summarize statistics of the datasets used in our experiments. Figure~\ref{fig:a2q} presents experiment results varying $q$ on the A2Q dataset in~\cite{DBLP:journals/pacmmod/0007S023}, where the result is similar to Figure~\ref{fig:c2a} in the main text.

\begin{table}[t]
    \centering
    \caption{Graph Statistics in the experiments.}     \label{tab:graph}
    
    \begin{tabular}{ccccc}
    \toprule
       \textbf{Dataset}  & \textbf{Vertices} & \textbf{Edges} & \textbf{Max degree} & \textbf{Degree bound} \\
    \midrule
       Deezer & 143,884 & 846,915 & 420 & 1024 \\
       Amazon  & 334,863 & 925,872 & 549 & 1024  \\
       RoadnetCA  & 1,965,206 & 2,766,607 & 12 & 16 \\
    \midrule
       A2Q  & 1,192,133 & 1,468,092 & 23 & 1024 \\
       C2A  & 679,729 & 1,425,352 & 40 & 1024 \\
    \bottomrule
    \end{tabular}
    \vspace{1mm}
\end{table}

\begin{table}[t]
    \centering
    \caption{Table Sizes in TPC-H SF4.}     \label{tab:tpch}
    
    \vspace{-2mm}
    \begin{tabular}{lc}
    \toprule
       \textbf{Table} & \textbf{Number of Records} \\
    \midrule
       Region & 5\\
       Nation & 25\\
       Supplier & 40,000 \\
       Customer & 600,000 \\
       Part & 800,000 \\
       PartSupp & 3,200,000 \\
       Orders & 6,000,000 \\
       Lineitem & 23,996,604\\
    \bottomrule
    \end{tabular}
\end{table}

In addition, we run an experiment to show the effect of data skewness on the mechanisms.
The datasets Deezer and Amazon we consider in the main text are both highly skewed, with a maximum degree over 400 but a majority of vertices having degree at most 10.
We include a less skewed road network dataset RoadnetCA for completeness, where $99\%$ of the vertices have degrees between 1 and 4.
Recall that {\sf DP-S4S} sample each join result independently, therefore its sampling error is independent with respect to the skewness of the data.
Contrarily, for {\sf S\&E}, less skewed data is favorable for its sampling error since the variance in the sample size is smaller.
As shown in Figure~\ref{fig:roadnetca-q}, since we use a smaller degree bound $D=16$ for this dataset, the difference between {\sf DP-S4S} and {\sf S\&E} is smaller. 
Nevertheless, we still observe improvements on $Q_{2-}$, $Q_\triangle$ and $Q_{\square}$.

%% file: fig/a2q.tex
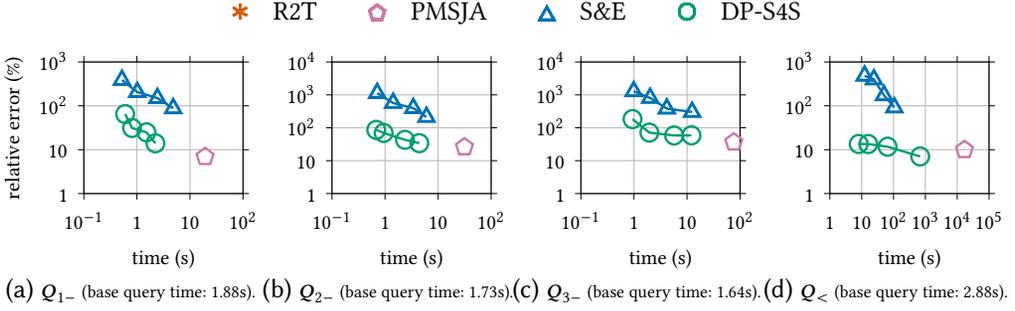
\begin{figure}[t]
\centering

\begin{tikzpicture}
  \node{
    \begin{tabular}{@{}c@{\quad}c@{\qquad}c@{\quad}c@{\qquad}c@{\quad}c@{\qquad}c@{\quad}c@{}}
    \tikz{\draw[thick, draw=r2tCol]     plot[mark=asterisk, mark size=3pt] (0,0);} & R2T &
    \tikz{\draw[thick, draw=pmsjaCol]      plot[mark=pentagon, mark size=3pt, mark options={fill=white, line width=1pt}] (0,0);} & PMSJA &
    \tikz{\draw[thick, draw=spCol]      plot[mark=triangle*, mark size=3pt, mark options={fill=white, line width=1pt}] (0,0);} & S\&E &
    \tikz{\draw[thick, draw=dpsCol]     plot[mark=o, mark size=3pt] (0,0);} & DP-S4S
    \end{tabular}
  };
\end{tikzpicture}

\makebox[\textwidth][c]{%
   \resizebox{\linewidth}{!}{%
\begin{tabular}{@{}c@{}c@{}c@{}c@{}}

\hspace{-2mm}
\subcaptionbox{$\boldsymbol{Q}_{1-}$ (base query time: 1.88s).}[ \W ]{%
\begin{tikzpicture}
\begin{axis}[
  width=\W, height=\Hgt,
  xmode=log, ymode=log,
  log ticks with fixed point, scaled ticks=false,
  minor tick style={draw=none},
  grid=major, major tick length=2pt,
  minor grid style={draw=none},  
  xtick={0.1,1,10,100,1000,10000,100000}, xticklabels={10$^{-1}$,1,10,10$^2$,10$^3$,10$^4$,10$^5$},
  xticklabel style={font=\scriptsize, text height=1.5ex, text depth=.25ex}, 
  ytick={0.1,1,10,100,1000,10000}, 
  yticklabels={10$^{-1}$,1,10,10$^2$,10$^3$,10$^4$},
  yticklabel style={font=\scriptsize, text width=1.5em, align=right, text height=1.5ex, text depth=.25ex},   
  tick align=outside, tick style={black},
  label style={font=\scriptsize},
  xlabel={time (s)}, ylabel={relative error (\%)},
  xmin=0.1, xmax=100, ymin=1, ymax=1000,
  ylabel style={yshift=-4mm}
]
  \addplot[only marks, mark=pentagon, mark size=3pt, thick, draw=pmsjaCol] coordinates {
(19.20,6.91)
}; 
  \addplot[mark=triangle, mark size=3pt, mark options={fill=white, line width=1pt}, thick, draw=spCol] coordinates { 
(0.5214,388.65) (1.015,205.71) (2.432,150.169) (4.824,87.84)
 }; 
  \addplot[mark=o, mark size=3pt, thick, draw=dpsCol] coordinates {
  (0.600,64.81) (0.810,31.29) (1.5143, 24.85) (2.23,14.00)
}; 
\end{axis}
\end{tikzpicture}
} & \hspace{-5mm}
\subcaptionbox{$\boldsymbol{Q}_{2-}$ (base query time: 1.73s).}[ \W ]{%
\begin{tikzpicture}
\begin{axis}[
  width=\W, height=\Hgt,
  xmode=log, ymode=log,
  log ticks with fixed point, scaled ticks=false,
  minor tick style={draw=none},
  grid=major, major tick length=2pt,
  minor grid style={draw=none}, 
  xtick={0.1,1,10,100,1000,10000,100000}, xticklabels={10$^{-1}$,1,10,10$^2$,10$^3$,10$^4$,10$^5$},
  xticklabel style={font=\scriptsize, text height=1.5ex, text depth=.25ex}, 
  ytick={0.1,1,10,100,1000,10000}, 
  yticklabels={10$^{-1}$,1,10,10$^2$,10$^3$,10$^4$},
  yticklabel style={font=\scriptsize, text width=1.5em, align=right, text height=1.5ex, text depth=.25ex},
  tick align=outside, tick style={black},
  label style={font=\scriptsize},
  xlabel={time (s)}, ylabel={},
  xmin=0.1, xmax=100, ymin=1, ymax=10000
]
  \addplot[only marks, mark=pentagon, mark size=3pt, thick, draw=pmsjaCol] coordinates {
(31.37,26.00)
}; 
  \addplot[mark=triangle, mark size=3pt, mark options={fill=white, line width=1pt}, thick, draw=spCol] coordinates {
(0.711,1131.60) (1.421,567.45) (3.405,393.99) (6.0475,213.05)
}; 
  \addplot[mark=o, mark size=3pt, thick, draw=dpsCol] coordinates {
  (0.682,86.08) (0.942, 69.68) (2.376, 42.96) (4.428,34.09)
  }; 
\end{axis}
\end{tikzpicture}
} & \hspace{-6mm}
\subcaptionbox{$\boldsymbol{Q}_{3-}$ (base query time: 1.64s).}[ \W ]{%
\begin{tikzpicture}
\begin{axis}[
  width=\W, height=\Hgt,
  xmode=log, ymode=log,
  log ticks with fixed point, scaled ticks=false,
  minor tick style={draw=none},
  grid=major,major tick length=2pt,
  minor grid style={draw=none}, 
  xtick={0.1,1,10,100,1000,10000,100000}, xticklabels={10$^{-1}$,1,10,10$^2$,10$^3$,10$^4$,10$^5$},
  xticklabel style={font=\scriptsize, text height=1.5ex, text depth=.25ex}, 
  ytick={0.1,1,10,100,1000,10000}, 
  yticklabels={10$^{-1}$,1,10,10$^2$,10$^3$,10$^4$},
  yticklabel style={font=\scriptsize, text width=1.5em, align=right, text height=1.5ex, text depth=.25ex},
  tick align=outside, tick style={black},
  label style={font=\scriptsize},
  xlabel={time (s)}, ylabel={},
  xmin=0.1, xmax=100, ymin=1, ymax=10000
]
  \addplot[only marks, mark=pentagon, mark size=3pt, thick, draw=pmsjaCol] coordinates {
(75.98,36.45)
}; 
  \addplot[mark=triangle, mark size=3pt, mark options={fill=white, line width=1pt}, thick, draw=spCol] coordinates {
(0.99,1287.86) (2.017,806.09) (4.158,381.237) (12.46,301.861)
}; 
  \addplot[mark=o, mark size=3pt, thick, draw=dpsCol] coordinates {
(0.939,183.43) (1.952, 71.28) (5.70, 58.34) (12.12, 58.91) 
}; 
\end{axis}
\end{tikzpicture}
} & \hspace{-6mm}
\subcaptionbox{$\boldsymbol{Q}_{<}$ (base query time: 2.88s).}[ \W ]{%
\begin{tikzpicture}
\begin{axis}[
  width=\W, height=\Hgt,
  xmode=log, ymode=log,
  log ticks with fixed point, scaled ticks=false,
  minor tick style={draw=none},
  grid=major, major tick length=2pt,
  minor grid style={draw=none}, 
  xtick={0.1,1,10,100,1000,10000,100000}, xticklabels={10$^{-1}$,1,10,10$^2$,10$^3$,10$^4$,10$^5$},
  xticklabel style={font=\scriptsize, text height=1.5ex, text depth=.25ex}, 
  ytick={0.1,1,10,100,1000,10000}, 
  yticklabels={10$^{-1}$,1,10,10$^2$,10$^3$,10$^4$},
  yticklabel style={font=\scriptsize, text width=1.5em, align=right, text height=1.5ex, text depth=.25ex},
  tick align=outside, tick style={black},
  label style={font=\scriptsize},
  xlabel={time (s)}, ylabel={},
  xmin=1, xmax=100000, ymin=1, ymax=1000
]
  \addplot[only marks, mark=pentagon, mark size=3pt, thick, draw=pmsjaCol] coordinates {
(16937.83,9.967)
}; 
  \addplot[mark=triangle, mark size=3pt, mark options={fill=white, line width=1pt}, thick, draw=spCol] coordinates {
(12.33,484.39) (24.16,399.94) (50.06,174.20) (106.05,91.81)
}; 
  \addplot[mark=o, mark size=3pt, thick, draw=dpsCol] coordinates {
(8.04,13.53)(16.04,13.46)(65.07,11.71)(686.63,7.023)
}; 
\end{axis}
\end{tikzpicture}
} 
\end{tabular}
}}
\vspace{-3mm}
\caption{Query error vs.\ time on A2Q with $\varepsilon = 4$, $\delta=10^{-7}$. (Sampling rates $q\in\left\{ \frac{1}{32},\tfrac{1}{16}, \tfrac{1}{8}, \tfrac{1}{4}\right\}$)}
\label{fig:a2q}
\end{figure}

%% file: fig/roadnetca.tex
\begin{figure}[t]
\centering

\makebox[\textwidth][c]{%
   \resizebox{\linewidth}{!}{%
\begin{tabular}{@{}c@{}c@{}c@{}c@{}}

\hspace{-2mm}
\subcaptionbox{$\boldsymbol{Q}_{1-}$ (base query time: 0.74s).}[ \W ]{%
\begin{tikzpicture}
\begin{axis}[
  width=\W, height=\Hgt,
  xmode=log, ymode=log,
  log ticks with fixed point, scaled ticks=false,
  minor tick style={draw=none},
  grid=major, major tick length=2pt,
  minor grid style={draw=none},  
  xtick={0.1,1,10,100}, xticklabels={0.1,1,10,100},
  ytick={0.01,0.1,1,10,100}, yticklabels={0.01,0.1,1,10,100},
  tick align=outside, tick style={black},
  tick label style={font=\scriptsize}, label style={font=\scriptsize},
  xlabel={time (s)}, ylabel={relative error (\%)},
  xmin=0.1, xmax=100, ymin=0.01, ymax=1,
  ylabel style={yshift=-4mm}
]
  \addplot[only marks, mark=asterisk, mark size=3pt, thick, draw=r2tCol] coordinates {
( 54.545363227526344,0.020765892987966907)}; 
  \addplot[mark=triangle, mark size=3pt, mark options={fill=white, line width=1pt}, thick, draw=spCol] coordinates { 
( 1.0796512365341187, 0.41212942718068585) 
(2.5249050855636597,0.18603370273861927)
(6.287203669548035, 0.1292764553157737)
(11.088283697764078,0.07752370376295699)
(23.226762255032856,0.05649072089128087)
 }; 
  \addplot[mark=o, mark size=3pt, thick, draw=dpsCol] coordinates {
( 19.79918082555135,0.08674804267339216) 
(9.105802416801453, 0.12609734049243015)
(4.64126396179199211, 0.22565298227213929)
(2.154212315877279,0.18776164622027997)
(0.8605647881825765,0.37826747627863816)
}; 
\end{axis}
\end{tikzpicture}
} & \hspace{-5mm}
\subcaptionbox{$\boldsymbol{Q}_{2-}$ (base query time: 7.36s).}[ \W ]{%
\begin{tikzpicture}
\begin{axis}[
  width=\W, height=\Hgt,
  xmode=log, ymode=log,
  log ticks with fixed point, scaled ticks=false,
  minor tick style={draw=none},
  grid=major, major tick length=2pt,
  minor grid style={draw=none}, 
  xtick={1,10,100,1000}, xticklabels={1,10,100,1000},
  ytick={0.1,1,10,100}, yticklabels={0.1,1,10,100},
  tick align=outside, tick style={black},
  tick label style={font=\scriptsize}, label style={font=\scriptsize},
  xlabel={time (s)}, ylabel={},
  xmin=1, xmax=1000, ymin=0.01, ymax=10
]
  \addplot[only marks, mark=asterisk, mark size=3pt, thick, draw=r2tCol] coordinates {
(255.84634550412497,  0.07651950954840248)
}; 
  \addplot[mark=triangle, mark size=3pt, mark options={fill=white, line width=1pt}, thick, draw=spCol] coordinates {
(4.194633960723877,1.1582687835641066) (11.821539680163065,0.5478162294580766)  (22.463021874427795, 0.3039231695580492) (42.04878409703573, 0.1773296420788949) ( 99.50857937335968, 0.12252154549388108)
}; 
  \addplot[mark=o, mark size=3pt, thick, draw=dpsCol] coordinates {
(2.4646236896514893,0.3543198819973329) (10.53205426534017, 0.15542274260263697)
(21.639310240745544,0.12876525175418022) ( 42.236300667126976, 0.14257251491273557)
(80.30322341124217,0.05571743745099758)
 }; 
\end{axis}
\end{tikzpicture}
} & \hspace{-6mm}
\subcaptionbox{$\boldsymbol{Q}_{\triangle}$ (base query time: 2.98s).}[ \W ]{%
\begin{tikzpicture}
\begin{axis}[
  width=\W, height=\Hgt,
  xmode=log, ymode=log,
  log ticks with fixed point, scaled ticks=false,
  minor tick style={draw=none},
  grid=major, major tick length=2pt,
  minor grid style={draw=none}, 
  xtick={0.01,0.1,1,10,100}, xticklabels={0.01,0.1,1,10,100},
  ytick={0.1,1,10,100}, yticklabels={0.1,1,10,100},
  tick align=outside, tick style={black},
  tick label style={font=\scriptsize}, label style={font=\scriptsize},
  xlabel={time (s)}, ylabel={},
  xmin=0.01, xmax=10, ymin=0.01, ymax=10
]
  \addplot[only marks, mark=asterisk, mark size=3pt, thick, draw=r2tCol] coordinates {(3.718867778778076, 0.2040841195029093)}; 
  \addplot[mark=triangle, mark size=3pt, mark options={fill=white, line width=1pt}, thick, draw=spCol] coordinates {
(0.056593616803487144,8.697080492124915) (0.08085731665293376,5.828798239924312) (0.2132432858149210,3.0103331798414668) ( 0.6116565863291422,1.5198792103815548) ( 1.0598469575246174, 0.8410447151500393)
}; 
  \addplot[mark=o, mark size=3pt, thick, draw=dpsCol] coordinates {
(0.04028503100077311,4.40729864386075) (0.06289347012837727, 2.523270954029436)
(0.18249154090881348,1.7696622984459491)  (0.5653611818949381,1.1268825582224025) ( 1.0265581210454304, 0.34111813832713733) 
  }; 
\end{axis}
\end{tikzpicture}
} & \hspace{-6mm}
\subcaptionbox{$\boldsymbol{Q}_{\square}$ (base query time: 7.47s).}[ \W ]{%
\begin{tikzpicture}
\begin{axis}[
  width=\W, height=\Hgt,
  xmode=log, ymode=log,
  log ticks with fixed point, scaled ticks=false,
  minor tick style={draw=none},
  grid=major, major tick length=2pt,
  minor grid style={draw=none}, 
  xtick={0.01,0.1,1,10,100,1000}, xticklabels={0.01,0.1,1,10,100,1000},
  ytick={0.1,1,10,100,1000}, yticklabels={0.1,1,10,100,1000},
  tick align=outside, tick style={black},
  tick label style={font=\scriptsize}, label style={font=\scriptsize},
  xlabel={time (s)}, ylabel={},
  xmin=0.01, xmax=10, ymin=0.01, ymax=10
]
  \addplot[only marks, mark=asterisk, mark size=3pt, thick, draw=r2tCol] coordinates {(7.110532482465108,  0.1371054685784109)}; 
  \addplot[mark=triangle, mark size=3pt, mark options={fill=white, line width=1pt}, thick, draw=spCol] coordinates {
( 0.13623833656311035,8.54766004190782)(0.3481195370356242,4.339479797601884)(0.7865856885910034, 1.9204761262918142)(1.5474845170974731,0.9942578910136742) ( 3.103283246358236,0.46504389023413184)
}; 
  \addplot[mark=o, mark size=3pt, thick, draw=dpsCol] coordinates {
( 0.08472235997517903,3.3835174730620357) (0.1328669786453247,3.430417267712052)
(0.6400534709294637, 2.0729479954649457) (1.295336922009786,0.9093297677280461) (2.265034794807434, 0.3788641115296967)
}; 
\end{axis}
\end{tikzpicture}
} 
\end{tabular}
}}
\vspace{-3mm}
\caption{Query error vs.\ time on RoadnetCA with $\varepsilon = 1$. (Sampling rates $q\in\left\{ \frac{1}{64},\tfrac{1}{32}, \tfrac{1}{16}, \tfrac{1}{8}, \tfrac{1}{4} \right\}$)}
\label{fig:roadnetca-q}
\end{figure}

%% file: ref.bib
@book{NesterovNemirovskii1994,
  author    = {Yurii Nesterov and Arkadii Nemirovskii},
  title     = {Interior-Point Polynomial Algorithms in Convex Programming},
  publisher = {SIAM},
  address   = {Philadelphia},
  year      = {1994},
  series    = {SIAM Studies in Applied Mathematics},
  doi       = {10.1137/1.9781611970791}
}

@article{PardalosVavasis1991,
  author  = {Panos M. Pardalos and Stephen A. Vavasis},
  title   = {Quadratic Programming with One Negative Eigenvalue is NP-hard},
  journal = {Journal of Global Optimization},
  year    = {1991},
  volume  = {1},
  number  = {1},
  pages   = {15--22},
  doi     = {10.1007/BF00120662}
}

@article{Ugander2011Facebook,
  author       = {Johan Ugander and
                  Brian Karrer and
                  Lars Backstrom and
                  Cameron Marlow},
  title        = {The Anatomy of the Facebook Social Graph},
  journal      = {arXiv preprint arXiv:1111.4503},
  year         = {2011},
  url          = {https://arxiv.org/abs/1111.4503}
}

@inproceedings{DBLP:conf/nips/BalleBG18,
  author       = {Borja Balle and
                  Gilles Barthe and
                  Marco Gaboardi},
  title        = {Privacy Amplification by Subsampling: Tight Analyses via Couplings
                  and Divergences},
  booktitle    = {NeurIPS},
  pages        = {6280--6290},
  year         = {2018}
}

@article{DBLP:journals/fttcs/DworkR14,
  author       = {Cynthia Dwork and
                  Aaron Roth},
  title        = {The Algorithmic Foundations of Differential Privacy},
  journal      = {Found. Trends Theor. Comput. Sci.},
  volume       = {9},
  number       = {3-4},
  pages        = {211--407},
  year         = {2014}
}

@inproceedings{DBLP:conf/sigmod/DongF0TM22,
  author       = {Wei Dong and
                  Juanru Fang and
                  Ke Yi and
                  Yuchao Tao and
                  Ashwin Machanavajjhala},
  title        = {{R2T:} Instance-optimal Truncation for Differentially Private Query
                  Evaluation with Foreign Keys},
  booktitle    = {{SIGMOD} Conference},
  pages        = {759--772},
  publisher    = {{ACM}},
  year         = {2022}
}

@article{DBLP:journals/pacmmod/FangY24,
  author       = {Juanru Fang and
                  Ke Yi},
  title        = {Privacy Amplification by Sampling under User-level Differential Privacy},
  journal      = {Proc. {ACM} Manag. Data},
  volume       = {2},
  number       = {1},
  pages        = {34:1--34:26},
  year         = {2024}
}

@article{DBLP:journals/pacmmod/0007S023,
  author       = {Wei Dong and
                  Dajun Sun and
                  Ke Yi},
  title        = {Better than Composition: How to Answer Multiple Relational Queries
                  under Differential Privacy},
  journal      = {Proc. {ACM} Manag. Data},
  volume       = {1},
  number       = {2},
  pages        = {123:1--123:26},
  year         = {2023}
}

@article{DBLP:journals/tods/KarwaRSY14,
  author       = {Vishesh Karwa and
                  Sofya Raskhodnikova and
                  Adam D. Smith and
                  Grigory Yaroslavtsev},
  title        = {Private Analysis of Graph Structure},
  journal      = {{ACM} Trans. Database Syst.},
  volume       = {39},
  number       = {3},
  pages        = {22:1--22:33},
  year         = {2014}
}

@incollection{DBLP:books/sp/17/Vadhan17,
  author       = {Salil P. Vadhan},
  title        = {The Complexity of Differential Privacy},
  booktitle    = {Tutorials on the Foundations of Cryptography},
  pages        = {347--450},
  publisher    = {Springer International Publishing},
  year         = {2017}
}

@article{DBLP:journals/isci/GilAL13,
  author       = {Manuel Gil and
                  Fady Alajaji and
                  Tam{\'{a}}s Linder},
  title        = {R{\'{e}}nyi divergence measures for commonly used univariate
                  continuous distributions},
  journal      = {Inf. Sci.},
  volume       = {249},
  pages        = {124--131},
  year         = {2013}
}

@inproceedings{DBLP:conf/stoc/NissimRS07,
  author       = {Kobbi Nissim and
                  Sofya Raskhodnikova and
                  Adam D. Smith},
  title        = {Smooth sensitivity and sampling in private data analysis},
  booktitle    = {{STOC}},
  pages        = {75--84},
  publisher    = {{ACM}},
  year         = {2007}
}

@inproceedings{DBLP:conf/csfw/Mironov17,
  author       = {Ilya Mironov},
  title        = {R{\'{e}}nyi Differential Privacy},
  booktitle    = {{CSF}},
  pages        = {263--275},
  publisher    = {{IEEE} Computer Society},
  year         = {2017}
}

@inproceedings{DBLP:conf/focs/DworkRV10,
  author       = {Cynthia Dwork and
                  Guy N. Rothblum and
                  Salil P. Vadhan},
  title        = {Boosting and Differential Privacy},
  booktitle    = {{FOCS}},
  pages        = {51--60},
  publisher    = {{IEEE} Computer Society},
  year         = {2010}
}

@inproceedings{DBLP:conf/nips/LevySAKKMS21,
  author       = {Daniel Levy and
                  Ziteng Sun and
                  Kareem Amin and
                  Satyen Kale and
                  Alex Kulesza and
                  Mehryar Mohri and
                  Ananda Theertha Suresh},
  title        = {Learning with User-Level Privacy},
  booktitle    = {NeurIPS},
  pages        = {12466--12479},
  year         = {2021}
}

@inproceedings{DBLP:conf/nips/GhaziK0MMZ23,
  author       = {Badih Ghazi and
                  Pritish Kamath and
                  Ravi Kumar and
                  Pasin Manurangsi and
                  Raghu Meka and
                  Chiyuan Zhang},
  title        = {User-Level Differential Privacy With Few Examples Per User},
  booktitle    = {NeurIPS},
  year         = {2023},
    volume={36},
  pages={19263--19290},
}

@inproceedings{DBLP:conf/tcc/KasiviswanathanNRS13,
  author       = {Shiva Prasad Kasiviswanathan and
                  Kobbi Nissim and
                  Sofya Raskhodnikova and
                  Adam D. Smith},
  title        = {Analyzing Graphs with Node Differential Privacy},
  booktitle    = {{TCC}},
  series       = {Lecture Notes in Computer Science},
  volume       = {7785},
  pages        = {457--476},
  publisher    = {Springer},
  year         = {2013}
}

@article{DBLP:journals/pvldb/JiangLYX24,
  author       = {Yangfan Jiang and
                  Xinjian Luo and
                  Yin Yang and
                  Xiaokui Xiao},
  title        = {Calibrating Noise for Group Privacy in Subsampled Mechanisms},
  journal      = {Proc. {VLDB} Endow.},
  volume       = {18},
  number       = {2},
  pages        = {322--334},
  year         = {2024}
}

@article{DBLP:journals/tit/ErvenH14,
  author       = {Tim van Erven and
                  Peter Harremo{\"{e}}s},
  title        = {R{\'{e}}nyi Divergence and Kullback-Leibler Divergence},
  journal      = {{IEEE} Trans. Inf. Theory},
  volume       = {60},
  number       = {7},
  pages        = {3797--3820},
  year         = {2014}
}

@misc{snapnets,
  author       = {Jure Leskovec and Andrej Krevl},
  title        = {{SNAP Datasets}: {Stanford} Large Network Dataset Collection},
  howpublished = {\url{http://snap.stanford.edu/data}},
  month        = jun,
  year         = 2014
}

@article{DBLP:journals/pvldb/KotsogiannisTHF19,
  author       = {Ios Kotsogiannis and
                  Yuchao Tao and
                  Xi He and
                  Maryam Fanaeepour and
                  Ashwin Machanavajjhala and
                  Michael Hay and
                  Gerome Miklau},
  title        = {PrivateSQL: {A} Differentially Private {SQL} Query Engine},
  journal      = {Proc. {VLDB} Endow.},
  volume       = {12},
  number       = {11},
  pages        = {1371--1384},
  year         = {2019}
}

@book{brent2013algorithms,
  title={Algorithms for minimization without derivatives},
  author={Brent, Richard P},
  year={2013},
  publisher={Courier Corporation}
}

@manual{mosek,
   author = "MOSEK ApS",
   title = "The MOSEK Python Fusion API manual. Version 11.0.",
   year = 2025,
   url = "https://docs.mosek.com/latest/pythonfusion/index.html"
 }

@inproceedings{DBLP:conf/dbsec/CaoCPNY23,
  author       = {Xuyang Cao and
                  Yang Cao and
                  Primal Pappachan and
                  Atsuyoshi Nakamura and
                  Masatoshi Yoshikawa},
  title        = {Differentially Private Streaming Data Release Under Temporal Correlations
                  via Post-processing},
  booktitle    = {DBSec},
  series       = {Lecture Notes in Computer Science},
  volume       = {13942},
  pages        = {184--200},
  publisher    = {Springer},
  year         = {2023}
}

@article{DBLP:journals/pacmmod/0007CLS024,
  author       = {Wei Dong and
                  Zijun Chen and
                  Qiyao Luo and
                  Elaine Shi and
                  Ke Yi},
  title        = {Continual Observation of Joins under Differential Privacy},
  journal      = {Proc. {ACM} Manag. Data},
  volume       = {2},
  number       = {3},
  pages        = {128},
  year         = {2024}
}

@inproceedings{DBLP:conf/stoc/DworkNPR10,
  author       = {Cynthia Dwork and
                  Moni Naor and
                  Toniann Pitassi and
                  Guy N. Rothblum},
  title        = {Differential privacy under continual observation},
  booktitle    = {{STOC}},
  pages        = {715--724},
  publisher    = {{ACM}},
  year         = {2010}
}

@inproceedings{DBLP:conf/pet/ChanLSX12,
  author       = {T.{-}H. Hubert Chan and
                  Mingfei Li and
                  Elaine Shi and
                  Wenchang Xu},
  title        = {Differentially Private Continual Monitoring of Heavy Hitters from
                  Distributed Streams},
  booktitle    = {Privacy Enhancing Technologies},
  series       = {Lecture Notes in Computer Science},
  volume       = {7384},
  pages        = {140--159},
  publisher    = {Springer},
  year         = {2012}
}

@article{DBLP:journals/vldb/TaoGMR25,
  author       = {Yuchao Tao and
                  Amir Gilad and
                  Ashwin Machanavajjhala and
                  Sudeepa Roy},
  title        = {Differentially private explanations for aggregate query answers},
  journal      = {{VLDB} J.},
  volume       = {34},
  number       = {2},
  pages        = {20},
  year         = {2025}
}

@inproceedings{DBLP:conf/sigmod/0002HIM19,
  author       = {Chang Ge and
                  Xi He and
                  Ihab F. Ilyas and
                  Ashwin Machanavajjhala},
  title        = {APEx: Accuracy-Aware Differentially Private Data Exploration},
  booktitle    = {{SIGMOD} Conference},
  pages        = {177--194},
  publisher    = {{ACM}},
  year         = {2019}
}

@inproceedings{DBLP:conf/icalp/Dwork06,
  author       = {Cynthia Dwork},
  title        = {Differential Privacy},
  booktitle    = {{ICALP} {(2)}},
  series       = {Lecture Notes in Computer Science},
  volume       = {4052},
  pages        = {1--12},
  publisher    = {Springer},
  year         = {2006}
}

@article{DBLP:journals/pacmmod/ZhangH23,
  author       = {Shufan Zhang and
                  Xi He},
  title        = {DProvDB: Differentially Private Query Processing with Multi-Analyst
                  Provenance},
  journal      = {Proc. {ACM} Manag. Data},
  volume       = {1},
  number       = {4},
  pages        = {267:1--267:27},
  year         = {2023}
}

@article{DBLP:journals/pacmmod/GuanYZXCXXZ25,
  author       = {Hong Guan and
                  Lei Yu and
                  Lixi Zhou and
                  Li Xiong and
                  Kanchan Chowdhury and
                  Lulu Xie and
                  Xusheng Xiao and
                  Jia Zou},
  title        = {Privacy and Accuracy-Aware {AI/ML} Model Deduplication},
  journal      = {Proc. {ACM} Manag. Data},
  volume       = {3},
  number       = {3},
  pages        = {203:1--203:28},
  year         = {2025}
}

@inproceedings{DBLP:conf/tcc/BunS16,
  author       = {Mark Bun and
                  Thomas Steinke},
  title        = {Concentrated Differential Privacy: Simplifications, Extensions, and
                  Lower Bounds},
  booktitle    = {{TCC} {(B1)}},
  series       = {Lecture Notes in Computer Science},
  volume       = {9985},
  pages        = {635--658},
  year         = {2016}
}

@inproceedings{DBLP:conf/focs/AtseriasGM08,
  author       = {Albert Atserias and
                  Martin Grohe and
                  D{\'{a}}niel Marx},
  title        = {Size Bounds and Query Plans for Relational Joins},
  booktitle    = {{FOCS}},
  pages        = {739--748},
  publisher    = {{IEEE} Computer Society},
  year         = {2008}
}

@online{full-ver,
  author = {Yuan Qiu and
            Xiaokui Xiao and
            Yin Yang},
  title = {DP-S4S: Accurate and Scalable Select-Join-Aggregate Query Processing with User-Level Differential Privacy},
  year = 2026,
  url = {https://arxiv.org/abs/2603.14994}
}
